\documentclass[reqno]{amsart}
\usepackage{fullpage,amsfonts,amssymb,amsthm,mathrsfs,graphicx,color,framed,esint,stackrel}
\usepackage{cancel,bm}
\usepackage{tikz}
\usetikzlibrary{decorations.markings}
\usepackage{booktabs,longtable}
\usepackage[initials]{amsrefs}
\usepackage{mdframed}

\DeclareMathAlphabet\mathbfcal{OMS}{cmsy}{b}{n}

\usepackage[colorlinks=true, pdfstartview=FitV, linkcolor=blue, citecolor=blue, urlcolor=blue]{hyperref}

\renewcommand{\d}{{\mathrm d}}

\newcommand{\im}{\mathrm{i}}
\newcommand{\e}{\mathrm{e}}

\def\tr{\mathop{\mathrm{tr}}\limits}

\newtheorem{theo}{Theorem}[section]

\newtheorem{lem}[theo]{Lemma}
\newtheorem{rem}[theo]{Remark}
\newtheorem{problem}[theo]{Riemann-Hilbert Problem}
\newtheorem{remark}[theo]{Remark}
\newtheorem{prop}[theo]{Proposition} 
\newtheorem{cor}[theo]{Corollary} 
\newtheorem{definition}[theo]{Definition}

\begin{document}
\title[Soft edge momenta in the anharmonic oscillator]{Momenta spacing distributions in anharmonic oscillators and the higher order finite temperature Airy kernel}

\author{Thomas Bothner}
\address{School of Mathematics, University of Bristol, Fry Building, Woodland Road, Bristol, BS8 1UG, United Kingdom}
\email{thomas.bothner@bristol.ac.uk}

\author{Mattia Cafasso}
\address{LAREMA, Universit\'e d'Angers, 2bd Lavoisier, 49045 Angers, France}
\email{cafasso@math.univ-angers.fr}

\author{Sofia Tarricone}
\address{LAREMA, Universit\'e d'Angers, 2bd Lavoisier, 49045 Angers, France}
\address{Department of Mathematics and Statistics, Concordia University, 1455 de Maisonneuve W., Montr\'eal, Qu\'ebec, Canada, H3G 1M8}
\email{sofia.tarricone@concordia.ca}

\date{\today}


\thanks{The work of T.B. is supported by the Engineering and Physical Sciences Research Council through grant EP/T013893/2. M.C. and S.T. are supported by the European Union Horizon 2020
research and innovation program under the Marie Sk\l odowska-Curie RISE 2017 grant
778010 IPaDEGAN}


\begin{abstract}
We rigorously compute the integrable system for the limiting $(N\rightarrow\infty)$ distribution function of the extreme momentum of $N$ noninteracting fermions when confined to an anharmonic trap $V(q)=q^{2n}$ for $n\in\mathbb{Z}_{\geq 1}$ at positive temperature. More precisely, the edge momentum statistics in the harmonic trap $n=1$ are known to obey the weak asymmetric KPZ crossover law which is realized via the finite temperature Airy kernel determinant or equivalently via a Painlev\'e-II integro-differential transcendent, cf. \cite{LW,ACQ}. For general $n\geq 2$, a novel higher order finite temperature Airy kernel has recently emerged in physics literature \cite{DMS} and we show that the corresponding edge law in momentum space is now governed by a distinguished Painlev\'e-II integro-differential hierarchy. Our analysis is based on operator-valued Riemann-Hilbert techniques which produce a Lax pair for an operator-valued Painlev\'e-II ODE system that naturally encodes the aforementioned hierarchy. As byproduct, we establish a connection of the integro-differential Painlev\'e-II hierarchy to a novel integro-differential mKdV hierarchy.
\end{abstract}


\maketitle

\section{Introduction and statement of results}\label{zsec0}
In this paper we present new results for the edge momentum distribution function of a noninteracting fermionic quantum many body system confined to an anharmonic trap at finite temperature. It is known from \cite[(5),(20)]{DMS} that the same extreme value distribution 
is expressible in terms of a Fredholm determinant of an integral operator whose kernel involves the higher order Airy function (equivalently, see \cite{Ko}, extended Airy function of the first kind)
\begin{equation}\label{i1}
	\textnormal{Ai}_n(x):=\frac{1}{\pi}\int_0^{\infty}\cos\left(\frac{y^{2n+1}}{2n+1}+xy\right)\d y,\ \ \ \ \ x\in\mathbb{R},\ \ \ n\in\mathbb{N}:=\mathbb{Z}_{\geq 1}.
\end{equation}
In the simplest case of a harmonic trap, the higher order Airy function \eqref{i1} becomes an ordinary Airy function and the extreme momentum statistics are known to be described by the weak asymmetric KPZ crossover law \cite[Proposition $1.2$]{ACQ}. It turns out there is a striking generalization of the  crossover law to general anharmonic traps with even monomial potential: using Fourier analytic and operator-valued Riemann-Hilbert techniques we will express the higher order finite temperature Airy kernel determinants in terms of a distinguished solution of an integro-differential Painlev\'e-II hierarchy.

\subsection{Fermionic coordinates in monomial anharmonic traps}\label{fermi1}
In order to be more explicit, consider the one-dimensional Schr\"odinger operator 
\begin{equation}\label{i2}
	H_q:=-\frac{\d^2}{\d q^2}+q^{2n},\ \ \ n\in\mathbb{N},
\end{equation}
in coordinate representation on $L^2(\mathbb{R})$ with monomial potential $L_{\textnormal{loc}}^1(\mathbb{R})\ni V(q):=q^{2n}$. By classical theory, see \cite[Chapter $2.3$]{BS} or \cite[Chapter $3$]{T}, the Hamiltonian \eqref{i2} considered on the domain $C_0^{\infty}(\mathbb{R})$ of smooth functions on $\mathbb{R}$ with compact support is essentially self-adjoint and its closure (again denoted by $H_q$) has pure point spectrum. In turn, there exists a complete orthonormal system $\{\psi_k\}_{k\in\mathbb{N}}$ for $L^2(\mathbb{R})$ consisting of eigenfunctions of $H_q$,
\begin{equation}\label{i3}
	H_q\psi_k=\lambda_k\psi_k,\ \ \ k\in\mathbb{N},
\end{equation}
with eigenvalues $\lambda_k\in\mathbb{R}$ that tend to $+\infty$ as $k\rightarrow\infty$. Moving ahead, in modeling the desired noninteracting fermionic qantum gas one recognizes the symmetrization postulate of quantum mechanics, cf. \cite[Chapter $4$, $\S$3]{T}, i.e. the locations $\{q_k\}_{k=1}^N\subset\mathbb{R}$ of $N$ identical, noninteracting fermions at zero temperature confined to the trap $V$ are distributed according to a biorthogonal point ensemble indexed by ${\bf k}:=(k_1,\ldots,k_N)\in\mathbb{N}^N$, that is to say the locations $\{q_k\}_{k=1}^N$ form a special instance of a determinantal point process with joint probability density function (pdf) 
\begin{equation}\label{i4}
	f_{0,{\bf k}}(q_1,\ldots,q_N):=\frac{1}{N!}\det\big[\psi_{k_i}^{\ast}(q_j)\big]_{i,j=1}^N\det\big[\psi_{k_i}(q_j)\big]_{i,j=1}^N,\ \ \ \ 1\leq k_1<k_2<\ldots<k_N,\ \ k_i\in\mathbb{N}.
\end{equation}
Here we use the single-coordinate wave functions $\psi_{k_i}$ given in \eqref{i3}. Note that
\begin{equation*}
	\Psi_{\bf k}(q_1,\ldots,q_N):=\frac{1}{\sqrt{N!}}\det[\psi_{k_i}(q_j)]_{i,j=1}^N,
\end{equation*}
is the standard antisymmetric $N$-coordinate wave function, an eigenfunction of the Hamiltonian 
\begin{equation*}
	H_N:=\sum_{i=1}^NH_{q_i}=\sum_{i=1}^N\left(-\frac{\d^2}{\d q_i^2}+q_i^{2n}\right),
\end{equation*}
which physically describes an eigenstate of the fermionic gas with energy  $\lambda_{{\bf k}}:=\sum_{i=1}^N\lambda_{k_i}$, see for instance \cite[Section IV]{DDMS}. 
At finite temperature $T>0$ all such eigenstates occur according to the Boltzmann-Gibbs distribution and the coordinate pdf \eqref{i4} gets generalized to
\begin{equation}\label{i5}
	f_T(q_1,\ldots,q_N):=\frac{1}{Z_N(\beta)}\sum_{\substack{{\bf k}\in\mathbb{N}^N\\ k_1<\ldots<k_N}}\big|\Psi_{\bf k}(q_1,\ldots,q_N)\big|^2\e^{-\beta\lambda_{\bf k}},\ \ \ \beta:=\frac{1}{T}>0,
\end{equation}
with the canonical partition function $Z_N(\beta):=\sum_{{\bf k}\in\mathbb{N}^N:\,k_1<\ldots<k_N}\e^{-\beta\lambda_{\bf k}}$, see \cite[$(64)$]{DDMS}\smallskip

Starting from \eqref{i5} one can now analyze various fine structure properties of the coordinate point process, in particular the large $N$ scaling behavior of its extreme value $q_{\max}(N):=\max_{1\leq i\leq N}q_i$ has been at the center of interest in theoretical physics in recent years, partially because experimental advances on cold atom imaging have made it possible to probe the positions of individual gas particles and one therefore requires a precise spatial description of the gas itself, see \cite{DDMS} for background. Mathematically, this task asks for the derivation of large $N$ limit laws and a first rigorous answer was given in \cite[$(25)$]{LW}, albeit for the harmonic trap,
\begin{equation}\label{i6}
	\lim_{N\rightarrow\infty}\mathbb{P}\left(\frac{q_{\max}(N)-\sqrt{2N}}{2^{-\frac{1}{2}}N^{-\frac{1}{6}}}\leq t\,\,\bigg|\,\,n=1,\,T=N^{\frac{1}{3}}\alpha^{-1},\ \alpha>0\ \textnormal{fixed}\right)=:F_{1}^{\alpha}(t),
\end{equation}
pointwise in $t\in\mathbb{R}$, where $F_1^{\alpha}(t)=\det(I-K_{t,1}^{\alpha}\upharpoonright_{L^2(\mathbb{R}_+)})$ equals the Fredholm determinant of the finite temperature Airy kernel (recall $\textnormal{Ai}_1\equiv\textnormal{Ai}$ as in \eqref{i1} - the reader should not confuse our $F_1$ notation as an abbreviation for the Tracy-Widom distribution in the Gaussian orthogonal ensemble)
\begin{equation}\label{i7}
	K_{t,1}^{\alpha}(x,y):=\int_{\mathbb{R}}\textnormal{Ai}_1(x+z+t)\textnormal{Ai}_1(z+y+t)\left(\frac{\e^{\alpha z}}{1+\e^{\alpha z}}\right)\d z.
\end{equation}
The same operator determinant had occurred prior to \cite{DDMS} and \cite{LW}, first in Johansson's work \cite[Theorem $1.3$]{J} on grand canonical scaling limits in the Moshe-Neuberger-Shapiro model and then in the paper \cite[Theorem $1.1$]{ACQ} by Amir-Corwin-Quastel on the probability distribution of the KPZ solution with narrow wedge initial condition. Following \cite{ACQ}, the distribution function $F_1^{\alpha}(t)$ interpolates with varying $\alpha\in(0,+\infty)$ between two universality classes (Tracy-Widom and Gumbel) and it was therefore coined a crossover distribution which proved to underlie the weak asymmetric limit of models in the KPZ universality class, cf. \cite{C,BCF,D}. Furthermore, and related to our analysis, \eqref{i6} can be expressed in terms of an integro-differential Painlev\'e-II transcendent 
\begin{equation}\label{i8}
	F_1^{\alpha}(t)=\exp\left[-\int_t^{\infty}(s-t)\left(\int_{\mathbb{R}}u^2(s|x)w_{\alpha}'(x)\,\d x\right)\d s\right],\ \ \ w_{\alpha}(x):=\frac{\e^{\alpha x}}{1+\e^{\alpha x}},\ \ w_{\alpha}':=\frac{\d w_{\alpha}}{\d x},
\end{equation}
where $u=u(s|t)$ is the unique real-valued, smooth in $s\in\mathbb{R}$ for any $t\in\mathbb{R}$, solution of the boundary value problem
\begin{equation}\label{i9}
	\frac{\d^2}{\d s^2}u(s|t)=\left[s+t+2\int_{\mathbb{R}}u^2(s|x)w_{\alpha}'(x)\,\d x\right]u(s|t),\ \ \ \ u(s|t)\stackrel{s\rightarrow\infty}{\sim}\textnormal{Ai}_1(s+t)\ \ \textnormal{pointwise in}\ t\in\mathbb{R}.
\end{equation}
Formula \eqref{i8} generalizes the Tracy-Widom formula \cite[$(1.17)$]{TW0} in the Gaussian unitary ensemble and \eqref{i9} the Hastings-McLeod Painlev\'e-II transcendent involved in it. Although \eqref{i6} has only been proven rigorously in the harmonic case, the use of local density approximations and functional methods in \cite[Section VII]{DDMS} has put forward convincing evidence that the limit law \eqref{i6} for $q_{\max}(N)$ holds true for all $V(q)=q^{2n}$, after appropriate $n$-dependent centering and scaling.

\subsection{Fermionic momenta in monomial anharmonic traps}\label{fermi2} Somewhat surprisingly, the above coordinate universality phenomenon does not appear to carry over to the momentum representation
\begin{equation*}
	H_p=p^2+(-1)^n\frac{\d^n}{\d p^n}
\end{equation*}
of the Hamiltonian \eqref{i2}. In particular, while the average \textit{coordinate} density 
\begin{equation*}
	\rho_N(q):=\int_{\mathbb{R}}\cdots\int_{\mathbb{R}}f_T(q,q_2,\ldots,q_N)\,\d q_2\cdots\d q_N
\end{equation*}
is expected to vanish square root like near $q_{\max}(N)$ in all anharmonic, even monomial traps, see \cite[Section VII, D]{DDMS}, the local behavior of the average \textit{momentum} density $\bar{\rho}_N(p)$ near $p_{\max}(N):=\max_{1\leq i\leq N}p_i$ is more sophisticated and highly $n$-dependent. Indeed, using Wigner's quasi pdf in \cite[(3)]{DDMS2}, the recent paper \cite{DMS} argued that there exist $n$-dependent factors $a_n,b_n>0$ such that for large $N$, with high probability,
\begin{equation*}
	p_{\max}(N)\sim a_nN^{\frac{n}{n+1}}=:p_{\textnormal{edge}}(N),
\end{equation*}
and consequently, also for large $N$,
\begin{equation}\label{i10}
	\bar{\rho}_N(p)\sim b_nN^{\frac{1}{2(n+1)}}\big(p_{\textnormal{edge}}(N)-p\big)^{\frac{1}{2n}},\ \ 0<p<p_{\textnormal{edge}}(N)\ \textnormal{close to the deterministic value}\ p_{\textnormal{edge}}(N),
\end{equation}	
see the supplemental material to \cite{DMS}, especially equations $(2)$ and $(3)$ therein, all for the anharmonic trap $V(q)=q^{2n}$. The $n$-dependence in \eqref{i10} hints at a novel edge momentum phenomenon in the fermionic gas, somewhat reminiscent of the non-generic higher order soft edge behavior in certain Hermitian random matrix models, see \cite{CV,CIK}. Although the techniques in \cite{DMS} are in general non-rigorous (rigorous only for $n=1$ when the coordinate and momentum representation of \eqref{i2} are in perfect Fourier duality), they have motivated an analogue  of \eqref{i6} for $p_{\max}(N)$ for arbitrary $n\in\mathbb{N}$. More precisely, see \cite[$(23),(24)$]{DMS}, it is expected that for some $n$-dependent factors $c_n,d_n>0$, pointwise in $t\in\mathbb{R}$,
\begin{equation}\label{i11}
	\lim_{N\rightarrow\infty}\mathbb{P}\left(\frac{p_{\max}(N)-p_{\textnormal{edge}}(N)}{c_nN^{-e_n}}\leq t\,\,\bigg|\,\,T=d_nN^{f_n}\alpha^{-1},\ \alpha>0\ \textnormal{fixed}\right)=:F_{n}^{\alpha}(t),
\end{equation}
where $e_n,f_n$ are the scaling exponents
\begin{equation*}
	e_n:=\frac{n}{(n+1)(2n+1)}\ \ \ \ \ \ \ \ \textnormal{and}\ \ \ \ \ \ \ f_n:=\frac{2n^2}{(n+1)(2n+1)}.
\end{equation*}
In \eqref{i11}, $F_n^{\alpha}(t)=\det(I-K_{t,n}^{\alpha}\upharpoonright_{L^2(\mathbb{R}_+)})$ is the Fredholm determinant of the higher order finite temperature Airy kernel defined as
\begin{equation}\label{i12}
	K_{t,n}^{\alpha}(x,y):=\int_{\mathbb{R}}\textnormal{Ai}_n(x+z+t)\textnormal{Ai}_n(z+y+t)\left(\frac{\e^{\alpha z}}{1+\e^{\alpha z}}\right)\d z,\ \ \ \ n\in\mathbb{N},
\end{equation}
in terms of \eqref{i1}. While \eqref{i11} is $n$-dependent, the same scaling limit is expected to be universal across the class of smooth confining potentials $V(q)$ with a single global minimum at $q_{\ast}$ such that $\lim_{|q|\rightarrow\infty}V(q)=+\infty$ and $V(q)-V(q_{\ast})\sim(q-q_{\ast})^{2n}$ near  $q_{\ast}$. In our first result below, see Theorem \ref{itheo1}, we will derive the analogues of \eqref{i8} and \eqref{i9} for the distribution function $F_n^{\alpha}(t)$. In fact, our analysis is valid for a larger class of kernels of Hankel composition operators than \eqref{i12} with the Fermi factor $w_{\alpha}(z)$. The details are as follows.
\subsection{An integro-differential Painlev\'e-II hierarchy} Abbreviate $\mathbb{R}_+:=(0,\infty)$ and consider an arbitrary positive, strictly increasing and differentiable weight function $w:\mathbb{R}\rightarrow\mathbb{R}_+$ such that for some $\omega,x_0>0$,
\begin{equation}\label{i13}
	\lim_{x\rightarrow+\infty}w(x)=1,\ \ \lim_{x\rightarrow-\infty}w(x)=0\ \ \ \ \ \ \ \ \textnormal{and}\ \ \ \ \ \ 0<w'(x)\leq \e^{-\omega|x|}\ \ \ \ \ \forall\,|x|\geq x_0.
\end{equation}
Define the integral operator $K_{t,n}:L^2(\mathbb{R}_+)\rightarrow L^2(\mathbb{R}_+)$ as follows,
\begin{equation}\label{i14}
	\big(K_{t,n}f\big)(x)=\int_{\mathbb{R}_+}K_{t,n}(x,y)f(y)\,\d y,\ \ \ \ K_{t,n}(x,y):=\int_{\mathbb{R}}\textnormal{Ai}_n(x+z+t)\textnormal{Ai}_n(z+y+t)w(z)\,\d z,
\end{equation}
and note that $K_{t,n}$ is trace class on $L^2(\mathbb{R}_+)$, see Corollary \ref{zcor:0} below. Hence, its Fredholm determinant
\begin{equation}\label{i15}
	D_n(t,\lambda):=\det(I-\lambda K_{t,n}\upharpoonright_{L^2(\mathbb{R}_+)}),\ \ \ \ \ \ (t,\lambda,n)\in\mathbb{R}\times\mathbb{C}\times\mathbb{N},
\end{equation}
is well-defined with $t\mapsto D_n(t,\cdot)$ differentiable (by \cite[Lemma $2.20$]{ACQ} since $t\mapsto K_{t,n}$ is differentiable with trace class derivative $\frac{\d}{\d t}K_{t,n}$, compare the proof of Lemma \ref{zlem:3} below) and $\lambda\mapsto D_n(\cdot,\lambda)$ entire, see \cite[Lemma $3.3$]{S}. In order to state the generalization of \eqref{i8}, \eqref{i9} to the higher order finite temperature Airy kernel determinant $D_n(t,\lambda)$, we require the following operator abbreviations.
\begin{definition}\label{idef1} Given a function $\mathbb{R}^2\ni (t,x)\mapsto f(t|x)$, we let $D_t^{\pm 1}f$ denote its fractional $t$-derivatives such that $(D_t^{-1}D_tf)(t|x)=f(t|x)$, i.e. $D_t$ is the ordinary $t$-derivative and $D_t^{-1}$ the $t$-antiderivative. Now define, for given $u=u(t|x)$,
\begin{align*}
	(\mathcal{L}_+^uf)(t|x):=&\,\im(D_tf)(t|x)-\im\big\langle (D_t^{-1}\{u,f\})(t|x,\cdot),u\big\rangle-2\im\big(D_t^{-1}\langle u,f\rangle\big)u(t|x),\\
	(\mathcal{L}_-^uf)(t|x):=&\,\im(D_tf)(t|x)+\im\big\langle (D_t^{-1}[u,f])(t|x,\cdot),u\big\rangle,
\end{align*}
where the rank two integral operators $[\alpha,\beta]:=\alpha\otimes\beta-\beta\otimes\alpha$ and $\{\alpha,\beta\}:=\alpha\otimes\beta+\beta\otimes\alpha$ have kernels
\begin{equation*}
	[\alpha,\beta](t|x,y)=\alpha(t|x)\beta(t|y)-\beta(t|x)\alpha(t|y),\ \ \  \{\alpha,\beta\}(t|x,y)=\alpha(t|x)\beta(t|y)+\beta(t|x)\alpha(t|y),
\end{equation*} 
and $\langle\cdot,\cdot\rangle$ denotes the weighted bilinear form
\begin{equation*}
	\langle f,g\rangle:=\int_{\mathbb{R}}f(t|x)g(t|x)w'(x)\,\d x,\ \ \ \ w'(x)=\frac{\d w}{\d x}(x).
\end{equation*}
\end{definition}
The operators $\mathcal{L}_{\pm}^u$ in Definition \ref{idef1} allow us to state our main result in the following compact fashion.
\begin{theo}\label{itheo1} For every $(t,\lambda,n)\in\mathbb{R}\times\overline{\mathbb{D}_1(0)}\times\mathbb{N}$, with the closed unit disk\, $\overline{\mathbb{D}_1(0)}:=\{\lambda\in\mathbb{C}:\,|\lambda|\leq 1\}$,
\begin{equation}\label{i16}
	D_n(t,\lambda)=\exp\left[-\int_t^{\infty}(s-t)\left(\int_{\mathbb{R}}u^2(s|x)w'(x)\,\d x\right)\d s\right],
\end{equation}
where $u(t|x)\equiv u(t|x;n,\lambda)$ is the unique solution of the boundary value problem
\begin{equation}\label{i17}
	-(t+x)u(t|x)=\big((\mathcal{L}_+^u\mathcal{L}_-^u)^nu\big)(t|x),\ \ \ \ \ \ \ u(t|x)\sim\lambda^{\frac{1}{2}}\textnormal{Ai}_n(t+x),\ \ \ t\rightarrow+\infty.
\end{equation}
The mapping $t\mapsto u(t|x;n,\lambda)$ is smooth for any $(x,\lambda,n)\in\mathbb{R}\times\overline{\mathbb{D}_1(0)}\times\mathbb{N}$, the asymptotic expansion in \eqref{i17} holds pointwise in $x\in\mathbb{R}$ and we choose an arbitrary, albeit fixed, branch for $\lambda^{\frac{1}{2}}$.
\end{theo}
\begin{rem}\label{irem1} Although $\mathcal{L}_{\pm}^u$ seemingly generate $t$-antiderivatives, the operator $D_t^{-1}$ as appearing in \eqref{i17} acts always on a total $t$-derivative. In turn, \eqref{i17} is completely $t$-localized, see Lemma \ref{z:lem11} and Corollary \ref{z:cor4} below.
\end{rem}
\begin{rem}\label{irem2} The solution $u=u(t|x;n,\lambda)$ of \eqref{i17} depends on the weight function $w$ in a non-trivial fashion, compare Definition \ref{idef1}. Yet we do not choose to explicitly indicate the $w$-dependence of $u$ in our notation.
\end{rem}
\begin{rem} Assumption \eqref{i13} on the exponential decay of $w'$ is naturally enforced by our proof method. Conjecturally, compare \cite[Section $9$]{B}, a power like vanishing of $w'$ at $\pm\infty$ is sufficient for the validity of \eqref{i16} and \eqref{i17}. We further comment on this feature in Remark \ref{momrem} and Subsection \ref{method} below.
\end{rem}
Before moving on, we explicitly list a few members of the integro-differential Painlev\'e-II hierarchy defined through the dynamical system \eqref{i17}. Indeed, using the shorthand
\begin{equation*}
	u=u(t|x),\ \ \ \ u'=(D_tu)(t|x),\ \ \ \ \ u''=(D_t^2u)(t|x),\ \ \ \ \ u'''=(D_t^3u)(t|x),\ \ \ \ \ldots
\end{equation*}
the first three members read as
\begin{align}
	n=1:\ \ \ \ \ (t+x)u=&\,\,u''-2u\langle u,u\rangle,\label{i18}\\
	n=2:\ \ \ -(t+x)u=&\,\,u''''-4u''\langle u,u\rangle-8u'\langle u',u\rangle-6u\langle u,u''\rangle-2u\langle u',u'\rangle+6u\langle u,u\rangle^2,\label{i19}
\end{align}
and
\begin{align}
	n=3&:\ \ (t+x)u=u''''''-6u''''\langle u,u\rangle-8u\langle u'''',u\rangle-24u'''\langle u',u\rangle-19u'\langle u,u'''\rangle-13u\langle u''',u'\rangle\nonumber\\
	&\hspace{1cm}-31u''\langle u'',u\rangle-11u\langle u'',u''\rangle-25u''\langle u',u'\rangle-45u'\langle u'',u'\rangle+15u''\langle u,u\rangle^2\nonumber\\
	&\hspace{1cm}+55u\langle u,u\rangle\langle u'',u\rangle+60u'\langle u',u\rangle\langle u,u\rangle+25u\langle u',u'\rangle\langle u,u\rangle+55u\langle u',u\rangle^2-20u\langle u,u\rangle^3.\label{i20}
\end{align}
Clearly \eqref{i9} is a special case of \eqref{i18} and \eqref{i19} matches \cite[$(4.13)$]{Kra} once the sign difference between \cite[$(4.1)$]{Kra} and our convention for \eqref{i1}, see Lemma \ref{zlem:1}, has been observed\footnote{We follow \cite{DMS} and use the generalized Airy equation $\frac{\d^{2n}}{\d z^{2n}}w=(-1)^{n+1}zw$. This is not the case in \cite[$(4.1)$]{Kra}.}. The third member \eqref{i20} has not appeared in the literature to the best of our knowledge. It formally reproduces the third member \cite[$(3.6)$]{CM} of the ordinary Painlev\'e-II hierarchy when $w'(x)=\delta_0(x)$ is the delta point mass at $x=0$, modulo the obvious typo correction $42(w'')^2\mapsto 42w(w'')^2$ in \cite[$(3.6)$]{CM}. More generally, when $w'(x)=\delta_0(x)$, the first equality in \eqref{i17} implies, formally, the classical Painlev\'e-II hierarchy for the function $u(t|0)$, as written in \cite[$(4.11)$]{Airault}. To see this, just observe that $\mathcal L_-^{u}\mathcal L_+^{u}$, composed with the evaluation at $x = 0$, reduces to the recursion operator in \cite[$(4.6)$]{Airault}.\bigskip

Given that $F_n^{\alpha}(t)=D_n(t,1)$ in \eqref{i15} with the particular choice $w(x)=w_{\alpha}(x)$ from \eqref{i8}, Theorem \ref{itheo1} is the sought after generalization of \eqref{i8} and \eqref{i9}. We now move to our next result.
\subsection{An integro-differential mKdV hierarchy} A well known fact in integrable systems and special function theory, originally observed by Airault \cite{Airault} and Flaschka, Newell \cite{FN}, states that the ordinary Painlev\'e-II hierarchy, cf. \cite[$(3.4)$]{CM}, is obtainable through a scaling reduction of the mKdV hierarchy, cf. \cite[$(3.3)$]{CM}. When generalized to the current integro-differential setting a natural question concerns the relation of \eqref{i17} to an appropriately defined integro-differential mKdV hierarchy. Our second result settles this question in the following affirmative fashion. First, we require the below two-variable extension of Definition \ref{idef1}.
\begin{definition}\label{idef2} For $\mathbb{R}\times\mathbb{R}_+\times\mathbb{R}\ni(t_1,t_{2n+1},x)\mapsto f(t_1,t_{2n+1}|x)$, we use $D_{t_1}^{\pm1}f$ to denote its fractional $t_1$-derivatives that obey $(D_{t_1}^{-1}D_{t_1}f)(t_1,t_{2n+1}|x)=f(t_1,t_{2n+1}|x)$. Given $v=v(t_1,t_{2n+1}|x)$, we now define
\begin{align*}
	(\mathcal{L}_+^vf)(t_1,t_{2n+1}|x):=&\,\im(D_{t_1}f)(t_1,t_{2n+1}|x)-\im\big\langle (D_{t_1}^{-1}\{v,f\})(t_1,t_{2n+1}|x,\cdot),v\big\rangle-2\im\big(D_{t_1}^{-1}\langle v,f\rangle\big)v(t_1,t_{2n+1}|x),\\
	(\mathcal{L}_-^vf)(t_1,t_{2n+1}|x):=&\,\im(D_{t_1}f)(t_1,t_{2n+1}|x)+\im\big\langle (D_{t_1}^{-1}[v,f])(t_1,t_{2n+1}|x,\cdot),v\big\rangle
\end{align*}
in terms of the rank two operators $[\alpha,\beta]:=\alpha\otimes\beta-\beta\otimes\alpha$ and $\{\alpha,\beta\}:=\alpha\otimes\beta+\beta\otimes\alpha$ with kernels
\begin{align*}
	[\alpha,\beta](t_1,t_{2n+1}|x,y)=&\,\,\alpha(t_1,t_{2n+1}|x)\beta(t_1,t_{2n+1}|y)-\beta(t_1,t_{2n+1}|x)\alpha(t_1,t_{2n+1}|y),\\
	\{\alpha,\beta\}(t_1,t_{2n+1}|x,y)=&\,\,\alpha(t_1,t_{2n+1}|x)\beta(t_1,t_{2n+1}|y)+\beta(t_1,t_{2n+1}|x)\alpha(t_1,t_{2n+1}|y),
\end{align*}
and the two-variable bilinear form, with weight $w(x)$ of the general type \eqref{i13},
\begin{equation*}
	\langle f,g\rangle:=\int_{\mathbb{R}}f(t_1,t_{2n+1}|x)g(t_1,t_{2n+1}|x)w'(x)\,\d x.
\end{equation*}
\end{definition}
In turn, the relation between Painlev\'e-II and mKdV in the integro-differential setting reads as follows.
\begin{theo}\label{itheo2} Suppose $u(t|x)=u(t|x;n),n\in\mathbb{N}$ solves the integro-differential Painlev\'e-II equation
\begin{equation*}
	-(t+x)u(t|x)=\big((\mathcal{L}_+^u\mathcal{L}_-^u)^nu\big)(t|x),\ \ \ (t,x)\in\mathbb{R}^2.
\end{equation*}
Now define, with $\tau\in\mathbb{R}_+$,
\begin{equation}\label{i20bis}
	v(t_1,t_{2n+1}|x)=v(t_1,t_{2n+1}|x;n):=\frac{1}{\tau}u\left(t\Big|\frac{x}{\tau}\right),\ \ \ \ \ t_1:=\tau t\in\mathbb{R},\ \ \ \ t_{2n+1}:=\frac{\tau^{2n+1}}{2n+1}\in\mathbb{R}_+,
\end{equation}
then $v(t_1,t_{2n+1}|x)$ solves the integro-differential mKdV equation
\begin{equation}\label{i21}
	\frac{\partial v}{\partial t_{2n+1}}(t_1,t_{2n+1}|x)=\left((\mathcal{L}_-^v\mathcal{L}_+^v)^n\frac{\partial v}{\partial t_1}\right)(t_1,t_{2n+1}|x),\ \ \ \ (t_1,t_{2n+1},x)\in\mathbb{R}\times\mathbb{R}_+\times\mathbb{R}.
\end{equation}
\end{theo}
\begin{remark}
In the non integro-differential setting one proves Theorem \ref{itheo2} by computing the $t_1$ and $t_{2n+1}$ derivatives of $v$ in terms of $u$ using \eqref{i20bis} and then recovers \eqref{i21} through \eqref{i17}.
In the integro-differential setting it is not clear how to extend this procedure because of the presence of the variable $x$ and its rescaling. Instead, it is preferable to use an approach similar to the one used in \cite{CJM} which allows us, en passant, to obtain also a Lax pair for the hierarchy, see Section \ref{zsec6}.
\end{remark}
\begin{remark} Once $w'(x)=\delta_0(x)$, our recursion \eqref{i21} formally aligns with \cite[page $17$]{AM}.
\end{remark}
A special case of the integro-differential PDE \eqref{i21} occurs in the $\alpha$-dependent context of \eqref{i12}. The details are as follows.
\begin{cor}\label{icor1} Let $F_n^{\alpha}(t)=\det(I-K_{t,n}^{\alpha}\upharpoonright_{L^2(\mathbb{R}_+)})$ denote the Fredholm determinant with kernel \eqref{i12}. Then, for every $(t,\alpha,n)\in\mathbb{R}\times\mathbb{R}_+\times\mathbb{N}$,
\begin{equation}\label{i22}
	F_n^{\alpha}(t)=\exp\left[-\int_{\alpha t}^{\infty}(s-\alpha t)\left(\int_{\mathbb{R}}v^2(s,t_{2n+1}|x)\left(\frac{\e^x}{1+\e^x}\right)'\,\d x\right)\d s\right],
\end{equation}
where $v(t_1,t_{2n+1}|x)\equiv v(t_1,t_{2n+1}|x;n)$ with $t_{2n+1}=\frac{\alpha^{2n+1}}{2n+1}$ is the unique solution of the boundary value problem
\begin{equation}\label{i23}
	\frac{\partial v}{\partial t_{2n+1}}(t_1,t_{2n+1}|x)=\left((\mathcal{L}_-^v\mathcal{L}_+^v)^n\frac{\partial v}{\partial t_1}\right)(t_1,t_{2n+1}|x),\ \ \ \ \ \ v(t_1,t_{2n+1}|x)\sim\frac{1}{\alpha}\textnormal{Ai}_n\left(\frac{t_1+x}{\alpha}\right),
\end{equation}
with the last expansion valid as $t_1\rightarrow+\infty$, pointwise in $(t_{2n+1},x)\in\mathbb{R}_+\times\mathbb{R}$.
\end{cor}
\begin{rem} The operators $\mathcal{L}_{\pm}^v$ are $t$-localized, see Section \ref{zsec6} below. Furthermore, compare Remark \ref{irem2}, we do not indicate the non-trivial dependence of $v(t_1,t_{2n+1}|x)$ on the weight $w$.
\end{rem}
\begin{rem} It might seems unnatural, from a mathematical viewpoint, to deduce the integro-differential mKdV hierarchy from its self-similar reduction \eqref{i17}. We proceed in this way to underline the fact that, in applications to non-interacting fermionic systems, the inverse of the temperature ($\alpha$ in \eqref{i8} or, more precisely, $\frac{\alpha^{2n + 1}}{2n + 1}$) gives rise to an integrable dynamics; namely it plays the role of the time variable in \eqref{i23}.
\end{rem}
\begin{rem}
When $n = 1$, it was proven in \cite{QR} (for the case of $w_\alpha$ equal to the Fermi factor) and in \cite{DoussalKP} (for more general weights), that \eqref{i15} also relates to the classical KdV equation. The formalism used in \cite{CCR} provides a framework in which all three equations, for the case $n = 1$, (integro-differential mKdV, integro-differential Painlevé-II and classical KdV) can be obtained.
\end{rem}
In conclusion of this short subsection, we write out the first two members of the integro-differential mKdV hierarchy \eqref{i21}. First
\begin{equation}\label{i24}
	n=1:\ \ \ \frac{\partial v}{\partial t_3}=-\frac{\partial^3 v}{\partial t_1^3}+3\frac{\partial v}{\partial t_1}\langle v,v\rangle+3v\left\langle\frac{\partial v}{\partial t_1},v\right\rangle,
\end{equation}
and second
\begin{align}
	n=2:\ \ \ -\frac{\partial v}{\partial t_5}=&-\frac{\partial^5v}{\partial t_1^5}+5\frac{\partial^3v}{\partial t_1^3}\langle v,v\rangle+5v\left\langle\frac{\partial^3v}{\partial t_1^3},v\right\rangle+15\frac{\partial^2v}{\partial t_1^2}\left\langle\frac{\partial v}{\partial t_1},v\right\rangle+15\frac{\partial v}{\partial t_1}\left\langle\frac{\partial^2v}{\partial t_1^2},v\right\rangle\nonumber\\
	&\,\,\,+10v\left\langle\frac{\partial^2v}{\partial t_1^2},\frac{\partial v}{\partial t_1}\right\rangle+10\frac{\partial v}{\partial t_1}\left\langle\frac{\partial v}{\partial t_1},\frac{\partial v}{\partial t_1}\right\rangle-10\frac{\partial v}{\partial t_1}\langle v,v\rangle^2-20v\langle v,v\rangle\left\langle\frac{\partial v}{\partial t_1},v\right\rangle.\label{i25}
\end{align}
The first equation \eqref{i24} exactly reproduces the first member of the ordinary mKdV hierarchy when $w'(x)=\delta_0(x)$, cf. \cite[page $55$]{CJM}, and the second equation \eqref{i25} the second member up to the sign flip $t_5\mapsto -t_5$ (because of our sign convention in Lemma \ref{zlem:1}) and the obvious typo correction $-40 v_xv_{xx} \mapsto -40 v v_x v_{xx}$. This completes the current subsection.
\subsection{Other recent occurrences of \eqref{i1} and \eqref{i15}} Throughout, our motivation for the analysis of the higher order finite temperature Airy kernel determinant \eqref{i15} stems from its occurrence in the theory of non-interacting quantum many body systems, see Subsections \ref{fermi1} and \ref{fermi2}. There are, however, a few other recent studies in mathematics and mathematical physics that involve determinants of the type \eqref{i15} with general $n\in\mathbb{N}$. 
These works concern the step function weight choice $w(x)=\chi_{\mathbb{R}_+}(x)$ throughout, and we now provide a short chronological survey.\bigskip

Firstly, in \cite{CCG} it is shown that $D_n(t,\lambda)$ with a slightly different higher order Airy function than our \eqref{i1} is related to the ordinary Painlev\'e-II hierarchy. The difference stems from the normalization of the higher order Airy function in \cite[$(1.12)$]{CCG}. We emphasize that the occurrence of the ordinary Painlev\'e-II hierarchy was first established for selected values of $n$ in the arXiv version \cite{DMSa} of \cite{DMS}. Indeed, the system \cite[$(109),(110)$]{DMSa} can be transformed to a closed form differential equation with the help of conserved quantities and this transformation was made explicit for $n=1,2$ in \cite[page $16$, $(116)$]{DMSa}. The case of general $n$ is resolved in \cite{CCG}. Furthermore, \cite{CCG} derives leading order tail expansions for $D_n(t,1)$ as $t\rightarrow-\infty$, en route confirming earlier tail decay predictions in \cite{DMS}. We also mention that the higher order Airy functions in \cite{CCG} are of the form
\begin{equation}\label{i26}
	\frac{1}{\pi}\int_0^{\infty}\cos\left(\frac{y^{2n+1}}{2n+1}+\sum_{j=1}^{n-1}\omega_jy^{2j+1}+xy\right)\d y,\ \ \ \ \ x\in\mathbb{R},\ \ n\in\mathbb{N},
\end{equation}
and thus depend on additional parameters $\{\omega_j\}_{j=1}^{n-1}\subset\mathbb{R}$. The parameter dependent function \eqref{i26} has not yet appeared in the context of non-interacting fermions to the best of our knowledge. Still, it is clear that the methods developed in this paper for \eqref{i1},\eqref{i13},\eqref{i15} can be extended to the $\omega_j$ dependent setup \eqref{i26} and, as in \cite{CCG}, this extension does not present particular conceptual difficulties. Secondly, the paper \cite{Tar} investigates a matrix-valued version of the higher order Airy function \eqref{i1} and associated Fredholm determinant. In this case the Fredholm determinant connects to a fully non commutative version of the Painlev\'e-II hierarchy which is realized as the compatibility condition of a suitable matrix-valued Lax pair. Thirdly, $D_n(t,1)$ has appeared in recent studies \cite{BBW,KZ} of fine tuned Schur measures for which the typical edge fluctuation exponent $\frac{1}{3}$ gets replaced by $\frac{1}{2n+1},n\in\mathbb{N}$. The work on random partitions\footnote{The ordinary finite temperature Airy kernel determinant \eqref{i6}, i.e. a special case of \eqref{i15} with $n=1$, appears also in the theory of random partitions, precisely in models of cylindrical partitions, see \cite{BeB}.} provides a natural bridge between the zero-temperature fermionic models and the non-generic Hermitian random matrix models mentioned in Subsection \ref{fermi2}, see \cite[page $3$]{BBW}. Still, at the moment, it is not clear if the general finite-temperature determinant \eqref{i13},\eqref{i14},\eqref{i15} plays a role in the theory of random matrices.\bigskip

\subsection{Methodology and outline of paper}\label{method} The remaining sections of the paper are organized as follows. In Section \ref{zsec1} we collect a series of basic results for the higher order Airy function \eqref{i1} and the determinant \eqref{i15}: these are analytic and asymptotic properties of \eqref{i1}, the fact that \eqref{i15} is well-defined in the indicated parameter range within the class \eqref{i13}, the fact that $I-\lambda K_{t,n}$ is invertible on $L^2(\mathbb{R}_+)$ for certain values of $(t,\lambda)$ and finally the fact that $K_{t,n}$ is indeed the correlation kernel of a determinantal point process. Our work in Section \ref{zsec1} is valid for an arbitrary bounded weight function $w:\mathbb{R}\rightarrow\mathbb{R}_+$ such that $\d\sigma(x):=w'(x)\d x$ is a positive Borel probability measure on $\mathbb{R}$ with finite first moment. The need for the exponential decay in \eqref{i13} becomes clear in Section \ref{zsec2}. Indeed, at present, there are three ways one can obtain formul\ae\,of the type \eqref{i16},\eqref{i17} for a generic finite temperature Fredholm determinant. One is algebraic and was used in \cite{ACQ} in the derivation of \eqref{i8},\eqref{i9}. This approach is an extension of the original method of Tracy and Widom \cite{TW0}, recently generalized in \cite[Section $3$]{Kra} to a larger class of weighted Hankel composition operators than our \eqref{i14}. This algebraic method requires minimal decay and regularity from the weight function $w$, however its complexity relies heavily on the degree of the differential equation underlying \eqref{i1}, in our case $2n$ for given $n\in\mathbb{N}$, see Lemma \ref{zlem:1}. This makes the derivation of the full hierarchy \eqref{i17} by the algebraic method cumbersome. The second approach was first used in mathematics literature in \cite[Section $9$]{B}, though parts of it were already present in \cite[XV.$6$]{KBI}, albeit in a non-rigorous fashion. In this analytic approach one first rewrites a kernel of the form \eqref{i14} as
\begin{equation}\label{i27}
	(x-y)K_{t,n}(x,y)=\int_{\mathbb{R}}\left(\sum_{i=1}^mf_{i,t}(x|z)g_{i,t}(y|z)\right)\d\sigma(z)
\end{equation}
for some suitable functions $f_{i,t},g_{i,t}$ (this is possible for \eqref{i14}, see equation $(40)$ in the supplementary material of \cite{DMS}), and afterwards associates an operator-valued Riemann-Hilbert problem (RHP) with the resolvent of $I-K_{t,n}$, see the workings in \cite[Subsection $9.2$]{B} for $n=1$. Unfortunately, the size of the relevant RHP depends on $m$ in \eqref{i27} and thus on $n\in\mathbb{N}$. The third approach consists in associating the Fredholm determinant \eqref{i15} to a matrix-valued RHP, as done in \cite{CCR} for $n = 1$ (see also the more recent \cite{Bel} for the case of the Bessel kernel) and then in recovering the integro-differential equation as equation satisfied by the eigenfunction of the Lax pair associated to the RHP. In this way, one can deduce both \eqref{i18} and \eqref{i21}, but unfortunately also in this case the size of the RHP depends on $n$, and the extension of this procedure to generic $n$ seems non-trivial. For this reason we return to the second approach but deviate from the initial steps carried out in \cite[Subsection $9.2$]{B} and first employ Fourier analytic transformations to the kernel \eqref{i14}, see Section \ref{zsec2}. The transformations change the operator, but leave the determinant $D_n(t,\lambda)$ invariant and avoid in turn large sized operator-valued RHPs. This first step of our approach is reminiscent of the conjugation techniques used in \cite[Section $3$]{BC} for the ordinary Airy kernel and in \cite[Section $2$]{CCG} for the ordinary higher order Airy kernels, but it requires a payoff in the form of higher regularity and decay assumptions for $w$. Our second step is carried out in Section \ref{zsec3} and constitutes in the setup of the relevant $2\times 2$ operator-valued RHP, the proof of its unique solvability for certain values of $(t,\lambda,n)$ and the derivation of symmetry and small-norm corollaries. Once done we then employ the approach of \cite[Subsection $9.3$]{B} and derive an operator-valued Lax pair for the solution of the RHP, see Section \ref{zsec4}. This Lax pair naturally encodes the integro-differential hierarchy \eqref{i9} once we analyze the underlying operator kernels in Section \ref{zsec5} and exploit various symmetries of the Lax pair. In turn, the proof of Theorem \ref{itheo1} will be completed in Section \ref{zsec5} and Section \ref{zsec6} is devoted to Theorem \ref{itheo2}. Again we use operator-valued Riemann-Hilbert techniques for this part and the integro-differential PDE hierarchy \eqref{i21} follows once more naturally from an operator-valued Lax system. The last part of the article summarizes certain auxiliary results in Appendix \ref{appA} as well as all relevant operator-valued Riemann-Hilbert terminology from \cite[Subsection $9.1$]{B} in Appendix \ref{appC}.
\begin{rem} The structure of the operator-valued RHPs \ref{zmaster} and \ref{kdmaster} turns out to be canonical within a suitable class of weighted integral Hankel composition operators. In particular, the jump condition \eqref{z34} and normalization \eqref{z35} are to a large extent independent of the contour integral formul\ae\,and Fourier analytic techniques used in the present paper. See the forthcoming work \cite{Bf} for details.
\end{rem}

\section{Basic properties of the kernel \eqref{i14} and the determinant \eqref{i15}}\label{zsec1}

Recall the conditionally convergent integral \eqref{i1} used in the definition of $y=\textnormal{Ai}_n(x)$. When $x$ lies off the real axis the same integral diverges and we therefore first transform \eqref{i1} into a contour integral.
\begin{lem}\label{zlem:1} Let $\Gamma$ denote any smooth contour oriented from $\infty\e^{\im a}$ to $\infty\e^{\im b}$ with $a\in(\frac{2n\pi}{2n+1},\pi)$ and $b\in(0,\frac{\pi}{2n+1})$, compare Figure \ref{figz:1} for one possible realization. Then
\begin{equation}\label{z2}
	\textnormal{Ai}_n(z)=\frac{1}{2\pi}\int_{\Gamma}\exp\left(\im\left[\frac{\lambda^{2n+1}}{2n+1}+z\lambda\right]\right)\d\lambda
\end{equation}
constitutes the analytic continuation of \eqref{i1} to the whole $z$-plane as an entire function. Moreover, \eqref{z2} solves the differential equation $\frac{\d^{2n}}{\d z^{2n}}w=(-1)^{n+1}zw$ with real $x\rightarrow+\infty$ asymptotic behavior
\begin{equation}\label{z3}
	\textnormal{Ai}_n(x)\sim\frac{1}{\sqrt{n\pi}}\,x^{-\frac{2n-1}{4n}}\e^{-\frac{2n}{2n+1}x^{\frac{2n+1}{2n}}}\cos\phi_n(x),\ \ \ \textnormal{Ai}_n(-x)\sim\frac{1}{\sqrt{n\pi}}\,x^{-\frac{2n-1}{4n}}\cos\left(\frac{2n}{2n+1}x^{\frac{2n+1}{2n}}-\frac{\pi}{4}\right).
\end{equation}
Both estimates in \eqref{z3} are valid for every $n\in\mathbb{N}$ and we abbreviate
\begin{equation*}
	\phi_n(x):=\frac{2n}{2n+1}x^{\frac{2n+1}{2n}}\cos\left(\frac{\pi}{2n}\right)+\frac{1}{2}\left(-\frac{\pi}{2}+\frac{\pi}{2n}\right).
\end{equation*}
\begin{figure}[tbh]
\begin{tikzpicture}[xscale=0.7,yscale=0.7]
\draw [->] (-4,0) -- (4,0) node[below]{{\small $\Re z$}};
\draw [->] (0,-1) -- (0,2) node[left]{{\small $\Im z$}};

\draw [very thin, dashed, color=darkgray,-] (0,0) -- (3.153391037,1.518593088);
\draw [very thin, dashed, color=darkgray,-] (0,0) -- (-3.153391037,1.518593088);
\node [color=darkgray] at (3.9,1.6) {$\frac{\pi}{2n+1}$};
\node [color=darkgray] at (-3.9,1.6) {$\frac{2n\pi}{2n+1}$};

\draw [thick, color=red, decoration={markings, mark=at position 0.25 with {\arrow{<}}}, decoration={markings, mark=at position 0.75 with {\arrow{<}}}, postaction={decorate}] plot[domain=-2.3:2.3] ({-sinh(\x)},{0.2169*cosh(\x)}) node[above]{$\Gamma$};

\draw [fill=red, dashed, opacity=0.5] (0,0) -- (3.153391037,1.518593088) arc (25.71428571:0:3.5cm) -- (0,0);
\draw [fill=red, dashed, opacity=0.5] (0,0) -- (-3.153391037,1.518593088) arc (154.2857143:180:3.5cm) -- (0,0);

\end{tikzpicture}
\caption{An admissible choice for the integration path $\Gamma$ in \eqref{z2}.}
\label{figz:1}
\end{figure}
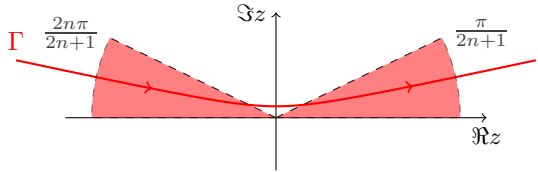
\end{lem}

\begin{proof} By \eqref{i1} for every $x\in\mathbb{R}$,
\begin{equation}\label{z4}
	\textnormal{Ai}_n(x)=\frac{1}{2\pi}\int_{-\infty}^{\infty}\exp\left(\im\left[\frac{y^{2n+1}}{2n+1}+xy\right]\right)\d y.
\end{equation}
Hence, assuming $x>0$ temporarily, we consider
\begin{equation*}
	I(r):=\int_{\Gamma_r}\exp\left(\im\left[\frac{\lambda^{2n+1}}{2n+1}+x\lambda\right]\right)\d \lambda,\ \ \ \ r>0,
\end{equation*}
 where the integration path $\Gamma_r$ connects $r$ to $r\e^{\im\frac{\pi}{2n+1}}$ along the shorter arc of $|\lambda|=r$. Since
\begin{align*}
	\big|I(r)\big|\leq\frac{r}{2n+1}&\int_0^{\pi}\exp\left(-\frac{r^{2n+1}}{2n+1}\sin\theta-xr\sin\left(\frac{\theta}{2n+1}\right)\right)\d\theta\leq\frac{r}{2n+1}\int_0^{\pi}\exp\left(-\frac{r^{2n+1}}{2n+1}\sin\theta\right)\d\theta\\
	&=\frac{2r}{2n+1}\int_0^{\frac{\pi}{2}}\exp\left(-\frac{r^{2n+1}}{2n+1}\sin\theta\right)\d\theta\leq\frac{2r}{2n+1}\int_0^{\frac{\pi}{2}}\exp\left(-\frac{r^{2n+1}}{2n+1}\frac{2\theta}{\pi}\right)\d\theta<\frac{\pi}{r^{2n}},
\end{align*}
using Jordan's inequality \cite[$4.18.1$]{NIST} in the fourth step, we see that $I(r)$ vanishes as $r\rightarrow+\infty$. Given that the same is true with $\Gamma_r$ replaced by its conjugate path reflected through the origin, Cauchy's theorem yields \eqref{z2} for any $z>0$. However, by choice of $\Gamma$ in \eqref{z2}, the factor $\exp(\im(\frac{\lambda^{2n+1}}{2n+1}))$ dominates $\exp(\im z\lambda)$ as $\Gamma\ni\lambda\rightarrow\infty$, so \eqref{z2} converges absolutely and uniformly in $z\in\mathbb{C}$ chosen from compact subsets. In turn, \eqref{z2} is the analytic continuation of \eqref{i1} to the whole plane and \eqref{z2} an entire function. Next, differentiating under the integral sign in \eqref{z2}, we find
\begin{align*}
	\frac{\d^{2n}}{\d z^{2n}}\textnormal{Ai}_n(z)-(-1)^{n+1}z\textnormal{Ai}_n(z)=&\,\,\frac{(-1)^n}{2\pi}\int_{\Gamma}\left(\lambda^{2n}+z\right)\exp\left(\im\left[\frac{\lambda^{2n+1}}{2n+1}+z\lambda\right]\right)\d\lambda\\
	=&\,\,\frac{(-1)^n}{2\pi\im}\exp\left(\im\left[\frac{\lambda^{2n+1}}{2n+1}+z\lambda\right]\right)\bigg|_{\partial\Gamma}=0,\ \ \ z\in\mathbb{C},
\end{align*}
by the aforementioned asymptotic properties of $\exp(\im(\frac{\lambda^{2n+1}}{2n+1}))$ on $\Gamma$. It remains to establish \eqref{z3} and we begin with
\begin{equation}\label{z5}
	\textnormal{Ai}_n(-x)\stackrel{\eqref{i1}}{=}\frac{1}{\pi}\,\Re\left\{\int_0^{\infty}\exp\left(\im\left[\frac{y^{2n+1}}{2n+1}-xy\right]\right)\d y\right\}=\frac{1}{\pi}x^{\frac{1}{2n}}\,\Re\left\{\int_0^{\infty}\e^{\im tp(s)}\,\d s\right\},\ \ \ x>0
\end{equation}
where $t:=x^{\frac{2n+1}{2n}}>0$ and $p(s):=\frac{s^{2n+1}}{2n+1}-s$. Using the stationary phase method \cite[$\S 15.3$]{S2} we find, 
\begin{equation*}
	\int_0^{\infty}\e^{\im tp(s)}\,\d s=\int_0^1\e^{\im tp(1-v)}\,\d v+\int_0^{\infty}\e^{\im tp(1+v)}\,\d v=\frac{1}{2}\sqrt{\frac{\pi}{nt}}\,\e^{\im\frac{\pi}{4}}\e^{-\im t\frac{2n}{2n+1}}+\frac{1}{2}\sqrt{\frac{\pi}{nt}}\,\e^{\im\frac{\pi}{4}}\e^{-\im t\frac{2n}{2n+1}}+o\big(t^{-\frac{1}{2}}\big),
\end{equation*}
as $t\rightarrow+\infty$. Taking real parts and substituting into \eqref{z5} yields the asymptotic result for $\textnormal{Ai}_n(-x)$. Finally, 
\begin{equation}\label{z6}
	\textnormal{Ai}_n(x)\stackrel{\eqref{z2}}{=}\frac{1}{2\pi}x^{\frac{1}{2n}}\int_{\Gamma}\e^{tq(\lambda)}\,\d\lambda,\ \ \ t=x^{\frac{2n+1}{2n}}>0,
\end{equation}
where $q(\lambda):=\im(\frac{\lambda^{2n+1}}{2n+1}+\lambda)$ has stationary points $\lambda_k=\e^{\frac{\im}{2n}(\pi+2\pi k)},k=0,1,\ldots,2n-1$ with
\begin{equation*}
	q(\lambda_k+z)=\frac{2n\im\lambda_k}{2n+1}+n\,\overline{\im\lambda_k}\,z^2+\mathcal{O}\big(z^3\big),\ \ z\rightarrow 0.
\end{equation*}
We thus deform $\Gamma\mapsto\Sigma$ in \eqref{z6} to pass through the stationary points $\{\lambda_k\}_{k=0}^{n-1}$ in the upper half-plane and to be tangent to the curves of constant phase $\Im q(\lambda)=\Im q(\lambda_k)$ near each point $\{\lambda_k\}_{k=0}^{n-1}$ such that $\Re q(\lambda_k+u+\im v)$ has a local maximum at $u=v=0$ along the deformed path. Observe that this is possible in a way that $\Re q(\lambda)<0$ along the deformed contour $\Sigma$ and the method of steepest descent \cite[$\S 15.4$]{S2} yields in turn
\begin{equation*}
	\textnormal{Ai}_n(x)=\frac{1}{2\pi}x^{\frac{1}{2n}}\left[\sum_{k=0}^{n-1}\e^{tq(\lambda_k)}\left(\e^{\im w_k}\sqrt{\frac{2\pi}{|q''(\lambda_k)|t}}+\mathcal{O}\left(t^{-\frac{3}{2}}\right)\right)\right],\ \ \ \ t=x^{\frac{2n+1}{2n}}\rightarrow+\infty,
\end{equation*}
where $w_k=\frac{1}{2}(-\frac{\pi}{2}+\frac{\pi+2\pi k}{2n})$ is the angle of the tangent direction at $\lambda_k$ as we traverse $\Sigma$ from $\infty\e^{\im a}$ to $\infty\e^{\im b}$. Simplifying the last expression ($k=0$ and $k=n-1$ are the dominating terms) we obtain the outstanding asymptotic result for $\textnormal{Ai}_n(x)$ and have therefore concluded our proof.
\end{proof}
Moving ahead, we now derive the central estimate used in the proof that \eqref{i15} is well-defined.
%
%
\begin{lem}\label{zlem:2} For every $t\in\mathbb{R}$ and $n\in\mathbb{N}$,
\begin{equation}\label{z8}
	\int_{\mathbb{R}}\int_{\mathbb{R}_+}\big|\textnormal{Ai}_n(x+y+t)\big|^2w(x)\,\d y\,\d x<\infty.
\end{equation}
\end{lem}
\begin{proof} By \eqref{i13} we have $w(x)\leq\hat{c}\,\e^{\omega x}$ for all $(-x)\geq x_0>0$ and $w(x)\leq 1$ for all $x\in\mathbb{R}$ with some $\hat{c}>0$. Hence, using \eqref{z3} and the fact that $y=\textnormal{Ai}_n(z)$ is entire, we obtain for every  $t\in\mathbb{R}$ and $n\in\mathbb{N}$ (the value of $c=c(t,n)>0$ below changes from line to line),
\begin{align*}
		\int_{\mathbb{R}}\bigg(\int_{\mathbb{R}_+}&\,\big|\textnormal{Ai}_n(x+y+t)\big|^2\,\d y\bigg)w(x)\,\d x\\
		=&\,\int_{\mathbb{R}_+}\left(\int_{\mathbb{R}_+}\big|\textnormal{Ai}_n(x+y+t)\big|^2\,\d y\right)w(x)\,\d x+\int_{\mathbb{R}_+}\left(\int_{\mathbb{R}_+}\big|\textnormal{Ai}_n(-x+y+t)\big|^2\,\d y\right)w(-x)\,\d x\\
		\leq&\,c\int_{\mathbb{R}_+}\int_{\mathbb{R}_+}\e^{-\frac{4n}{2n+1}(x+y+t)}\,\d y\,\d x+\int_{\mathbb{R}_+}\left(c+\int_0^x\big|\textnormal{Ai}_n(-x+y+t)\big|^2\d y\right)w(-x)\,\d x\\
		\leq&\,c\left[1+\int_{\mathbb{R}_+}\left(1+x^{\frac{1}{2n}}\right)w(-x)\,\d x\right].
\end{align*}
This proves \eqref{z8} through the imposed exponential decay of $w(-x)$ as $x\rightarrow+\infty$.
\end{proof}
\begin{cor}\label{zcor:0} The operator $K_{t,n}$ with kernel \eqref{i14} is trace class on $L^2(\mathbb{R}_+)$ for every $(t,n)\in\mathbb{R}\times\mathbb{N}$.
\end{cor}
\begin{proof} By \eqref{z8}, the linear transformations $C_{t,n}:L^2(\mathbb{R})\rightarrow L^2(\mathbb{R}_+)$ and $D_{t,n}:L^2(\mathbb{R}_+)\rightarrow L^2(\mathbb{R})$ defined as
\begin{equation*}
	\big(D_{t,n}f\big)(x):=\int_{\mathbb{R}_+}\sqrt{w(x)}\,\textnormal{Ai}_n(x+y+t)f(y)\,\d y\ \ \ \ \textnormal{and}\ \ \ \ \big(C_{t,n}g\big)(x):=\int_{\mathbb{R}}\textnormal{Ai}_n(x+y+t)\sqrt{w(y)}g(y)\,\d y,
\end{equation*}
are Hilbert-Schmidt transformations. In turn, their composition $C_{t,n}D_{t,n}=K_{t,n}:L^2(\mathbb{R}_+)\rightarrow L^2(\mathbb{R}_+)$ is trace class on $L^2(\mathbb{R}_+)$, cf. \cite[Theorem $3.7.4$]{S3}. 
\end{proof}
\begin{lem}\label{zlem:3} For every $(t,n)\in\mathbb{R}\times\mathbb{N}$, the self-adjoint operator $K_{t,n}$ satisfies $0\leq K_{t,n}\leq 1$ and $I-\lambda K_{t,n}$ is invertible on $L^2(\mathbb{R}_+)$ for all $\lambda\in\overline{\mathbb{D}_1(0)}=\{\lambda\in\mathbb{C}:\,|\lambda|\leq 1\}$.
\end{lem}
\begin{proof} Integrating by parts and using \eqref{z3},\eqref{i13} we find from \eqref{i14},
\begin{equation*}
	\frac{\d K_{t,n}}{\d t}(x,y)=-\int_{\mathbb{R}}\textnormal{Ai}_n(x+z+t)\textnormal{Ai}_n(z+t+t)\,\d\sigma(z),\ \ \ (x,y)\in\mathbb{R}_+^2
\end{equation*}
with the positive Borel probability measure $\d\sigma(z)=w'(z)\d z$. Hence, by the dominated convergence theorem, Fubini's theorem and \eqref{z3},
\begin{equation*}
	K_{t,n}(x,y)=-\int_t^{\infty}\frac{\d K_{s,n}}{\d s}(x,y)\,\d s=\int_{\mathbb{R}}\left[\int_{\mathbb{R}_+}\textnormal{Ai}_n(x+z+s+t)\textnormal{Ai}_n(z+y+s+t)\,\d s\right]\d\sigma(z).
\end{equation*}
Using this representation for the kernel of $K_{t,n}$ we derive for any $f\in L^2(\mathbb{R}_+)$ (note that $\textnormal{Ai}_n:\mathbb{R}\rightarrow\mathbb{R}$)
\begin{equation}\label{z11}
	\langle f,K_{t,n}f\rangle_{L^2(\mathbb{R}_+)}=\int_{\mathbb{R}}\left[\int_{z+t}^{\infty}\left|\int_{\mathbb{R}}\textnormal{Ai}_n(x+s)f_+(x)\,\d x\right|^2\d s\right]\d\sigma(z)\geq 0
\end{equation}
with $f_+(x):=f(x)\chi_{\mathbb{R}_+}(x)$ using the characteristic function of the half ray $\mathbb{R}_+$. However, from the Fourier representation \eqref{z4} we obtain, for any $s\in\mathbb{R}$,
\begin{equation*}
	\int_{\mathbb{R}}\textnormal{Ai}_n(x+s)f_+(x)\,\d x=\frac{1}{\sqrt{2\pi}}\int_{\mathbb{R}}\exp\left(\im\left[\frac{y^{2n+1}}{2n+1}+sy\right]\right)\check{f}_+(-y)\,\d y=\check{g}_n(-s),
\end{equation*}
with $\check{f}_+(y):=\frac{1}{\sqrt{2\pi}}\int_{\mathbb{R}}\e^{-\im xy}f_+(x)\,\d x$ and $g(y):=\e^{\im\frac{y^{2n+1}}{2n+1}}\check{f}_+(-y)$. This allows us to estimate \eqref{z11} as follows,
\begin{align}
	\langle f,K_{t,n}f\rangle_{L^2(\mathbb{R}_+)}=&\,\int_{\mathbb{R}}\left[\int_{z+t}^{\infty}|\check{g}_n(-s)|^2\,\d s\right]\d\sigma(z)\leq\int_{\mathbb{R}}\left[\int_{\mathbb{R}}|\check{g}_n(-s)|^2\d s\right]\d\sigma(z)=\|\check{g}_n\|_{L^2(\mathbb{R})}^2=\|g_n\|_{L^2(\mathbb{R})}^2\nonumber\\
	=&\,\|\check{f}_+\|_{L^2(\mathbb{R})}^2=\|f_+\|_{L^2(\mathbb{R})}^2=\langle f,f\rangle_{L^2(\mathbb{R}_+)},\label{z12}
\end{align}
given that $\d\sigma$ is a probability measure and by using Plancherel's theorem in the third and fifth equality. Combining \eqref{z11} and \eqref{z12} we have therefore $0\leq K_{t,n}\leq 1$ for every $(t,n)\in\mathbb{R}\times\mathbb{N}$ and in turn, by self-adjointness for the $L^2(\mathbb{R}_+)$ operator norm,
\begin{equation*}
	\|K_{t,n}\|=\sup\big\{\langle f,K_{t,n}f\rangle_{L^2(\mathbb{R}_+)}:\,\|f\|_{L^2(\mathbb{R}_+)}=1\big\}\leq 1.
\end{equation*}
The last estimate proves invertibility of $I-\lambda K_{t,n}$ on $L^2(\mathbb{R}_+)$ in the open disk $|\lambda|<1$ by the Neumann series. For the corresponding statement on all of $\overline{\mathbb{D}_1(0)}$ we use that $K_{t,n}$ is a compact operator on $L^2(\mathbb{R}_+)$ by Corollary \ref{zcor:0} and thus assume there exists $f\in L^2(\mathbb{R}_+)\setminus\{0\}$ such that $\e^{\im\theta} K_{t,n}f=f$ for some $\theta\in[0,2\pi)$. Since for this $f$,
\begin{equation*}
	\e^{-\im\theta}\langle f,K_{t.n}f\rangle_{L^2(\mathbb{R}_+)}=\langle f,\e^{\im\theta}K_{t.n}f\rangle_{L^2(\mathbb{R}_+)}=\|f\|_{L^2(\mathbb{R}_+)}>0,
\end{equation*}
we conclude from \eqref{z11} that necessarily $\theta=0$. But then all inequalities in \eqref{z12} must be equalities, so in particular
\begin{equation*}
	\int_{\mathbb{R}}\left[\int_{z+t}^{\infty}|\check{g}_n(-s)|^2\,\d s\right]\d\sigma(z)=\|\check{g}_n\|_{L^2(\mathbb{R})}^2=\int_{\mathbb{R}}|\check{g}_n(-s)|^2\,\d s,
\end{equation*}
which yields 
\begin{equation}\label{z13}
	\forall\,t\in\mathbb{R}:\ \ \ \ \int_{\mathbb{R}}\left[\int_{-\infty}^{z+t}|\check{g}_n(-s)|^2\,\d s\right]\d\sigma(z)=0.
\end{equation}
Since $\d\sigma$ is an absolutely continuous positive Borel measure, recall \eqref{i13}, we now derive from \eqref{z13} that
\begin{equation*}
	 \int_{-\infty}^{y}|\check{g}_n(-s)|^2\,\d s=0\ \ \textnormal{a.e.}\ \ \ \ \ \Rightarrow\ \ \ \ \ \ \check{g}_n(-y)=\int_{\mathbb{R}_+}\textnormal{Ai}_n(x+y)f(x)\,\d x=0\ \ \textnormal{a.e.}.
\end{equation*}
But the last integral is continuous as function of $y\in\mathbb{C}$ by the dominated convergence theorem, Cauchy-Schwarz, Lemma \ref{zlem:1} and thus an entire function by Fubini's and Morera's theorem (given that \eqref{z2} is also entire, see Lemma \ref{zlem:1}). Hence, by the identity theorem, $\check{g}_n(z)\equiv 0$ for all $z\in\mathbb{C}$ and thus $f=0\in L^2(\mathbb{R}_+)$ in contradiction to our initial assumption. All together, this shows that $I-\lambda K_{t,n}$ is injective for $|\lambda|=1$, hence invertible by the Fredholm Alternative. This concludes our proof.
\end{proof}
\begin{cor}\label{zcor:1} For every $(t,\lambda,n)\in\mathbb{R}\times[0,1]\times\mathbb{N}$ there exists a unique determinantal point process with correlation kernel $\lambda K_{t,n}$ and the distribution function of the last particle in this process equals $D_n(t,\lambda)$.
\end{cor}
\begin{proof} Since $0\leq \lambda K_{t,n}\leq 1$ in the indicated parameter range by Lemma \ref{zlem:3}, abstract theory \cite[Theorem $3$]{Sos} guarantees existence of a unique determinantal process with correlation kernel $\lambda K_{t,n}$. Moreover, since 
\begin{equation*}
	\int_{\mathbb{R}_+}K_{t,n}(x,x)\,\d x<\infty
\end{equation*}
by Lemma \ref{zlem:2}, this process almost surely has a last particle, see \cite[Theorem $4$a)]{Sos}, with distribution function $D_n(t,\lambda)$, compare \cite[Lemma $3$]{Sos} and \cite[Proposition $2.4$]{Joh} as well as \cite[Theorem $3.10$]{S}, using that $K_{t,n}$ is trace class on $L^2(\mathbb{R}_+)$ with continuous kernel, so 
\begin{equation*}
	\tr_{L^2(\mathbb{R}_+)}K_{t,n}=\int_{\mathbb{R}_+}K_{t,n}(x,x)\,\d x
\end{equation*}
by \cite[Theorem $3.9$]{S}. This completes the proof.
\end{proof}
More directly we can show as follows that the determinant \eqref{i15} is the distribution function of some random variable for all $(t,\lambda,n)\in\mathbb{R}\times[0,1]\times\mathbb{N}$. Indeed, we already know that $t\mapsto D_n(t,\cdot)$ is differentiable and $0\leq K_{t,n}\leq 1$ yields $D_n(t,\lambda)\in[0,1]$ for all $(t,\lambda,n)\in\mathbb{R}\times[0,1]\times\mathbb{N}$. Also, by \eqref{z11} for every $f\in L^2(\mathbb{R}_+)$,
\begin{equation*}
	\langle f,K_{s,n}f\rangle_{L^2(\mathbb{R}_+)}>\langle f,K_{t,n}f\rangle_{L^2(\mathbb{R}_+)}\ \ \ \ \textnormal{whenever}\ \ s<t,
\end{equation*}
i.e. $t\mapsto D_n(t,\cdot)$ is strictly increasing on $\mathbb{R}$ for $(\lambda,n)\in[0,1]\times\mathbb{N}$. It thus remains to analyze the limiting behavior of $D_n(t,\lambda)$ as $t\rightarrow\pm\infty$. First, by positivity of $K_{t,n}$, for every $(t,\lambda,n)\in\mathbb{R}\times[0,1]\times\mathbb{N}$,
\begin{equation}\label{z14}
	D_n(t,\lambda)=\exp\left[-\int_0^{\lambda}\tr_{L^2(\mathbb{R}_+)}\big((I-\mu K_{t,n})^{-1}K_{t,n}\big)\,\d\mu\right]\leq\exp\left[-\lambda\tr_{L^2(\mathbb{R}_+)} K_{t,n}\right],
\end{equation}
where, assuming $t<0$,
\begin{equation*}
	\tr_{L^2(\mathbb{R}_+)}K_{t,n}=\int_{\mathbb{R}_+}K_{t,n}(x,x)\,\d x\stackrel{\eqref{i14}}{=}\int_t^{\infty}\left[\int_{\mathbb{R}}\textnormal{Ai}_n^2(x+z)\,w(z)\d z\right]\d x\stackrel{\eqref{z8}}{=}c+\int_t^0\left[\int_{\mathbb{R}}\textnormal{Ai}_n^2(x+z)\,w(z)\d z\right]\d x,
\end{equation*}
with $0<c<\infty$ independent of $t$, so by Fubini's theorem and the fact that $w$ is strictly increasing,
\begin{align*}
	\tr_{L^2(\mathbb{R}_+)}K_{t,n}=&\,c+\int_t^0\left[\int_{\mathbb{R}}\textnormal{Ai}_n^2(y)\,w(y-x)\d y\right]\d x=c+\int_{\mathbb{R}}\textnormal{Ai}_n^2(y)\left[\int_t^0w(y-x)\,\d x\right]\d y\\
	=&\,c+\int_{\mathbb{R}}\textnormal{Ai}_n^2(y)\left[\int_0^{|t|}w(y+x)\,\d x\right]\d y> c+|t|\int_{\mathbb{R}}\textnormal{Ai}_n^2(y)\,w(y)\d y\stackrel{\eqref{z3}}{\geq} c(1+|t|),\ \ \ c>0.
\end{align*}
As expected, $K_{t,n}$ becomes therefore unbounded in trace norm for large negative $t$, so by \eqref{z14},
\begin{equation*}
	D_n(t,\lambda)\leq\exp\left[-\lambda\tr_{L^2(\mathbb{R}_+)}K_{t,n}\right]\leq\e^{-c\lambda(1+|t|)}\rightarrow 0\ \ \ \ \textnormal{as}\ \ t\rightarrow-\infty.
\end{equation*}
The remaining behavior of $D_n(t,\lambda)$ as $t\rightarrow+\infty$ can be derived as follows: since $K_{t,n}$ is the composition of two Hilbert-Schmidt transformations by the proof of Corollary \ref{zcor:0}, we have in trace norm
\begin{equation*}
	\|K_{t,n}\|_1\leq \|M_{t,n}\|_2\|N_{t,n}\|_2=\int_{\mathbb{R}}\left(\int_{\mathbb{R}_+}|\textnormal{Ai}_n(x+y+t)|^2\d y\right)w(x)\d x,
\end{equation*}	
with the Hilbert-Schmidt norm $\|\cdot\|_2$. Hence, we can find $c=c(n)>0$ such that for every $t>0$,
\begin{equation}\label{z15}
	\|K_{t,n}\|_1\leq c\,\e^{-\frac{4nt}{2n+1}}+\int_{\mathbb{R}_+}\left(\int_{\mathbb{R}_+}|\textnormal{Ai}_n(-x+y+t)|^2\,\d y\right)w(-x)\d x.
\end{equation}
However, pointwise in $x\in\mathbb{R}_+$ by \eqref{z3},
\begin{equation}\label{z16}
	\int_{\mathbb{R}_+}|\textnormal{Ai}_n(-x+y+t)|^2\,\d y\leq c\,\e^{-\frac{4nt}{2n+1}}+\int_{t-x}^t|\textnormal{Ai}_n(y)|^2\,\d y\rightarrow 0\ \ \ \textnormal{as}\ \ t\rightarrow+\infty,
\end{equation}
and, uniformly in $(t,x)\in\mathbb{R}_+^2$,
\begin{equation*}
	w(-x)\int_{\mathbb{R}_+}|\textnormal{Ai}_n(-x+y+t)|^2\,\d y\stackrel{\eqref{z16}}{\leq}c(1+x)w(-x)\in L^1(\mathbb{R}_+).
\end{equation*}
Thus, by the dominated convergence theorem, \eqref{z15} yields $K_{t,n}\rightarrow 0$ in trace norm as $t\rightarrow+\infty$ and therefore
\begin{equation*}
	\big|D_n(t,\lambda)-1\big|\leq\lambda\|K_{t,n}\|_1\e^{\|K_{t,n}\|_1+1}\rightarrow 0\ \ \ \ \textnormal{as}\ \  t\rightarrow+\infty\ \ \ \textnormal{for every}\ (\lambda,n)\in[0,1]\times\mathbb{N},
\end{equation*}
cf. \cite[Theorem $3.4$]{S}. This proves the outstanding limiting behavior of $D_n(t,\lambda)$.
\begin{rem}\label{momrem} As can be seen from the proofs of Lemma \ref{zlem:2} and the discussion below Corollary \ref{zcor:1}, we could have chosen \eqref{i13} with $w'(x)\leq |x|^{-2-\epsilon}$ for $|x|\geq x_0$ and arbitrary $\epsilon>0$ instead of the exponential decay. This would not alter the results in Section \ref{zsec1}, however we will need the rapid decay of $w'(x)$ in the upcoming Section \ref{zsec2}.
\end{rem}

\section{Exploiting the determinant's conjugation invariance}\label{zsec2}
In this section we make crucial use of the fact that the operator traces
\begin{equation*}
	\tr_{L^2(\mathbb{R}_+)}(K_{t,n})^m,\ \ \ \ \ m\in\mathbb{N},
\end{equation*}
and hence the determinant $D_n(t,\lambda)$ itself, remain invariant under conjugation of $K_{t,n}$ with bounded invertible $L^2(\mathbb{R}_+)$ operators. More precisely, our next move will rephrase the combination $K_{t,n}(x,y)\chi_{\mathbb{R}_+}(y)$ in Fourier variables and we begin with the following contour integral representations for $\textnormal{Ai}_n(z)$,
\begin{equation}\label{z17}
	\textnormal{Ai}_n(x)\stackrel[\eqref{z4}]{\eqref{z2}}{=}\frac{1}{2\pi}\int_{\Gamma_{\alpha}}\e^{\im\psi_n(\alpha,x)}\,\d\alpha=\frac{1}{2\pi}\int_{\Gamma_{\beta}}\e^{-\im\psi_n(\beta,x)}\,\d\beta,\ \ \ x\in\mathbb{R};\ \ \ \ \ \ \ \psi_n(\lambda,z):=\frac{\lambda^{2n+1}}{2n+1}+z\lambda,
\end{equation}
where $\Gamma_{\alpha}$, resp. $\Gamma_{\beta}$, denotes any smooth contour oriented from $\infty\e^{\im a}$ to $\infty\e^{\im b}$, resp. $\infty\e^{\im c}$ to $\infty\e^{\im d}$, with 
\begin{equation*}
	a\in\left(\frac{2n\pi}{2n+1},\pi\right]\ \ \textnormal{and}\ \ b\in\left[0,\frac{\pi}{2n+1}\right),\ \ \ \ \textnormal{resp.}\ \ \ \ c\in\left(\pi,\frac{(2n+2)\pi}{2n+1}\right)\ \ \textnormal{and}\ \  d\in\left(\frac{(4n+1)\pi}{2n+1},2\pi\right),
\end{equation*}
such that $0<\Im(\alpha-\beta)<\frac{\omega}{2}$ and $\Im\beta<0$ is satisfied for $\alpha\in\Gamma_{\alpha}$ and $\beta\in\Gamma_{\beta}$ with $\omega>0$ as in \eqref{i13}, see Figure \ref{figz:2} below for a possible choice. Recognizing these constraints for the contours we find in turn from \eqref{i13} that
\begin{equation*}
	\forall\,(\alpha,\beta)\in\Gamma_{\alpha}\times\Gamma_{\beta}:\ \ \ \ \lim_{\substack{z\rightarrow+\infty\\ z\in\mathbb{R}}}\e^{\im z(\alpha-\beta)}w(z)=0,\ \ \ \ \lim_{\substack{z\rightarrow-\infty\\ z\in\mathbb{R}}}\e^{\im z(\alpha-\beta)}w(z)=0.
\end{equation*}
Hence, upon insertion of \eqref{z17} into \eqref{i14} and integration by parts,
\begin{eqnarray}
	K_{t,n}(x,y)&=&\frac{1}{(2\pi)^2}\int_{\Gamma_{\alpha}}\int_{\Gamma_{\beta}}\e^{\im(\psi_n(\alpha,x+t)-\psi_n(\beta,y+t))}\left[\int_{\mathbb{R}}\e^{\im z(\alpha-\beta)}w(z)\,\d z\right]\d\beta\d\alpha\nonumber\\
	&=&\frac{\im}{(2\pi)^2}\int_{\mathbb{R}}\left[\int_{\Gamma_{\alpha}}\int_{\Gamma_{\beta}}\e^{\im\psi_n(\alpha,x+z+t)}\e^{-\im\psi_n(\beta,z+y+t)}\frac{\d\beta\,\d\alpha}{\alpha-\beta}\right]\d\sigma(z),\label{z18}
\end{eqnarray}
with $\d\sigma(z)=w'(z)\d z$ as before in the proof of Lemma \ref{zlem:3}.
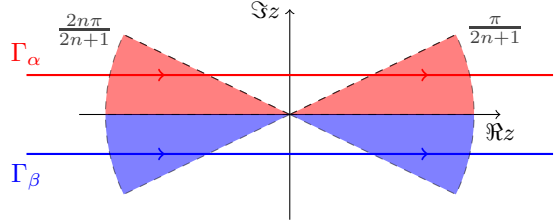
\begin{figure}[tbh]
\begin{tikzpicture}[xscale=0.7,yscale=0.7]
\draw [->] (-4,0) -- (4,0) node[below]{{\small $\Re z$}};
\draw [->] (0,-2) -- (0,2) node[left]{{\small $\Im z$}};

\draw [very thin, dashed, color=darkgray,-] (0,0) -- (3.153391037,1.518593088);
\draw [very thin, dashed, color=darkgray,-] (0,0) -- (-3.153391037,1.518593088);
\node [color=darkgray] at (3.9,1.6) {$\frac{\pi}{2n+1}$};
\node [color=darkgray] at (-3.9,1.6) {$\frac{2n\pi}{2n+1}$};

\draw [thick, color=red, decoration={markings, mark=at position 0.25 with {\arrow{<}}}, decoration={markings, mark=at position 0.75 with {\arrow{<}}}, postaction={decorate},-] (5,0.75) -- (-5,0.75) node[above]{$\Gamma_{\alpha}$};
\draw [thick, color=blue, decoration={markings, mark=at position 0.25 with {\arrow{<}}}, decoration={markings, mark=at position 0.75 with {\arrow{<}}}, postaction={decorate},-] (5,-0.75) -- (-5,-0.75) node[below]{$\Gamma_{\beta}$};

\draw [fill=red, dashed, opacity=0.5] (0,0) -- (3.153391037,1.518593088) arc (25.71428571:0:3.5cm) -- (0,0);
\draw [fill=red, dashed, opacity=0.5] (0,0) -- (-3.153391037,1.518593088) arc (154.2857143:180:3.5cm) -- (0,0);
\draw [fill=blue, dashed, opacity=0.5] (0,0) -- (3.153391037,-1.518593088) arc (-25.71428571:0:3.5cm) -- (0,0);
\draw [fill=blue, dashed, opacity=0.5] (0,0) -- (-3.153391037,-1.518593088) arc (-154.2857143:-180:3.5cm) -- (0,0);

\end{tikzpicture}
\caption{An admissible (and very simple) choice for the integration paths $\Gamma_{\alpha}$ and $\Gamma_{\beta}$ in \eqref{z17}, ensuring throughout $0<\Im(\alpha-\beta)<\frac{\omega}{2}$ and $\Im\beta<0$ for $(\alpha,\beta)\in\Gamma_{\alpha}\times\Gamma_{\beta}$.}
\label{figz:2}
\end{figure}
Next, since $D_n(t,\lambda)=\det(I-\lambda K_{t,n}\chi_{\mathbb{R}_+}\upharpoonright_{L^2(\mathbb{R})})$ where $\chi_{\mathbb{R}_+}$ is the operator of multiplication by the characteristic function $\chi_{\mathbb{R}_+}(x)$ of $\mathbb{R}_+$, we use the following integral identity, cf. \cite[Lemma $2.2$]{BB} for a similar one: for every $\beta\in\Gamma_{\beta}$ (with $\Im\beta<0$ in our setup) and $y\in\mathbb{R}\setminus\{0\}$,
\begin{equation}\label{z19}
	-\frac{1}{2\pi\im}\int_{\mathbb{R}}\e^{-\im y(\gamma-\beta)}\frac{\d\gamma}{\gamma-\beta}=\chi_{\mathbb{R}_+}(y).
\end{equation}
Hence, combining \eqref{z18} with \eqref{z19} and fixing $\Gamma_{\alpha}=\mathbb{R}$ for simplicity (this is an admissible choice for $\Gamma_{\alpha}$, see \eqref{z4}), we obtain by Fubini's theorem
\begin{equation}\label{z20}
	K_{t,n}(x,y)\chi_{\mathbb{R}_+}(y)=\int_{\mathbb{R}}\int_{\mathbb{R}}\frac{\e^{\im x\alpha}}{\sqrt{2\pi}}\underbrace{\left[\frac{1}{(2\pi)^2}\int_{\mathbb{R}}\int_{\Gamma_{\beta}}\frac{\e^{\im\psi_n(\alpha,z+t)}\e^{-\im\psi_n(\beta,z+t)}}{(\alpha-\beta)(\beta-\gamma)}\,\d\beta\,\d\sigma(z)\right]}_{=:L_{t,n}(\alpha,\gamma)}\frac{\e^{-\im y\gamma}}{\sqrt{2\pi}}\,\d\alpha\,\d\gamma,
\end{equation}
which shows that $K_{t,n}\chi_{\mathbb{R}_+}$ on $L^2(\mathbb{R})$ equals the operator composition $\mathcal{F}^{-1}L_{t,n}\,\mathcal{F}$. Here, $L_{t,n}:L^2(\mathbb{R})\rightarrow L^2(\mathbb{R})$ has kernel $L_{t,n}(\alpha,\gamma)$ as written in \eqref{z20} and $\mathcal{F}:L^1(\mathbb{R})\cap L^2(\mathbb{R})\rightarrow L^2(\mathbb{R})$ is the standard Fourier transform
\begin{equation}\label{z21}
	(\mathcal{F}f)(x):=\frac{1}{\sqrt{2\pi}}\int_{\mathbb{R}}\e^{-\im xy}f(y)\,\d y,\ \ \ \ \ \ x\in\mathbb{R},
\end{equation}
that extends unitarily on $L^2(\mathbb{R})$. Note that, by general trace ideal properties, $L_{t,n}=\mathcal{F}K_{t,n}\chi_{\mathbb{R}_+}\mathcal{F}^{-1}$ is trace class on $L^2(\mathbb{R})$, alternatively we can argue as follows.
\begin{lem}\label{zlem:4} The integral operator $L_{t,n}:L^2(\mathbb{R})\rightarrow L^2(\mathbb{R})$ with kernel written in \eqref{z20} is trace class for every $(t,n)\in\mathbb{R}\times\mathbb{N}$.
\end{lem}
\begin{proof} Define the linear transformations $M_{t,n}:L^2(\Gamma_{\beta})\rightarrow L^2(\mathbb{R})$ and $N_{t,n}:L^2(\mathbb{R})\rightarrow L^2(\Gamma_{\beta})$\footnote{We equip $L^2(\Gamma_{\beta})$ and later on $L^2(\Gamma_{\alpha})$ with the arc-length measure.} via
\begin{equation*}
	(M_{t,n}f)(\alpha)=\int_{\Gamma_{\beta}}M_{t,n}(\alpha,\beta)f(\beta)\,\d\beta,\ \ \ \ \ \ M_{t,n}(\alpha,\beta):=\frac{1}{2\pi}\frac{\e^{\im\psi_n(\alpha,t)}\e^{-\frac{\im}{2}\psi_n(\beta,t)}}{\alpha-\beta}\left[\int_{\mathbb{R}}\e^{\im z(\alpha-\beta)}\,\d\sigma(z)\right]
\end{equation*}
and
\begin{equation*}
	(N_{t,n}g)(\beta)=\int_{\mathbb{R}}N_{t,n}(\beta,\gamma)g(\gamma)\,\d\gamma,\ \ \ \ \ \ \ N_{t,n}(\beta,\gamma):=\frac{1}{2\pi}\frac{\e^{-\frac{\im}{2}\psi_n(\beta,t)}}{\beta-\gamma}.
\end{equation*}
Since, uniformly in $(\alpha,\beta)\in\mathbb{R}\times\Gamma_{\beta}$,
\begin{equation*}
	\left|\int_{\mathbb{R}}\e^{\im z(\alpha-\beta)}\,\d\sigma(z)\right|\leq\int_{\mathbb{R}}\e^{z\Im\beta}\,\d\sigma(z)\leq\int_{\mathbb{R}_+}\!\!\d\sigma(z)+\int_{\mathbb{R}_+}\!\!\e^{-z\Im\beta}\d\sigma(-z)\stackrel{\eqref{i13}}{\leq}1+c\int_{\mathbb{R}_+}\!\!\e^{-z(\omega+\Im\beta)}\,\d z<\infty
\end{equation*}
with $c>0$, we find at once, for every $(t,n)\in\mathbb{R}\times\mathbb{N}$,
\begin{equation*}
	\int_{\Gamma_{\beta}}\int_{\mathbb{R}}\big|M_{t,n}(\alpha,\beta)\big|^2\d\alpha\,|\d\beta|<\infty\ \ \ \ \ \textnormal{and}\ \ \ \ \ \int_{\mathbb{R}}\int_{\Gamma_{\beta}}\big|N_{t,n}(\beta,\gamma)\big|^2|\d\beta|\,\d\gamma<\infty.
\end{equation*}
In short, $M_{t,n}$ and $N_{t,n}$ are both Hilbert-Schmidt transformations and thus their composition $K_{t,n}=M_{t,n}N_{t,n}$ trace class on $L^2(\mathbb{R})$, cf. \cite[Theorem $3.7.4$]{S3}. This completes the proof.
\end{proof}
In our next step we consider $P_n:L^2(\mathbb{R})\rightarrow L^2(\mathbb{R})$ as multiplication
\begin{equation}\label{z22}
	(P_nf)(\alpha):=\e^{-\frac{\im}{2}\psi_n(\alpha,0)}f(\alpha),
\end{equation}
and observe that $L_{t,n}=(M_{t,n}N_{t,n}P_n^{-1})P_n:L^2(\mathbb{R})\rightarrow L^2(\mathbb{R})$ is trace class by Lemma \ref{zlem:4}. The same is true for $P_nL_{t,n}P_n^{-1}$ as shown in the upcoming Lemma.
\begin{lem}\label{zlem:5} The integral operator $J_{t,n}:=P_nL_{t,n}P_n^{-1}:L^2(\mathbb{R})\rightarrow L^2(\mathbb{R})$ is trace class for every $(t,n)\in\mathbb{R}\times\mathbb{N}$.
\end{lem} 
\begin{proof}
We can either use abstract trace ideal reasoning ($P_n$ and $P_n^{-1}$ are bounded on $L^2(\mathbb{R})$) or we note that $J_{t,n}$ has kernel
\begin{equation}\label{z23}
	J_{t,n}(\alpha,\gamma)=\frac{1}{(2\pi)^2}\int_{\Gamma_{\beta}}\frac{\e^{\frac{\im}{2}\psi_n(\alpha,2t)-\frac{\im}{2}\psi_n(\beta,2t)}}{\alpha-\beta}\left[\int_{\mathbb{R}}\e^{\im z(\alpha-\beta)}\,\d\sigma(z)\right]\frac{\e^{-\frac{\im}{2}\psi_n(\beta,0)+\frac{\im}{2}\psi_n(\gamma,0)}}{\beta-\gamma}\,\d\beta
\end{equation}
and can thus be factored as $J_{t,n}=A_{t,n}B_{n}$ where $A_{t,n}:L^2(\Gamma_{\beta})\rightarrow L^2(\mathbb{R})$ and $B_{n}:L^2(\mathbb{R})\rightarrow L^2(\Gamma_{\beta})$ are Hilbert-Schmidt transformations with kernels
\begin{equation}\label{z24}
	A_{t,n}(\alpha,\beta):=\frac{1}{2\pi}\frac{\e^{\frac{\im}{2}\psi_n(\alpha,2t)-\frac{\im}{2}\psi_n(\beta,2t)}}{\alpha-\beta}\left[\int_{\mathbb{R}}\e^{\im z(\alpha-\beta)}\,\d\sigma(z)\right],\ \ \ \ \ \ B_{n}(\beta,\gamma):=\frac{1}{2\pi}\frac{\e^{-\frac{\im}{2}\psi_n(\beta,0)+\frac{\im}{2}\psi_n(\gamma,0)}}{\beta-\gamma}.
\end{equation}
This concludes our proof.
\end{proof}
We now summarize the above steps:
\begin{prop}\label{zprop:1} For every $(t,\lambda,n)\in\mathbb{R}\times\mathbb{C}\times\mathbb{N}$, on $L^2(\mathbb{R})$,
\begin{equation*}
	I-\lambda K_{t,n}\chi_{\mathbb{R}_+}=\mathcal{F}^{-1}P_n^{-1}(I-\lambda J_{t,n})P_n\mathcal{F},
\end{equation*}
where the bounded linear operators $\mathcal{F},P_n$ and $J_{t,n}$ on $L^2(\mathbb{R})$ are defined in \eqref{z21},\eqref{z22} and \eqref{z23}. In particular we record the determinant equality
\begin{equation}\label{z25}
	D_n(t,\lambda)=\det(I-\lambda J_{t,n}\upharpoonright_{L^2(\mathbb{R})}).
\end{equation}
\end{prop}
\begin{proof} By Sylvester's determinant identity \cite[Chapter IV, $(5.9)$]{GGK} and by tracing back our steps in \eqref{z20} and \eqref{z23}.
\end{proof}
The above Proposition concludes the content of Section \ref{zsec2}.

\section{The operator-valued Riemann-Hilbert problem}\label{zsec3}
In this section we introduce a distinguished operator-valued Riemann-Hilbert problem (RHP) which will be central in the proof of Theorem \ref{itheo1}. Our starting point is \eqref{z25} above and the following stability result, cf. \cite[Proposition $1$]{TW} for a somewhat similar argument.
\begin{prop}\label{zprop:2} Let $\Gamma_{\alpha}$ denote the reflection of $\Gamma_{\beta}$ across the real axis where we fix $\Gamma_{\beta}:=\mathbb{R}-\im\Delta$ with $0<\Delta<\frac{\omega}{2}$, see Figure \ref{figz:3} below for the contours. Now define $J_{t,n}^{\circ}:L^2(\Gamma_{\alpha})\rightarrow L^2(\Gamma_{\alpha})$ as
\begin{equation*}
	(J_{t,n}^{\circ}f)(\xi):=\int_{\Gamma_{\alpha}}J_{t,n}(\xi,\eta)f(\eta)\,\d\eta,\ \ \ \ f\in L^2(\Gamma_{\alpha}),
\end{equation*}
with kernel $J_{t,n}(\xi,\eta)$ given in \eqref{z23}. Then $J_{t,n}^{\circ}$ is trace class on $L^2(\Gamma_{\alpha})$ and we have the equality
\begin{equation}\label{z26}
	D_n(t,\lambda)=\det(I-\lambda J_{t,n}^{\circ}\upharpoonright_{L^2(\Gamma_{\alpha})}),\ \ \ \ \ \ \ \ \ (t,\lambda,n)\in\mathbb{R}\times\mathbb{C}\times\mathbb{N}.
\end{equation}
\end{prop}
\begin{proof} The operator $J_{t,n}^{\circ}$ is well-defined on $L^2(\Gamma_{\alpha})$ since $\Gamma_{\alpha}\cap\Gamma_{\beta}=\emptyset$ by our choice of contours and because of our discussion in Subsection \ref{FSint}. More is true: if $A_{t,n}^{\circ}:L^2(\Gamma_{\beta})\rightarrow L^2(\Gamma_{\alpha})$ and $B_{n}^{\circ}:L^2(\Gamma_{\alpha})\rightarrow L^2(\Gamma_{\beta})$ are the linear transformations with kernels given in \eqref{z24}, then 
\begin{equation*}
	\int_{\Gamma_{\alpha}}\int_{\Gamma_{\beta}}\big|A_{t,n}(\alpha,\beta)\big|^2\,|\d\beta|\,|\d\alpha|<\infty\ \ \ \ \ \textnormal{and}\ \ \ \ \ \int_{\Gamma_{\beta}}\int_{\Gamma_{\alpha}}\big|B_{n}(\beta,\alpha)\big|^2\,|\d\alpha|\,|\d\beta|<\infty,
\end{equation*}
so their composition $A_{t,n}^{\circ}B_{n}^{\circ}=J_{t,n}^{\circ}$ is trace class on $L^2(\Gamma_{\alpha})$. However, for any $m\in\mathbb{N}$,
\begin{equation}\label{z27}
	\tr_{L^2(\mathbb{R})}J_{t,n}^m=\int_{\mathbb{R}}\cdots\int_{\mathbb{R}}J_{t,n}(\zeta_1,\zeta_2)\cdot\ldots\cdot J_{t,n}(\zeta_{m-1},\zeta_m)J_t(\zeta_m,\zeta_1)\,\d\zeta_1\cdots\d\zeta_m
\end{equation}
and $J_{t,n}(z,w)$ in \eqref{z23} is analytic in a neighborhood of $(z,w)\in\mathbb{R}\times\mathbb{R}$ which contains $\Gamma_{\alpha}\times\Gamma_{\alpha}$, compare the discussion in Subsection \ref{Stab}. Thus, using the analytic and asymptotic properties of $J_{t,n}(z,w)$, we may consecutively replace $\mathbb{R}$ in \eqref{z27} by $\Gamma_{\alpha}$, obtaining en route
\begin{equation*}
	\tr_{L^2(\mathbb{R})}J_{t,n}^m=\int_{\Gamma_{\alpha}}\cdots\int_{\Gamma_{\alpha}}J_{t,n}(\zeta_1,\zeta_2)\cdot\ldots\cdot J_{t,n}(\zeta_{m-1},\zeta_m)J_t(\zeta_m,\zeta_1)\,\d\zeta_1\cdots\d\zeta_m=\tr_{L^2(\Gamma_{\alpha})}(J_{t,n}^{\circ})^m.
\end{equation*}
This shows that the trace of $(J_{t,n}^{\circ})^m$ on $L^2(\Gamma_{\alpha})$ equals the trace of $J_{t,n}^m$ on $L^2(\mathbb{R})$, in turn \eqref{z26} follows from \eqref{z25}, our above discussion on the trace class property of $J_{t,n}^{\circ}$ and the Plemelj-Smithies formula \cite[Chapter II, Theorem $3.1$]{GGK}. This concludes our proof.
\end{proof}

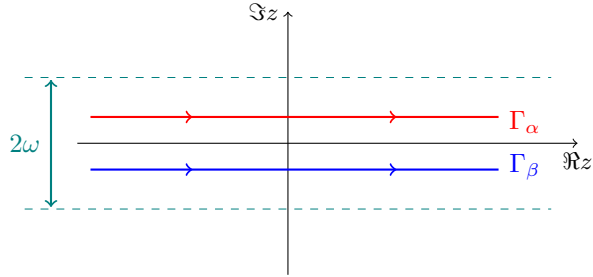
\begin{figure}[tbh]
\begin{tikzpicture}[xscale=0.7,yscale=0.7]
\draw [->] (-4,0) -- (5.5,0) node[below]{{\small $\Re z$}};
\draw [->] (0,-2.5) -- (0,2.5) node[left]{{\small $\Im z$}};

\draw [thick, color=red, decoration={markings, mark=at position 0.25 with {\arrow{>}}}, decoration={markings, mark=at position 0.75 with {\arrow{>}}}, postaction={decorate}] (-3.75,0.5) -- (4,0.5);
\node [color=red] at (4.5,0.4) {$\textcolor{red}{\Gamma_{\alpha}}$};
\draw [thick, color=blue, decoration={markings, mark=at position 0.25 with {\arrow{>}}}, decoration={markings, mark=at position 0.75 with {\arrow{>}}}, postaction={decorate}] (-3.75,-0.5) -- (4,-0.5);
\node [color=blue] at (4.5,-0.45) {$\Gamma_{\beta}$};

\draw [very thin, dashed, color=teal] (-5,-1.25) -- (5,-1.25);
\draw [very thin, dashed, color=teal] (-5,1.25) -- (5,1.25);
\draw [thick, color=teal, decoration={markings, mark=at position 0.05 with {\arrow{<}}}, decoration={markings, mark=at position 1 with {\arrow{>}}}, postaction={decorate}] (-4.5,-1.2) -- (-4.5,1.2);
\node [color=teal] at (-5,0) {$2\omega$};

\end{tikzpicture}
\caption{Our choice for $\Gamma_{\alpha,\beta}$ in Proposition \ref{zprop:2} with $0<\Delta=\textnormal{dist}(\Gamma_{\alpha},\mathbb{R})=\textnormal{dist}(\Gamma_{\beta},\mathbb{R})<\frac{\omega}{2}$.}
\label{figz:3}
\end{figure}

Equipped with the stability identities \eqref{z25},\eqref{z26} we now extend $J_{t,n}^{\circ}$ in the following sense: replace $J_{t,n}^{\circ}$ by
\begin{equation*}
	J_{t,n}^{\textnormal{ext}}:L^2(\Sigma)\rightarrow L^2(\Sigma),\ \ \ (J_{t,n}^{\textnormal{ext}}f)(\xi)=\int_{\Sigma}J_{t,n}^{\textnormal{ext}}(\xi,\eta)f(\eta)\,\d\eta,\ \ \ \ J_{t,n}^{\textnormal{ext}}(\xi,\eta):=J_{t,n}(\xi,\eta)\chi_{\Gamma_{\alpha}}(\xi)\chi_{\Gamma_{\alpha}}(\eta),
\end{equation*}
where the oriented contour $\Sigma\subset\mathbb{C}$ with $\Sigma\supset\Gamma_{\alpha}$, equipped with the arc-length measure, will be determined in Lemma \ref{zlem:6} below. Observe that this extension leaves $D_n(t,\lambda)$ invariant and we have
\begin{equation}\label{z28}
	D_n(t,\lambda)\stackrel[\eqref{z26}]{\eqref{z25}}{=}\det(I-\lambda J_{t,n}^{\textnormal{ext}}\upharpoonright_{L^2(\Sigma)})=\det(I-\lambda A_{t,n}^{\textnormal{ext}}B_{n}^{\textnormal{ext}}\upharpoonright_{L^2(\Sigma)}),\ \ \ (t,\lambda,n)\in\mathbb{R}\times\mathbb{C}\times\mathbb{N},
\end{equation}
where, recall \eqref{z24},
\begin{equation*}
	\,\,\,A_{t,n}^{\textnormal{ext}}:L^2(\Sigma)\rightarrow L^2(\Sigma),\ \ \ \ \ \ (A_{t,n}^{\textnormal{ext}}f)(\xi)=\int_{\Sigma}A_{t,n}^{\textnormal{ext}}(\xi,\eta)f(\eta)\,\d\eta,\ \ \ A_{t,n}^{\textnormal{ext}}(\xi,\eta):=A_{t,n}(\xi,\eta)\chi_{\Gamma_{\alpha}}(\xi)\chi_{\Gamma_{\beta}}(\eta),
\end{equation*}
\begin{equation*}
	B_{n}^{\textnormal{ext}}:L^2(\Sigma)\rightarrow L^2(\Sigma),\ \ \ \ \ \ (B_{n}^{\textnormal{ext}}g)(\eta)=\int_{\Sigma}B_{n}^{\textnormal{ext}}(\eta,\zeta)g(\zeta)\,\d\zeta,\ \ \ B_{n}^{\textnormal{ext}}(\eta,\zeta):=B_{n}(\eta,\zeta)\chi_{\Gamma_{\beta}}(\eta)\chi_{\Gamma_{\alpha}}(\zeta).
\end{equation*}
The benefit of extending $A_{t,n},B_{n}$ and thus $J_{t,n}^{\circ}$ in this fashion stems from the following improvement of the proof working in Lemma \ref{zlem:5}.
\begin{lem}\label{zlem:6} The operators $A_{t,n}^{\textnormal{ext}},B_{n}^{\textnormal{ext}}:L^2(\Sigma)\rightarrow L^2(\Sigma)$ with $\Sigma:=\mathbb{R}\sqcup\Gamma_{\beta}\sqcup\Gamma_{\alpha}$ are trace class on $L^2(\Sigma)$ for every $(t,n)\in\mathbb{R}\times\mathbb{N}$.
\end{lem}
\begin{proof} By residue theorem, cf. \cite[Lemma $3.1$]{BB}, for every $(\gamma,\beta)\in\Gamma_{\alpha}\times\Gamma_{\beta}$,
\begin{equation*}
	-\frac{1}{2\pi\im}\int_{\mathbb{R}}\frac{\d\delta}{(\gamma-\delta)(\delta-\beta)}=\frac{1}{\gamma-\beta},
\end{equation*}
so that with \eqref{z24}
\begin{equation*}
	B_{n}^{\textnormal{ext}}(\beta,\gamma)=-\frac{\im}{(2\pi)^2}\int_{\mathbb{R}}\frac{\e^{-\frac{\im}{2}\psi_n(\beta,0)+\frac{\im}{2}\psi_n(\gamma,0)}}{(\gamma-\delta)(\delta-\beta)}\,\d\delta\,\chi_{\Gamma_{\beta}}(\beta)\chi_{\Gamma_{\alpha}}(\gamma).
\end{equation*}
This allows us to factor $B_{n}^{\textnormal{ext}}$ as $B_{n}^{\textnormal{ext}}=B_{n,1}^{\textnormal{ext}}B_{n,2}^{\textnormal{ext}}$ where $B_{n,j}^{\textnormal{ext}}:L^2(\Sigma)\rightarrow L^2(\Sigma)$ have Hilbert-Schmidt kernels
\begin{equation*}
	B_{n,1}^{\textnormal{ext}}(\beta,\delta):=-\frac{\im}{2\pi}\frac{\e^{-\frac{\im}{2}\psi_n(\beta,0)}}{\beta-\delta}\chi_{\Gamma_{\beta}}(\beta)\chi_{\mathbb{R}}(\delta),\ \ \ \ \ B_{n,2}^{\textnormal{ext}}(\delta,\gamma):=\frac{1}{2\pi}\frac{\e^{\frac{\im}{2}\psi_n(\gamma,0)}}{\delta-\gamma}\chi_{\mathbb{R}}(\delta)\chi_{\Gamma_{\alpha}}(\gamma).
\end{equation*}
Similarly, using again \eqref{z24} and integration by parts,
\begin{equation*}
	A_{t,n}^{\textnormal{ext}}(\alpha,\beta)=-\frac{\im}{2\pi}\int_{\mathbb{R}}\e^{\frac{\im}{2}\psi_n(\alpha,2t)-\frac{\im}{2}\psi_n(\beta,2t)}\e^{\im z(\alpha-\beta)}w(z)\,\d z\,\chi_{\Gamma_{\alpha}}(\alpha)\chi_{\Gamma_{\beta}}(\beta)
\end{equation*}
and thus $A_{t,n}^{\textnormal{ext}}=A_{t,n,1}^{\textnormal{ext}}A_{t,n,2}^{\textnormal{ext}}$ where $A_{t,n,j}^{\textnormal{ext}}:L^2(\Sigma)\rightarrow L^2(\Sigma)$ have Hilbert-Schmidt kernels
\begin{equation*}
	A_{t,n,1}^{\textnormal{ext}}(\alpha,z):=-\frac{\im}{2\pi}\e^{\frac{\im}{2}\psi_n(\alpha,2t)+\im z\alpha}\sqrt{w(z)}\,\chi_{\Gamma_{\alpha}}(\alpha)\chi_{\mathbb{R}}(z),\ \ \ A_{t,n,2}^{\textnormal{ext}}(z,\beta):=\e^{-\frac{\im}{2}\psi_n(\beta,2t)-\im z\beta}\sqrt{w(z)}\,\chi_{\mathbb{R}}(z)\chi_{\Gamma_{\beta}}(\beta).
\end{equation*}
In summary, both $A_{t,n}^{\textnormal{ext}}$ and $B_n^{\textnormal{ext}}$ admit Hilbert-Schmidt factorizations on $L^2(\Sigma)$ and are thus trace class on the same space. This completes our proof.
\end{proof}
Lemma \ref{zlem:6} is useful since, by continuity of $(\alpha,\beta)\mapsto A_{t,n}^{\textnormal{ext}}(\alpha,\beta)$ and $(\beta,\gamma)\mapsto B_{n}^{\textnormal{ext}}(\beta,\gamma)$ on $\Sigma\times\Sigma$, it allows us to compute the operator traces, cf. \cite[Theorem $3.9$]{S},
\begin{equation}\label{z29}
	\tr_{L^2(\Sigma)}A_{t,n}^{\textnormal{ext}}=\int_{\Sigma}A_{t,n}^{\textnormal{ext}}(z,z)\,\d z=0,\ \ \ \ \ \ \ \ \tr_{L^2(\Sigma)}B_{n}^{\textnormal{ext}}=\int_{\Sigma}B_{n}^{\textnormal{ext}}(z,z)\,\d z=0.
\end{equation} 
More importantly, and more generally, the operators $A_{t,n}^{\textnormal{ext}},B_{n}^{\textnormal{ext}}:L^2(\Sigma)\rightarrow L^2(\Sigma)$ are nilpotent,
\begin{equation}\label{z30}
	A_{t,n}^{\textnormal{ext}}A_{t,n}^{\textnormal{ext}}=0\ \ \ \ \textnormal{and}\ \ \ \ \ B_{n}^{\textnormal{ext}}B_{n}^{\textnormal{ext}}=0\ \ \ \textnormal{on}\ \ L^2(\Sigma)
\end{equation}
for every $(t,n)\in\mathbb{R}\times\mathbb{N}$ given that $\Gamma_{\alpha}\cap\Gamma_{\beta}=\emptyset$. This simple observation lies at the heart of the following factorization identity.
\begin{lem}\label{zlem:7} For every $(t,\lambda,n)\in\mathbb{R}\times\mathbb{C}\times\mathbb{N}$, we have on $L^2(\Sigma)$,
\begin{equation*}
	\big(I+\lambda^{\frac{1}{2}}A_{t,n}^{\textnormal{ext}}\big)\big(I-\lambda^{\frac{1}{2}}(A_{t,n}^{\textnormal{ext}}+B_{n}^{\textnormal{ext}})\big)\big(I+\lambda^{\frac{1}{2}}B_{n}^{\textnormal{ext}}\big)=I-\lambda A_{t,n}^{\textnormal{ext}}B_{n}^{\textnormal{ext}}=I-\lambda J_{t,n}^{\textnormal{ext}},
\end{equation*}
with an arbitrary, but throughout fixed, branch for $\lambda^{\frac{1}{2}}$.
\end{lem}
\begin{proof} By straightforward algebra, using \eqref{z30}.
\end{proof}
It is now time to summarize our previous steps.
\begin{prop}\label{zprop:3} For every $(t,\lambda,n)\in\mathbb{R}\times\mathbb{C}\times\mathbb{N}$,
\begin{equation}\label{z31}
	D_n(t,\lambda)=\det(I-\lambda^{\frac{1}{2}}C_{t,n}\upharpoonright_{L^2(\Sigma)}),
\end{equation}
where $C_{t,n}:L^2(\Sigma)\rightarrow L^2(\Sigma)$ with $\Sigma=\mathbb{R}\sqcup\Gamma_{\beta}\sqcup\Gamma_{\alpha}$ is trace class and has kernel
\begin{align}\label{z32}
	(\xi-\eta)C_{t,n}(\xi,\eta):=\frac{1}{2\pi}\int_{\mathbb{R}}\Big(\e^{\frac{\im}{2}(\psi_n(\xi,2t+2z)-\psi_n(\eta,2t+2z))}&\,\chi_{\Gamma_{\alpha}}(\xi)\chi_{\Gamma_{\beta}}(\eta)\nonumber\\
	&\,+\e^{-\frac{\im}{2}(\psi_n(\xi,0)-\psi_n(\eta,0))}\chi_{\Gamma_{\beta}}(\xi)\chi_{\Gamma_{\alpha}}(\eta)\Big)\d\sigma(z).
\end{align}
\end{prop}
\begin{proof} The operator $C_{t,n}=A_{t,n}^{\textnormal{ext}}+B_{n}^{\textnormal{ext}}$ is trace class on $L^2(\Sigma)$ as sum of two trace class operators on the same space, recall Lemma \ref{zlem:6}. Moreover, by the Plemelj-Smithies formula and \eqref{z29},\eqref{z30} we have
\begin{equation*}
	\det(I+\lambda^{\frac{1}{2}}A_{t,n}^{\textnormal{ext}}\upharpoonright_{L^2(\Sigma)})=\exp\left[-\sum_{k=1}^{\infty}(-1)^k\frac{\lambda^{\frac{k}{2}}}{k}\tr_{L^2(\Sigma)}(A_{t,n}^{\textnormal{ext}})^k\right]=1,
\end{equation*}
and similarly $\det(I+\lambda^{\frac{1}{2}}B_{n}^{\textnormal{ext}}\upharpoonright_{L^2(\Sigma)})=1$. Hence, by Lemma \ref{zlem:7}, the factorization identity \cite[$(3.10)$]{S} and \eqref{z28} we find \eqref{z31} with the indicated kernel \eqref{z32}. This completes our proof.
\end{proof}

After having arrived at \eqref{z31} with the particular kernel structure \eqref{z32} we now proceed as in \cite[Section $9.1,9.2$]{B}, throughout relying on the abbreviations summarized in Appendix \ref{appC}.
\begin{definition}\label{zdef:1} Let $M_i(\zeta)\otimes K_j(\zeta)\in\mathcal{I}(\mathcal{H}_1),i,j=1,2$ denote the $\Sigma\ni\zeta$-parametric family of rank one integral operators with kernels
\begin{equation*}
	\big(M_i(\zeta)\otimes K_j(\zeta)\big)(x,y):=m_i(\zeta|x)k_j(\zeta|y),\ \ \ \ \ x,y\in\mathbb{R},
\end{equation*}
defined in terms of the $\Sigma\ni\zeta$-parametric family of functions\footnote{These functions also depend on $(t,n)\in\mathbb{R}\times\mathbb{N}$, but we do not highlight this in our notation.}
\begin{equation*}
	k_1(\zeta|y):=\frac{1}{2\pi}\e^{\frac{\im}{2}\psi_n(\zeta,2t+2y)}\chi_{\Gamma_{\alpha}}(\zeta),\ \ k_2(\zeta|y):=\frac{1}{2\pi}\e^{-\frac{\im}{2}\psi_n(\zeta,0)}\chi_{\Gamma_{\beta}}(\zeta),\ \ m_1(\zeta|x):=\e^{-\frac{\im}{2}\psi_n(\zeta,2t+2x)}\chi_{\Gamma_{\beta}}(\zeta),
\end{equation*}
\begin{equation}\label{z33}
	m_2(\zeta|x):=\e^{\frac{\im}{2}\psi_n(\zeta,0)}\chi_{\Gamma_{\alpha}}(\zeta).
\end{equation}
Equivalently, $M_i(\zeta)$ are the operators on $\mathcal{H}_1$ which multiply by the functions $m_i(\zeta|x)$ and $K_j(\zeta)$ are the integral operators on $\mathcal{H}_1$ with kernel $k_j(\zeta|y)$.
\end{definition}
\begin{remark}
Observe that, in this way, we can write the kernel of the operator $C_{t,n}$ in ``integrable'' form. Indeed, \eqref{z32} reads
\begin{equation}\label{z33bis}
	(\xi - \eta)C_{t,n}(\xi,\eta) = \int_{\mathbb R} \Big( k_1(\xi|z)m_1(\eta|z) + k_2(\xi|z)m_2(\eta|z) \Big) \d\sigma(z),
\end{equation}
see also equation \eqref{i27}.
\end{remark}
Now consider the below $\mathcal{I}(\mathcal{H}_2)$-valued RHP, the central operator-valued RHP of this text.
\begin{problem}\label{zmaster} Given $(t,\lambda,n)\in\mathbb{R}\times\overline{\mathbb{D}_1(0)}\times\mathbb{N}$, determine ${\bf X}(\zeta)={\bf X}(\zeta;t,\lambda,n)\in\mathcal{I}(\mathcal{H}_2)$ such that
\begin{enumerate}
	\item[(1)] ${\bf X}(\zeta)=\mathbb{I}_2+{\bf X}_0(\zeta)$ and ${\bf X}_0(\zeta)\in\mathcal{I}(\mathcal{H}_2)$ with kernel ${\bf X}_0(\zeta|x,y)$ is analytic in $\mathbb{C}\setminus\Sigma$.
	\item[(2)] ${\bf X}(\zeta)$ admits continuous boundary values ${\bf X}_{\pm}(\zeta)\in\mathcal{I}(\mathcal{H}_2)$ on $\Sigma$, oriented as shown in Figure \ref{figz:3}, which satisfy ${\bf X}_+(\zeta)={\bf X}_-(\zeta){\bf G}(\zeta)$ with
	\begin{equation}\label{z34}
		{\bf G}(\zeta)=\mathbb{I}_2+2\pi\im\lambda^{\frac{1}{2}}\begin{bmatrix}M_1(\zeta)\otimes K_1(\zeta) & M_1(\zeta)\otimes K_2(\zeta)\smallskip\\
		M_2(\zeta)\otimes K_1(\zeta) & M_2(\zeta)\otimes K_2(\zeta)\end{bmatrix},\ \ \ \ \zeta\in\Sigma,
	\end{equation}
	using the same branch for $\lambda^{\frac{1}{2}}$ as in Lemma \ref{zlem:7}.
	\item[(3)] There exists $c=c(n,t)>0$ such that for $\zeta\in\mathbb{C}\setminus\Sigma$,
	\begin{equation}\label{z35}
		\|{\bf X}_0(\zeta|x,y)\|\leq\frac{c\sqrt{|\lambda|}}{1+|\zeta|}\,\Delta^{-\frac{1}{4n}}\e^{-\frac{(-1)^n\Delta}{2(2n+1)}\Delta^{2n}}\e^{\Delta(|x|+|y|+|t|)},\ \ \Delta:=\textnormal{dist}(\Gamma_{\alpha},\mathbb{R})=\textnormal{dist}(\Gamma_{\beta},\mathbb{R})>0,
	\end{equation}
	uniformly in $(x,y)\in\mathbb{R}^2$ and $\lambda\in\overline{\mathbb{D}_1(0)}$.
\end{enumerate}
\end{problem}
\begin{rem} The real line does not enter explicitly in property $(2)$ of RHP \ref{zmaster}, compare \eqref{z33}. Nevertheless we require $\Sigma=\mathbb{R}\sqcup\Gamma_{\alpha}\sqcup\Gamma_{\beta}$ in the proof of Theorem \ref{z:theo1} below: since $C_{t,n}$ is trace class on $L^2(\Sigma)$, the non-vanishing of $D_n(t,\lambda)$ for every $(t,\lambda,n)\in\mathbb{R}\times\overline{\mathbb{D}_1(0)}\times\mathbb{N}$ by Lemma \ref{zlem:3} yields invertibility of $I-\lambda^{\frac{1}{2}}C_{t,n}$ on $L^2(\Sigma)$ by \eqref{z31} in the same parameter range, cf. \cite[Theorem $3.5$(b)]{S}. This observation turns out to be central to the solvability of RHP \ref{zmaster}.
\end{rem}
In our first observation below we record that properties $(1)-(3)$ in RHP \ref{zmaster} define ${\bf X}(\zeta)$ unambiguously.
\begin{lem}\label{z:lem7} The $\mathcal{I}(\mathcal{H}_2)$-valued RHP \ref{zmaster}, if solvable, is uniquely solvable.
\end{lem}
\begin{proof} Let ${\bf X}(\zeta)$ denote any solution of RHP \ref{zmaster}. By condition (1), estimate \eqref{z35} and our discussion in Appendix \ref{appA}, the Fredholm determinant
\begin{equation}\label{z36}
	d(\zeta):=\det(\mathbb{I}_2+{\bf X}_0(\zeta)\upharpoonright_{\mathcal{H}_2}),\ \ \zeta\in\mathbb{C}\setminus\Sigma
\end{equation}
is well-defined in the indicated $\zeta$-domain. Also, by Morera's and Fubini's theorem, the imposed analyticity of ${\bf X}_0(\zeta)\in\mathcal{I}(\mathcal{H}_2)$ away from $\Sigma$ yields analyticity of $d(\zeta)$ for $\zeta\in\mathbb{C}\setminus\Sigma$. Furthermore, by properties $(2),(3)$ and the dominated convergence theorem, the non-tangential boundary values $d_{\pm}(\zeta),\zeta\in\Sigma$ exist and satisfy
\begin{equation*}
	d_{\pm}(\zeta)=\det({\bf X}_{\pm}(\zeta)\upharpoonright_{\mathcal{H}_2}),\ \ \ \ \zeta\in\Sigma.
\end{equation*}
However, from condition $(2)$ and a simple estimation of \eqref{z33}, we obtain ${\bf G}(\zeta)=\mathbb{I}_2+{\bf G}_0(\zeta)$ where ${\bf G}_0(\zeta)\in\mathcal{I}(\mathcal{H}_2)$ is trace class on $\mathcal{H}_2$ and satisfies
\begin{equation*}
	\|{\bf G}_0(\zeta|x,y)\|\leq c\sqrt{|\lambda|}\,\e^{-\frac{(-1)^n\Delta}{2n+1}\Delta^{2n}}\e^{\Delta(|x|+|y|+|t|)},\ \ \ c=c(n)>0,
\end{equation*}
uniformly in $(\zeta,x,y)\in\Sigma\times\mathbb{R}^2$ and $(t,\lambda)\in\mathbb{R}\times\overline{\mathbb{D}_1(0)}$. Thus, $\det(\mathbb{I}_2+{\bf G}_0(\zeta)\upharpoonright_{\mathcal{H}_2})$ exists for $\zeta\in\Sigma$ by Hadamard's inequality and since
\begin{equation*}
	\tr_{\mathcal{H}_2}{\bf G}_0(\zeta)=0,\ \ \ \ \ \ \big({\bf G}_0(\zeta)\big)^2={\bf 0},\ \ \ \ \zeta\in\Sigma,
\end{equation*}
the Plemelj-Smithies formula \cite[Chapter II, Theorem $3.1$]{GGK} yields that $\det(\mathbb{I}_2+{\bf G}_0(\zeta)\upharpoonright_{\mathcal{H}_2})=1$ for all $\zeta\in\Sigma$. In turn, using the multiplicativity of Fredholm determinants, \eqref{z36} thus satisfies (by property (2) and our above discussion) the scalar jump condition
\begin{equation*}
	d_+(\zeta)=d_-(\zeta),\ \ \zeta\in\Sigma.
\end{equation*}
Since also $d(\zeta)\rightarrow 1$ as $\zeta\rightarrow\infty$ by property (3) we can conclude (using the continuity of the boundary values $d_{\pm}(\zeta)$ on $\Sigma$) that $d(\zeta)$ must be an entire scalar-valued function normalized to unity at infinity. Hence $d(\zeta)\equiv 1$ for all $\zeta\in\mathbb{C}$. Consequently, any solution ${\bf X}(\zeta)\in\mathcal{I}(\mathcal{H}_2)$ of RHP \ref{zmaster} is invertible for all $\zeta\in\mathbb{C}\setminus\Sigma$ and so are its continuous boundary values ${\bf X}_{\pm}(\zeta),\zeta\in\Sigma$. Moving ahead, we now consider two solutions ${\bf X}_1(\zeta),{\bf X}_2(\zeta)$ of RHP \ref{zmaster} and introduce
\begin{equation*}
	\mathcal{I}(\mathcal{H}_2)\ni{\bf Y}(\zeta):={\bf X}_1(\zeta)\big({\bf X}_2(\zeta)\big)^{-1}\equiv\mathbb{I}_2+{\bf Y}_0(\zeta),\ \ \ \ \zeta\in\mathbb{C}\setminus\Sigma.
\end{equation*}
By RHP \ref{zmaster}, ${\bf Y}(\zeta)$ is analytic in $\zeta\in\mathbb{C}\setminus\Sigma$, attains continuous boundary values ${\bf Y}_{\pm}(\zeta)$ on $\Sigma$ and satisfies ${\bf Y}_+(\zeta)={\bf Y}_-(\zeta)$ on $\Sigma$. Thus, $\zeta\mapsto{\bf Y}_0(\zeta|x,y)$ is entire $(\sigma\otimes\sigma)$-almost everywhere and since ${\bf Y}_0(\zeta|x,y)\rightarrow{\bf 0}$ as $\zeta\rightarrow\infty$ also $(\sigma\otimes\sigma)$-almost everywhere, we conclude by Liouville's theorem that ${\bf Y}(\zeta)=\mathbb{I}_2$, i.e. ${\bf X}_1(\zeta)={\bf X}_2(\zeta)$ and uniqueness is therefore established. This concludes the proof.
\end{proof}
Complementing Lemma \ref{z:lem7} we now show that RHP \ref{zmaster} is solvable: First, using the chain of determinant equalities \eqref{z25},\eqref{z26},\eqref{z28},\eqref{z31}, 
\begin{equation}\label{z37}
	\forall\,(t,\lambda,n)\in\mathbb{R}\times\mathbb{C}\times\mathbb{N}:\ \ \ D_n(t,\lambda)=\det(I-\lambda K_{t,n}\upharpoonright_{L^2(\mathbb{R}_+)})=\det\big(I-\lambda^{\frac{1}{2}}C_{t,n}\upharpoonright_{L^2(\Sigma)}\big),
\end{equation}
and $I-\lambda K_{t,n}$ is invertible on $L^2(\mathbb{R}_+)$ by Lemma \ref{zlem:3} if $\lambda\in\overline{\mathbb{D}_1(0)}$. Hence $D_n(t,\lambda)\neq 0$ for all $(t,\lambda,n)\in\mathbb{R}\times\overline{\mathbb{D}_1(0)}\times\mathbb{N}$ which yields invertibility of $I-\lambda^{\frac{1}{2}}C_{t,n}$ on $L^2(\Sigma)$ in the same parameter range. This simple observation lies at the heart of the below solvability theorem.
\begin{theo}\label{z:theo1} For every $(t,\lambda,n)\in\mathbb{R}\times\overline{\mathbb{D}_1(0)}\times\mathbb{N}$ consider the integral operator
\begin{equation}\label{z38}
	{\bf X}(\zeta)=\mathbb{I}_2+\lambda^{\frac{1}{2}}\int_{\Sigma}\begin{bmatrix}N_1(\eta)\otimes K_1(\eta) & N_1(\eta)\otimes K_2(\eta)\smallskip\\
	N_2(\eta)\otimes K_1(\eta) & N_2(\eta)\otimes K_2(\eta)\end{bmatrix}\frac{\d\eta}{\eta-\zeta},\ \ \ \ \zeta\in\mathbb{C}\setminus\Sigma,
\end{equation}
where $N_i(\eta)$ are the operators on $\mathcal{H}_1$ which multiply by the functions $n_i(\eta|x)$ determined via
\begin{equation}\label{z39}
	\big(I-\lambda^{\frac{1}{2}}C_{t,n}^{\ast}\upharpoonright_{L^2(\Sigma)}\big)n_i(\cdot|x)=m_i(\cdot|x),\ \ \ i=1,2,
\end{equation}
with $x\in\mathbb{R}$ and the real adjoint $C_{t,n}^{\ast}$ of $C_{t,n}$. Then \eqref{z38} solves RHP \ref{zmaster}.
\end{theo}
\begin{proof} Since \eqref{z39} is uniquely solvable, compare \eqref{z37}, we conclude from the boundedness of the resolvent,
\begin{equation}\label{z40}
	\|n_i(\cdot|x)\|_{L^2(\Sigma)}\leq c\|m_i(\cdot|x)\|_{L^2(\Sigma)},\ \ \ c=c(n,t)>0,\ \ \ i=1,2,
\end{equation}
uniformly in $x\in\mathbb{R}$ and $\lambda\in\overline{\mathbb{D}_1(0)}$. Now consider the right hand side of \eqref{z38} and note that its nontrivial kernel is built out of the functions
\begin{equation*}
	X_0^{ij}(\zeta|x,y)=\lambda^{\frac{1}{2}}\int_{\Sigma}n_i(\eta|x)k_j(\eta|y)\frac{\d\eta}{\eta-\zeta},\ \ \ \zeta\notin\Sigma,\ \ \ (x,y)\in\mathbb{R}^2.
\end{equation*}
But \eqref{z33},\eqref{z40}, Figure \ref{figz:3} and Cauchy-Schwarz yield
\begin{equation}\label{z41}
	\big|X_0^{ij}(\zeta|x,y)\big|\leq\frac{c\sqrt{|\lambda|}}{\textnormal{dist}(\zeta,\Sigma)}\Delta^{-\frac{1}{4n}}\e^{-\frac{(-1)^n\Delta}{2(2n+1)}\Delta^{2n}}\e^{\Delta(|x|+|y|+|t|)},\ \ \ c=c(n,t)>0,\ \ i,j=1,2,
\end{equation}
uniformly in $(x,y)\in\mathbb{R}^2$ and $\lambda\in\overline{\mathbb{D}_1(0)}$. Hence, $(x,y)\mapsto X_0(\zeta|x,y)$ is in $L^2(\mathbb{R}^2,\d\sigma\otimes\d\sigma;\mathbb{C}^{2\times 2})$ for every $\zeta\in\mathbb{C}\setminus\Sigma$. Moreover, using the regularity properties of the resolvent, we deduce from \eqref{z39} that $\eta\mapsto n_i(\eta|x)k_j(\eta|y)$ are H\"older continuous on $\Sigma$ for every $(x,y)\in\mathbb{R}^2$ and $i,j\in\{1,2\}$. Thus, by the Plemelj-Sokhotski theorem, $\zeta\mapsto X_0(\zeta|x,y)$ is analytic in $\zeta\in\mathbb{C}\setminus\Sigma$ for every $(x,y)\in\mathbb{R}^2$, i.e. the right hand side of \eqref{z38} all together analytic in $\mathbb{C}\setminus\Sigma$ in the sense of Definition \ref{c2}. This establishes property (1) of RHP \ref{zmaster} for \eqref{z38}. Moving ahead, estimate \eqref{z35}, i.e. property (3) for \eqref{z38}, follows at once from \eqref{z41}. Finally, by the H\"older continuity of $\eta\mapsto n_i(\eta|x)k_j(\eta|y)$, the boundary values ${\bf X}_{\pm}(\zeta),\zeta\in\Sigma$ not only exist by the Plemelj-Sokhotski theorem but are also H\"older continuous by the Plemelj-Privalov theorem \cite[Chapter $2$, $\S$ 19]{M} and are again in $\mathcal{I}(\mathcal{H}_2)$, using en route that $\Sigma$ is a union of smooth contours. It now remains to verify that the right hand side of \eqref{z38} satisfies the jump condition \eqref{z34} in property (2) of RHP \ref{zmaster}. To this end we first compute from \eqref{z38},
\begin{equation}\label{z42}
	{\bf X}_+(\zeta)-{\bf X}_-(\zeta)=2\pi\im\lambda^{\frac{1}{2}}\begin{bmatrix}N_1(\zeta)\otimes K_1(\zeta) & N_1(\zeta)\otimes K_2(\zeta)\smallskip\\
	N_2(\zeta) \otimes K_1(\zeta) & N_2(\zeta)\otimes K_2(\zeta)\end{bmatrix},\ \ \ \ \zeta\in\Sigma.
\end{equation}
However, we also have from \eqref{z34} and \eqref{z38},
\begin{align}
	{\bf X}_-(\zeta){\bf G}(\zeta)=&\,\,{\bf X}_-(\zeta)\left\{\mathbb{I}_2+2\pi\im\lambda^{\frac{1}{2}}\begin{bmatrix}M_1(\zeta)\otimes K_1(\zeta) & M_1(\zeta)\otimes K_2(\zeta)\smallskip\\
	M_2(\zeta)\otimes K_1(\zeta) & M_2(\zeta)\otimes K_2(\zeta)\end{bmatrix}\right\}\nonumber\\
	=&\,\,{\bf X}_-(\zeta)+2\pi\im\lambda^{\frac{1}{2}}\left\{\mathbb{I}_2+\lambda^{\frac{1}{2}}\int_{\Sigma}\begin{bmatrix}N_1(\eta)\otimes K_1(\eta) & N_1(\eta)\otimes K_2(\eta)\smallskip\\
	N_2(\eta)\otimes K_1(\eta) & N_2(\eta)\otimes K_2(\eta)\end{bmatrix}\frac{\d\eta}{\eta-\zeta_-}\right\}\nonumber\\
	&\hspace{1cm}\circ\begin{bmatrix}M_1(\zeta)\otimes K_1(\zeta) & M_1(\zeta)\otimes K_2(\zeta)\smallskip\\
	M_2(\zeta)\otimes K_1(\zeta) & M_2(\zeta)\otimes K_2(\zeta)\end{bmatrix},\ \ \ \zeta\in\Sigma.\label{z43}
\end{align}
Since, by general theory of rank one integral operators, \eqref{z32} and \eqref{z33bis},
\begin{align*}
	&\begin{bmatrix}N_1(\eta)\otimes K_1(\eta) & N_1(\eta)\otimes K_2(\eta)\smallskip\\
	N_2(\eta)\otimes K_1(\eta) & N_2(\eta)\otimes K_2(\eta)\end{bmatrix}\begin{bmatrix}M_1(\zeta)\otimes K_1(\zeta) & M_1(\zeta)\otimes K_2(\zeta)\smallskip\\
	M_2(\zeta)\otimes K_1(\zeta) & M_2(\zeta)\otimes K_2(\zeta)\end{bmatrix}\\
	&\hspace{1cm}=(\eta-\zeta)C_{t,n}(\eta,\zeta)\begin{bmatrix}N_1(\eta)\otimes K_1(\zeta) & N_1(\eta)\otimes K_2(\zeta)\smallskip\\
	N_2(\eta)\otimes K_1(\zeta) & N_2(\eta)\otimes K_2(\zeta)\end{bmatrix},\ \ \ \ (\eta,\zeta)\in\Sigma\times\Sigma,
\end{align*}
identity \eqref{z43} transforms to
\begin{align}
	{\bf X}_-(\zeta){\bf G}(\zeta)=&\,\,{\bf X}_-(\zeta)+2\pi\im\lambda^{\frac{1}{2}}\begin{bmatrix}M_1(\zeta)\otimes K_1(\zeta) & M_1(\zeta)\otimes K_2(\zeta)\smallskip\\
	M_2(\zeta)\otimes K_1(\zeta) & M_2(\zeta)\otimes K_2(\zeta)\end{bmatrix}\nonumber\\
	&\,\,\hspace{1cm}+2\pi\im\lambda\int_{\Sigma}C_{t,n}(\eta,\zeta)\begin{bmatrix}N_1(\eta)\otimes K_1(\zeta) & N_1(\eta)\otimes K_2(\zeta)\smallskip\\
	N_2(\eta)\otimes K_1(\zeta) & N_2(\eta)\otimes K_2(\zeta)\end{bmatrix}\,\d\eta,\ \ \ \zeta\in\Sigma.\label{z44}
\end{align}
Keeping in mind \eqref{z39}, equivalently the operator-valued integral equation
\begin{equation*}
	N_i(\zeta)-\lambda^{\frac{1}{2}}\int_{\Sigma}C_{t,n}(\eta,\zeta)N_i(\eta)\,\d\eta=M_i(\zeta),\ \ \ \ \zeta\in\Sigma,
\end{equation*}
we substitute this equation into \eqref{z44} and simplify
\begin{equation*}
	{\bf X}_-(\zeta){\bf G}(\zeta)={\bf X}_-(\zeta)+2\pi\im\lambda^{\frac{1}{2}}\begin{bmatrix}N_1(\zeta)\otimes K_1(\zeta) & N_1(\zeta)\otimes K_2(\zeta)\smallskip\\
	N_2(\zeta) \otimes K_1(\zeta) & N_2(\zeta)\otimes K_2(\zeta)\end{bmatrix}\stackrel{\eqref{z42}}{=}{\bf X}_+(\zeta),\ \ \ \zeta\in\Sigma.
\end{equation*}
In summary, property (2) is also satisfied by the right hand side of \eqref{z38} and thus ${\bf X}(\zeta)$ as defined in the same equation solves RHP \ref{zmaster} for every $(t,\lambda,n)\in\mathbb{R}\times\overline{\mathbb{D}_1(0)}\times\mathbb{N}$. This completes our proof.
\end{proof}
By Lemma \ref{z:lem7} and Theorem \ref{z:theo1}, the $\mathcal{I}(\mathcal{H}_2)$-valued RHP \ref{zmaster} is uniquely solvable for every $(t,\lambda,n)\in\mathbb{R}\times\overline{\mathbb{D}_1(0)}\times\mathbb{N}$. This result has several consequences which we summarize below and apply later on.
\begin{cor}\label{z:cor1} Let ${\bf X}(\zeta)\in\mathcal{I}(\mathcal{H}_2)$ denote the unique solution \eqref{z38} of RHP \ref{zmaster}. Then ${\bf X}(\zeta)$ is invertible on $\mathcal{H}_2$ for every $(t,\lambda,n)\in\mathbb{R}\times\overline{\mathbb{D}_1(0)}\times\mathbb{N}$ and we have
\begin{equation}\label{z45}
	\big({\bf X}(\zeta)\big)^{-1}=\mathbb{I}_2-\lambda^{\frac{1}{2}}\int_{\Sigma}\begin{bmatrix}M_1(\eta)\otimes L_1(\eta) & M_1(\eta)\otimes L_2(\eta)\smallskip\\
	M_2(\eta)\otimes L_1(\eta) & M_2(\eta)\otimes L_2(\eta)\end{bmatrix}\frac{\d\eta}{\eta-\zeta},\ \ \ \zeta\in\mathbb{C}\setminus\Sigma,
\end{equation}
where $L_i(\eta)$ are the integral operators on $\mathcal{H}_1$ with kernel $\ell_i(\eta|y)$ determined from the equation
\begin{equation}\label{z46}
	(I-\lambda^{\frac{1}{2}}C_{t,n}\upharpoonright_{L^2(\Sigma)})\ell_i(\cdot|y)=k_i(\cdot|y),\ \ \ i=1,2,
\end{equation}
with $y\in\mathbb{R}$. Moreover, for every $\zeta\in\Sigma$, independently of the choice of boundary values for ${\bf X}(\zeta)$,
\begin{equation}\label{z47}
	{\bf N}(\zeta)={\bf X}(\zeta){\bf M}(\zeta),\ \ \ \ \ \ \ \ \ {\bf L}(\zeta)={\bf K}(\zeta)\big({\bf X}(\zeta)\big)^{-1},
\end{equation}
where we introduce the vector-valued operators
\begin{equation*}
	{\bf N}(\zeta):=\big[N_1(\zeta),N_2(\zeta)\big]^{\top}\!,\ \ {\bf M}(\zeta):=\big[M_1(\zeta),M_2(\zeta)\big]^{\top}\!,\ \ 
	{\bf L}(\zeta):=\big[L_1(\zeta),L_2(\zeta)\big],\ \ {\bf K}(\zeta):=\big[K_1(\zeta),K_2(\zeta)\big].
\end{equation*}
\end{cor}
\begin{proof} We already know that each solution of RHP \ref{zmaster} is invertible on $\mathcal{H}_2$, see the proof workings of Lemma \ref{z:lem7}. Hence, with \eqref{z38} and \eqref{z45} (abbreviating its right hand side as ${\bf Y}(\zeta)$) we compute
\begin{equation}\label{z48}
	{\bf X}(\zeta){\bf Y}(\zeta)=\mathbb{I}_2+\lambda^{\frac{1}{2}}\int_{\Sigma}\big(\overline{{\bf X}}(\eta)-\overline{{\bf Y}}(\eta)\big)\frac{\d\eta}{\eta-\zeta}-\lambda\int_{\Sigma}\int_{\Sigma}\overline{{\bf X}}(\eta_1)\overline{{\bf Y}}(\eta_2)\frac{\d\eta_1}{\eta_1-\zeta}\frac{\d\eta_2}{\eta_2-\zeta},\ \ \ \zeta\notin\Sigma,
\end{equation}
where $\overline{{\bf X}}$ and $\overline{{\bf Y}}$ denote the finite rank integrands in \eqref{z38} and \eqref{z45}. Since
\begin{equation*}
	\overline{{\bf X}}(\eta_1)\overline{{\bf Y}}(\eta_2)\stackrel{\eqref{z32}}{=}(\eta_1-\eta_2)C_{t,n}(\eta_1,\eta_2)\begin{bmatrix}N_1(\eta_1)\otimes L_1(\eta_2) & N_1(\eta_1)\otimes L_2(\eta_2)\smallskip\\
	N_2(\eta_1)\otimes L_1(\eta_2) & N_2(\eta_1)\otimes L_2(\eta_2)\end{bmatrix},
\end{equation*}
we can use partial fractions in the iterated integral in \eqref{z48} and both equations \eqref{z39},\eqref{z46}. The result equals
\begin{equation*}
	\lambda^{\frac{1}{2}}\int_{\Sigma}\big(\overline{{\bf X}}(\eta)-\overline{{\bf Y}}(\eta)\big)\frac{\d\eta}{\eta-\zeta}-\lambda\int_{\Sigma}\int_{\Sigma}\overline{{\bf X}}(\eta_1)\overline{{\bf Y}}(\eta_2)\frac{\d\eta_1}{\eta_1-\zeta}\frac{\d\eta_2}{\eta_2-\zeta}={\bf 0},\ \ \zeta\notin\Sigma,
\end{equation*}
and therefore yields ${\bf X}(\zeta){\bf Y}(\zeta)=\mathbb{I}_2$ for $\zeta\in\mathbb{C}\setminus\Sigma$. In other words, ${\bf Y}(\zeta)$ is a right-sided inverse for ${\bf X}(\zeta)$ which is invertible on $\mathcal{H}_2$, see the proof of Lemma \ref{zlem:7}. This only happens when \eqref{z45} holds for all $\zeta\notin\Sigma$, as claimed. Next we revisit our proof of Theorem \ref{z:theo1}, precisely the last identity in it,
\begin{equation}\label{z49}
	{\bf X}_-(\zeta){\bf G}(\zeta)={\bf X}_-(\zeta)+2\pi\im\lambda^{\frac{1}{2}}\begin{bmatrix}N_1(\zeta)\otimes K_1(\zeta) & N_1(\zeta)\otimes K_2(\zeta)\smallskip\\
	N_2(\zeta)\otimes K_1(\zeta) & N_2(\zeta)\otimes K_2(\zeta)\end{bmatrix},\ \ \zeta\in\Sigma.
\end{equation}
Since
\begin{equation*}
	\big({\bf G}(\zeta)\big)^{-1}=\mathbb{I}_2-2\pi\im\lambda^{\frac{1}{2}}\begin{bmatrix}M_1(\zeta)\otimes K_1(\zeta) & M_1(\zeta)\otimes K_2(\zeta)\smallskip\\
	M_2(\zeta)\otimes K_1(\zeta) & M_2(\zeta)\otimes K_2(\zeta)\end{bmatrix},\ \ \ \zeta\in\Sigma,
\end{equation*}
we then repeat the steps leading to \eqref{z43},\eqref{z44} and conclude likewise
\begin{equation}\label{z50}
	{\bf X}_+(\zeta)\big({\bf G}(\zeta)\big)^{-1}={\bf X}_+(\zeta)-2\pi\im\lambda^{\frac{1}{2}}\begin{bmatrix}N_1(\zeta)\otimes K_1(\zeta) & N_1(\zeta)\otimes K_2(\zeta)\smallskip\\
	N_2(\zeta)\otimes K_1(\zeta) & N_2(\zeta)\otimes K_2(\zeta)\end{bmatrix},\ \ \ \zeta\in\Sigma.
\end{equation}
Hence, combining \eqref{z49} and \eqref{z50} with the explicit formula for ${\bf G}(\zeta)$, see \eqref{z34}, we find
\begin{equation*}
	{\bf X}_{\pm}(\zeta)\begin{bmatrix}M_1(\zeta)\otimes K_1(\zeta) & M_1(\zeta)\otimes K_2(\zeta)\smallskip\\
	M_2(\zeta)\otimes K_1(\zeta) & M_2(\zeta)\otimes K_2(\zeta)\end{bmatrix}=\begin{bmatrix}N_1(\zeta)\otimes K_1(\zeta) & N_1(\zeta)\otimes K_2(\zeta)\smallskip\\
	N_2(\zeta)\otimes K_1(\zeta) & N_2(\zeta)\otimes K_2(\zeta)\end{bmatrix},\ \ \ \zeta\in\Sigma,
\end{equation*}
and thus the first equality in \eqref{z47}. The second equality follows by similar logic: from \eqref{z45}, for $\zeta\in\Sigma$,
\begin{equation*}
	\big({\bf X}(\zeta)\big)^{-1}_+-\big({\bf X}(\zeta)\big)^{-1}_-=-2\pi\im\lambda^{\frac{1}{2}}\begin{bmatrix}M_1(\zeta)\otimes L_1(\zeta) & M_1(\zeta)\otimes L_2(\zeta)\smallskip\\
	M_2(\zeta)\otimes L_1(\zeta) & M_2(\zeta)\otimes L_2(\zeta)\end{bmatrix},
\end{equation*}
and
\begin{eqnarray}
	\big({\bf G}(\zeta)\big)^{-1}\big({\bf X}(\zeta)\big)^{-1}_-\!\!\!\!&=&\!\!\!\!\left\{\mathbb{I}_2-2\pi\im\lambda^{\frac{1}{2}}\begin{bmatrix}M_1(\zeta)\otimes K_1(\zeta) & M_1(\zeta)\otimes K_2(\zeta)\smallskip\\
	M_2(\zeta)\otimes K_1(\zeta) & M_2(\zeta)\otimes K_2(\zeta)\end{bmatrix}\right\}\big({\bf X}(\zeta)\big)^{-1}_-\nonumber\\
	\!\!\!\!&\stackrel{\eqref{z45}}{=}&\!\!\!\!\big({\bf X}(\zeta)\big)^{-1}_--2\pi\im\lambda^{\frac{1}{2}}\begin{bmatrix}M_1(\zeta)\otimes K_1(\zeta) & M_1(\zeta)\otimes K_2(\zeta)\smallskip\\
	M_2(\zeta)\otimes K_1(\zeta) & M_2(\zeta)\otimes K_2(\zeta)\end{bmatrix}\nonumber\\
	&&\hspace{1.5cm}\circ\left\{\mathbb{I}_2-\lambda^{\frac{1}{2}}\int_{\Sigma}\begin{bmatrix}M_1(\eta)\otimes L_1(\eta) & M_1(\eta)\otimes L_2(\eta)\smallskip\\
	M_2(\eta)\otimes L_1(\eta) & M_2(\eta)\otimes L_2(\eta)\end{bmatrix}\frac{\d\eta}{\eta-\zeta_-}\right\}.\label{z51}
\end{eqnarray}
Note however, with \eqref{z32},
\begin{align*}
	&\begin{bmatrix}M_1(\zeta)\otimes K_1(\zeta) & M_1(\zeta)\otimes K_2(\zeta)\smallskip\\
	M_2(\zeta)\otimes K_1(\zeta) & M_2(\zeta)\otimes K_2(\zeta)\end{bmatrix}\begin{bmatrix}M_1(\eta)\otimes L_1(\eta) & M_1(\eta)\otimes L_2(\eta)\smallskip\\
	M_2(\eta)\otimes L_1(\eta) & M_2(\eta)\otimes L_2(\eta)\end{bmatrix}\\
	&\hspace{1cm}=(\zeta-\eta)C_{t,n}(\zeta,\eta)\begin{bmatrix}M_1(\zeta)\otimes L_1(\eta) & M_1(\zeta)\otimes L_2(\eta)\smallskip\\
	M_2(\zeta)\otimes L_1(\eta) & M_2(\zeta)\otimes L_2(\eta)\end{bmatrix},\ \ \ (\zeta,\eta)\in\Sigma\times\Sigma,
\end{align*}
so that \eqref{z51} simplifies to
\begin{equation}\label{z52}
	\big({\bf G}(\zeta)\big)^{-1}\big({\bf X}(\zeta)\big)^{-1}_-\stackrel{\eqref{z46}}{=}\big({\bf X}(\zeta)\big)^{-1}_--2\pi\im\lambda^{\frac{1}{2}}\begin{bmatrix}M_1(\zeta)\otimes L_1(\zeta) & M_1(\zeta)\otimes L_2(\zeta)\smallskip\\
	M_2(\zeta)\otimes L_1(\zeta) & M_2(\zeta)\otimes L_2(\zeta)\end{bmatrix},\ \ \ \zeta\in\Sigma.
\end{equation}
On the other hand we also establish the chain of equalities
\begin{eqnarray}
	{\bf G}(\zeta)\big({\bf X}(\zeta)\big)^{-1}_+\!\!\!\!&=&\!\!\!\!\left\{\mathbb{I}_2+2\pi\im\lambda^{\frac{1}{2}}\begin{bmatrix}M_1(\zeta)\otimes K_1(\zeta) & M_1(\zeta)\otimes K_2(\zeta)\smallskip\\
	M_2(\zeta)\otimes K_1(\zeta) & M_2(\zeta)\otimes K_2(\zeta)\end{bmatrix}\right\}\big({\bf X}(\zeta)\big)^{-1}_+\nonumber\\
	\!\!\!\!&\stackrel{\eqref{z45}}{=}&\!\!\!\!\big({\bf X}(\zeta)\big)^{-1}_++2\pi\im\lambda^{\frac{1}{2}}\begin{bmatrix}M_1(\zeta)\otimes K_1(\zeta) & M_1(\zeta)\otimes K_2(\zeta)\smallskip\\
	M_2(\zeta)\otimes K_1(\zeta) & M_2(\zeta)\otimes K_2(\zeta)\end{bmatrix}\nonumber\\
	&&\hspace{1cm}\circ\left\{\mathbb{I}_2-\lambda^{\frac{1}{2}}\int_{\Sigma}\begin{bmatrix}M_1(\eta)\otimes L_1(\eta) & M_1(\eta)\otimes L_2(\eta)\smallskip\\
	M_2(\eta)\otimes L_1(\eta) & M_2(\eta)\otimes L_2(\eta)\end{bmatrix}\frac{\d\eta}{\eta-\zeta_+}\right\}\nonumber\\
	\!\!\!\!&\stackrel{\eqref{z46}}{=}&\!\!\!\!\big({\bf X}(\zeta)\big)^{-1}_++2\pi\im\lambda^{\frac{1}{2}}\begin{bmatrix}M_1(\zeta)\otimes L_1(\zeta) & M_1(\zeta)\otimes L_2(\zeta)\smallskip\\
	M_2(\zeta)\otimes L_1(\zeta) & M_2(\zeta)\otimes L_2(\zeta)\end{bmatrix},\ \ \ \zeta\in\Sigma.\label{z53}
\end{eqnarray}
Finally, combining \eqref{z52} with \eqref{z53} and \eqref{z34},
\begin{equation*}
	\begin{bmatrix}M_1(\zeta)\otimes K_1(\zeta) & M_1(\zeta)\otimes K_2(\zeta)\smallskip\\
	M_2(\zeta)\otimes K_1(\zeta) & M_2(\zeta)\otimes K_2(\zeta)\end{bmatrix}\big({\bf X}(\zeta)\big)^{-1}_{\pm}=\begin{bmatrix}M_1(\zeta)\otimes L_1(\zeta) & M_1(\zeta)\otimes L_2(\zeta)\smallskip\\
	M_2(\zeta)\otimes L_1(\zeta) & M_2(\zeta)\otimes L_2(\zeta)\end{bmatrix},\ \ \zeta\in\Sigma.
\end{equation*}
This proves the second equality in \eqref{z47} and thus concludes the proof.
\end{proof}
\begin{cor}\label{z:cor2} Consider the operators $M_i(\zeta),K_j(\zeta)$ and $N_i(\eta),L_j(\eta)$ as defined in \eqref{z33} and \eqref{z39},\eqref{z46}. Then for any $i,j\in\{1,2\}$, on $\mathcal{H}_1$,
\begin{equation}\label{z54}
	\int_{\Sigma}M_i(\eta)\otimes L_j(\eta)\,\d\eta=\int_{\Sigma}N_i(\eta)\otimes K_j(\eta)\,\d\eta, 
\end{equation}
followed by
\begin{equation}\label{z55}
	\int_{\Sigma}M_i(\eta)\otimes L_j(\eta)\,\eta\,\d\eta=\int_{\Sigma}N_i(\eta)\otimes K_j(\eta)\,\eta\,\d\eta+\lambda^{\frac{1}{2}}\left[\int_{\Sigma}N_i(\eta)\otimes K_j(\eta)\,\d\eta\right]^2.
\end{equation}
\end{cor}
\begin{proof} During the proof of Corollary \ref{z:cor1} we established the identity
\begin{equation*}
	\lambda^{\frac{1}{2}}\int_{\Sigma}\big(\overline{{\bf X}}(\eta)-\overline{{\bf Y}}(\eta)\big)\frac{\d\eta}{\eta-\zeta}-\lambda\int_{\Sigma}\int_{\Sigma}\overline{{\bf X}}(\eta_1)\overline{{\bf Y}}(\eta_2)\frac{\d\eta_1}{\eta_1-\zeta}\frac{\d\eta_2}{\eta_2-\zeta}={\bf 0},\ \ \ \ \ \zeta\notin\Sigma.
\end{equation*}
Using $\frac{1}{\eta-\zeta}=-\frac{1}{\zeta}-\frac{\eta}{\zeta^2}+\frac{\eta^2}{\zeta^2(\eta-\zeta)}$ for $\eta\neq\zeta$ and afterwards collecting powers in $\zeta^{-1}$ as $|\zeta|\rightarrow\infty$ yields precisely the stated identities. This concludes the proof.
\end{proof}
Besides the elementary symmetry constraints \eqref{z54} and \eqref{z55} we will make crucial use of the below identity.
\begin{lem}\label{z:lem8} For every $(x,y)\in\mathbb{R}^2$ and $(t,\lambda,n)\in\mathbb{R}\times\overline{\mathbb{D}_1(0)}\times\mathbb{N}$, 
\begin{equation}\label{z56}
	\int_{\Sigma}\big(N_2(\eta)\otimes K_1(\eta)\big)(x,y)\,\d\eta=\int_{\Sigma}\big(N_1(\eta)\otimes K_2(\eta)\big)(y,x)\,\d\eta,
\end{equation}
with $N_i(\eta)$ defined in \eqref{z39}.
\end{lem}
\begin{proof} By definition of $N_1(\eta)$ and $K_2(\eta)$, for every $(x,y)\in\mathbb{R}^2$,
\begin{align}
	\int_{\Sigma}\big(N_1(\eta)\otimes K_2(\eta)\big)&\,(x,y)\,\d\eta=\int_{\Sigma}n_1(\eta|x)k_2(\eta|y)\,\d\eta\nonumber\\
	&\,\stackrel[\eqref{z39}]{\eqref{z33}}{=}\frac{1}{2\pi}\int_{\Gamma_{\beta}}\left[\int_{\Gamma_{\beta}}(I-\lambda^{\frac{1}{2}}C_{t,n}\upharpoonright_{L^2(\Sigma)})^{-1}(\xi,\eta)\,\e^{-\frac{\im}{2}\psi_n(\xi,2t+2x)}\,\d\xi\right]\e^{-\frac{\im}{2}\psi_n(\eta,0)}\,\d\eta\nonumber\\
	&\,=\frac{1}{2\pi}\int_{\Gamma_{\alpha}}\left[\int_{\Gamma_{\alpha}}(I-\lambda^{\frac{1}{2}}C_{t,n}\upharpoonright_{L^2(\Sigma)})^{-1}(-\xi,-\eta)\,\e^{\frac{\im}{2}\psi_n(\xi,2t+2x)}\,\d\xi\right]\e^{\frac{\im}{2}\psi_n(\eta,0)}\,\d\eta,\label{z57}
\end{align}
where we have used the conjugation symmetry $\overline{\Gamma}_{\beta}=\Gamma_{\alpha}$, see Figure \ref{figz:3}, and the fact that $\lambda\mapsto\psi_n(\lambda,\cdot)$ is odd. Returning to the proof of Proposition \ref{zprop:3} we have $C_{t,n}=A_{t,n}^{\textnormal{ext}}+B_{n}^{\textnormal{ext}}$ on $L^2(\Sigma)$ and thus, by the mapping properties of $A_{t,n}^{\textnormal{ext}}:L^2(\Gamma_{\beta})\rightarrow L^2(\Gamma_{\alpha})$ and $B_{n}^{\textnormal{ext}}:L^2(\Gamma_{\alpha})\rightarrow L^2(\Gamma_{\beta})$,
\begin{equation}\label{z58}
	C_{t,n}^{2m+1}(-\xi,-\eta)=0\ \ \ \ \ \ \textnormal{for}\ \ \ \ \ \ (\xi,\eta)\in\Gamma_{\alpha}\times\Gamma_{\alpha},\ \ m\in\mathbb{Z}_{\geq 0},
\end{equation}
as well as, compare \eqref{z60} below,
\begin{equation}\label{z59}
	C_{t,n}^{2m}(-\xi,-\eta)=(B_{n}^{\textnormal{ext}}A_{t,n}^{\textnormal{ext}})^m(-\xi,-\eta)=C_{t,n}^{2m}(\eta,\xi)\ \ \ \ \ \ \textnormal{for}\ \ \ \ \ \ (\xi,\eta)\in\Gamma_{\alpha}\times\Gamma_{\alpha},\ \ m\in\mathbb{Z}_{\geq 0},
\end{equation}
where we used that $B_{n}(-\xi,-\eta)=B_{n}(\eta,\xi)$ for every $(\xi,\eta)\in\Gamma_{\alpha}\times\Gamma_{\beta}$ and $A_{t,n}(-\eta,-\zeta)=A_{t,n}(\zeta,\eta)$ for every $(\eta,\zeta)\in\Gamma_{\beta}\times\Gamma_{\alpha}$, compare \eqref{z24}, together with the aforementioned conjugation symmetry $\overline{\Gamma}_{\beta}=\Gamma_{\alpha}$. All together, given that
\begin{equation}\label{z60}
	C_{t,n}^k=\begin{cases}(A_{t,n}^{\textnormal{ext}}B_{n}^{\textnormal{ext}})^mA_{t,n}^{\textnormal{ext}}+(B_{n}^{\textnormal{ext}}A_{t,n}^{\textnormal{ext}})^mB_{n}^{\textnormal{ext}},&k=2m+1\smallskip\\
	(A_{t,n}^{\textnormal{ext}}B_{n}^{\textnormal{ext}})^m+(B_{n}^{\textnormal{ext}}A_{t,n}^{\textnormal{ext}})^m,&k=2m\end{cases},
\end{equation}
identities \eqref{z58} and \eqref{z59} show that for all $k\in\mathbb{N}:C_{t,n}^k(-\xi,-\eta)=C_{t,n}^k(\eta,\xi)$ whenever $(\xi,\eta)\in\Gamma_{\alpha}\times\Gamma_{\alpha}$ and thus by a Neumann series expansions argument back in \eqref{z57},
\begin{align*}
	\int_{\Sigma}\big(N_1(\eta)&\,\otimes K_2(\eta)\big)(x,y)\,\d\eta=\frac{1}{2\pi}\int_{\Gamma_{\alpha}}\left[\int_{\Gamma_{\alpha}}(I-\lambda^{\frac{1}{2}}C_{t,n}\upharpoonright_{L^2(\Sigma)})^{-1}(\eta,\xi)\,\e^{\frac{\im}{2}\psi_n(\xi,2t+2x)}\,\d\xi\right]\e^{\frac{\im}{2}\psi_n(\eta,0)}\,\d\eta\\
	&\stackrel[\eqref{z46}]{\eqref{z33}}{=}\int_{\Sigma}\ell_1(\eta|x)m_2(\eta|y)\,\d\eta=\int_{\Sigma}\big(M_2(\eta)\otimes L_1(\eta)\big)(y,x)\,\d\eta\stackrel{\eqref{z54}}{=}\int_{\Sigma}\big(N_2(\eta)\otimes K_1(\eta)\big)(y,x)\,\d\eta.
\end{align*}
This proves \eqref{z56} after relabelling $x$ and $y$.
\end{proof}
Yet another useful result deals with the following $t\rightarrow+\infty$ behavior of the solution \eqref{z38} to RHP \ref{zmaster}.
\begin{cor}\label{z:cor3} Let $i,j\in\{1,2\}$ and $m\in\mathbb{Z}_{\geq 0}$ with $0<\Delta<\frac{\omega}{2}$ fixed as indicated in Figure \ref{figz:3}. Then
\begin{equation}\label{z61}
	\int_{\Sigma}N_i(\eta)\otimes K_j(\eta)\,\eta^m\,\d\eta\rightarrow 0\ \ \ \ \textnormal{and}\ \ \ \ \int_{\Sigma}M_i(\eta)\otimes L_j(\eta)\,\eta^m\,\d\eta\rightarrow 0
\end{equation}
exponentially fast as $t\rightarrow+\infty$ in operator norm on $\mathcal{H}_1$.
\end{cor}
\begin{proof} We only focus on the first limit in \eqref{z61}, the second one follows by analogous logic. To this end we now show that $A_{t,n}^{\textnormal{ext}}\rightarrow 0$ exponentially fast as $t\rightarrow+\infty$ in operator and trace norm on $L^2(\Sigma)$. Indeed, from the proof of Lemma \ref{zlem:6}, for the trace norm, 
\begin{equation*}
	\|A_{t,n}^{\textnormal{ext}}\|_1\leq \|A_{t,n,1}^{\textnormal{ext}}\|_2\|A_{t,n,2}^{\textnormal{ext}}\|_2\leq c\,\e^{-2t\Delta},\ \ \ c=c(n,\Delta)>0,
\end{equation*}
uniformly in $t\in\mathbb{R}$. Hence, since the operator norm on $L^2(\Sigma)$ is dominated by the trace norm, we have likewise $\|A_{t,n}^{\textnormal{ext}}\|\leq c\,\e^{-2t\Delta}$ with $c=c(n,\Delta)>0$ uniformly in $t\in\mathbb{R}$. This shows by the Neumann series that for any $(\lambda,n)\in\overline{\mathbb{D}_1(0)}\times\mathbb{N}$ and $0<\Delta<\frac{\omega}{2}$ we can find $t_0=t_0(n,\lambda,\Delta)>0$ and $c=c(n,\Delta)>0$ such that
\begin{equation}\label{z62}
	\|(I-\lambda^{\frac{1}{2}}C_{t,n})^{-1}-I-\lambda^{\frac{1}{2}}C_{t,n}\|\leq c\,\e^{-2t\Delta}\ \ \ \ \ \forall\,t\geq t_0
\end{equation}
in operator norm on $L^2(\Sigma)$, compare \eqref{z60}\footnote{We have $\|C_{t,n}^k\|_1\rightarrow 0$ as $t\rightarrow+\infty$ for all $k\in\mathbb{Z}_{\geq 2}$ but not for $k=1$ since $\|B_{n}^{\textnormal{ext}}\|_1\nrightarrow 0$ as $t\rightarrow+\infty$.}. Equipped with \eqref{z62} we now return to \eqref{z39} and derive the following improvement of \eqref{z40}, for any $(\lambda,n)\in\overline{\mathbb{D}_1(0)}\times\mathbb{N}$ and $0<\Delta<\frac{\omega}{2}$ as well as $x\in\mathbb{R}$,
\begin{equation*}
	\|n_i(\cdot|x)-m_i(\cdot|x)-\lambda^{\frac{1}{2}}C_{t,n}^{\ast}m_i(\cdot|x)\|_{L^2(\Sigma)}\leq c\,\e^{-2t\Delta}\|m_i(\cdot|x)\|_{L^2(\Sigma)}\ \ \ \ \forall\,t\geq t_0,\ \ \ i\in\{1,2\},
\end{equation*}
with $c=c(n,\Delta)>0$. Consequently, by triangle inequality,
\begin{equation}\label{z63}
	\int_{\mathbb{R}}\|n_i(\cdot|x)\|_{L^2(\Sigma)}^2\d\sigma(x)\leq \big(1+c\,\e^{-2t\Delta}\big)\int_{\mathbb{R}}\|m_i(\cdot|x)\|_{L^2(\Sigma)}^2\d\sigma(x),
\end{equation}
where we used
\begin{equation*}
	\|C_{t,n}^{\ast}m_1(\cdot|x)\|_{L^2(\Sigma)}^2\leq c\,\e^{-2\Delta(t+x)},\ \ \ \ \|C_{t,n}^{\ast}m_2(\cdot|x)\|_{L^2(\Sigma)}^2\leq c\,\e^{-4\Delta t},\ \ c=c(n,\Delta)>0,
\end{equation*}
uniformly in $(t,x)\in\mathbb{R}^2$. This allows us to prove \eqref{z61} as follows: first, by Cauchy-Schwarz inequality (once on $\mathcal{H}_1$ and once on $L^2(\Sigma)$), Fubini's theorem and with the shorthand $T_{ij}^m:=\int_{\Sigma}N_i(\eta)\otimes K_j(\eta)\eta^m\d\eta$, for any $f\in\mathcal{H}_1$,
\begin{align*}
	\big|(T_{ij}^mf)(x)\big|=&\,\left|\int_{\mathbb{R}}\left(\int_{\Sigma}n_i(\eta|x)k_j(\eta|y)\eta^m\d\eta\right)f(y)\,\d\sigma(y)\right|\leq\int_{\Sigma}\left|n_i(\eta|x)\eta^m\left(\int_{\mathbb{R}}k_j(\eta|y)f(y)\,\d\sigma(y)\right)\right|\,|\d\eta|\\
	\leq&\,\int_{\Sigma}\big|n_i(\eta|x)\eta^m\big|\|k_j(\eta|\cdot)\|_{\mathcal{H}_1}\|f\|_{\mathcal{H}_1}\,|\d\eta|\leq\|n_i(\cdot|x)\|_{L^2(\Sigma)}\sqrt{\int_{\Sigma}\int_{\mathbb{R}}\big|k_j(\eta|y)\eta^m\big|^2\d\sigma(y)\,|\d\eta|}\,\|f\|_{\mathcal{H}_1},
\end{align*}
so that in operator norm on $\mathcal{H}_1$,
\begin{equation}\label{z64}
	\|T_{ij}^m\|^2\leq\left(\int_{\mathbb{R}}\|n_i(\cdot|x)\|_{L^2(\Sigma)}^2\,\d\sigma(x)\right)\left(\int_{\Sigma}\int_{\mathbb{R}}\big|k_j(\eta|y)\eta^m\big|^2\d\sigma(y)\,|\d\eta|\right).
\end{equation}
Second, using \eqref{z33}, for any $n\in\mathbb{N}$ and $0<\Delta<\frac{\omega}{2}$ there exists $c=c(n,\Delta)>0$ such that
\begin{equation*}
	\|m_1(\cdot|x)\|_{L^2(\Sigma)}^2\leq c\,\e^{-2\Delta(t+x)},\ \ \ \|m_2(\cdot|x)\|_{L^2(\Sigma)}^2\leq c,\ \ \ \|k_1(\cdot|y)\|_{L^2(\Sigma)}^2\leq c\,\e^{-2\Delta(t+y)},\ \ \ \|k_2(\cdot|y)\|_{L^2(\Sigma)}^2\leq c
\end{equation*}
hold true for all $(x,y,t)\in\mathbb{R}^3$. Moreover, for any fixed $m\in\mathbb{Z}_{\geq 0}$ there exists $\hat{c}=\hat{c}(n,\Delta,m)>0$ such that
\begin{equation*}
	\int_{\Sigma}\big|k_1(\eta|y)\eta^m\big|^2|\d\eta|\leq \hat{c}\,\e^{-2\Delta(t+y)},\ \ \ \ \int_{\Sigma}\big|k_2(\eta|y)\eta^m\big|^2|\d\eta|\leq\hat{c}
\end{equation*}
uniformly in $(y,t)\in\mathbb{R}^2$. Hence, combining \eqref{z64} with \eqref{z63} and using the last six estimates we immediately establish $T_{ij}^m\rightarrow 0$ exponentially fast in operator norm on $\mathcal{H}_1$ as $t\rightarrow+\infty$ in view of \eqref{i13}, provided $i,j\in\{1,2\}$ but $(i,j) \neq (2,2)$. Lastly, if $i=j=2$, then 
\begin{equation}\label{z65}
	T_{22}^m\stackrel{\eqref{z39}}{=}\int_{\Sigma}M_2(\eta)\otimes K_2(\eta)\,\eta^m\,\d\eta+\lambda^{\frac{1}{2}}C_{t,n}^{\ast}\int_{\Sigma}M_2(\eta)\otimes K_2(\eta)\,\eta^m\,\d\eta+\lambda (C_{t,n}^{\ast})^2\int_{\Sigma}N_2(\eta)\otimes K_2(\eta)\,\eta^m\,\d\eta
\end{equation}
where we used $(I-\lambda^{\frac{1}{2}}C_{t,n}^{\ast})^{-1}=I+\lambda^{\frac{1}{2}}C_{t,n}^{\ast}+\lambda(C_{t,n}^{\ast})^2(I-\lambda^{\frac{1}{2}}C_{t,n}^{\ast})^{-1}$. The first summand in \eqref{z65} vanishes identically by \eqref{z33} since $\Gamma_{\alpha}$ and $\Gamma_{\beta}$ are disjoint so we only need to control the following $\mathcal{H}_1$ operator norms
\begin{eqnarray*}
	\left\|C_{t,n}^{\ast}\int_{\Sigma}M_2(\eta)\otimes K_2(\eta)\,\eta^m\,\d\eta\right\|\!\!\!\!&\leq&\!\!\!\!\left(\int_{\mathbb{R}}\|C_{t,n}^{\ast}m_2(\cdot|x)\|_{L^2(\Sigma)}^2\,\d\sigma(x)\right)\left(\int_{\Sigma}\int_{\mathbb{R}}\big|k_2(\eta|y)\eta^m\big|^2\d\sigma(y)\,|\d\eta|\right)\\
	&\leq&\!\!\!\! \hat{c}\,\e^{-2\Delta t},\ \ \ \ \hat{c}=\hat{c}(n,\Delta,m)>0,
\end{eqnarray*}
and
\begin{eqnarray*}
	\left\|(C_{t,n}^{\ast})^2\int_{\Sigma}N_2(\eta)\otimes K_2(\eta)\,\eta^m\,\d\eta\right\|\!\!\!\!&\leq&\!\!\!\!\|(C_{t,n}^{\ast})^2\|\left(\int_{\mathbb{R}}\|n_2(\cdot|x)\|_{L^2(\Sigma)}^2\,\d\sigma(x)\right)\left(\int_{\Sigma}\int_{\mathbb{R}}\big|k_2(\eta|y)\eta^m\big|^2\d\sigma(y)\,|\d\eta|\right)\\
	&\leq&\!\!\!\! \hat{c}\,\|C_{t,n}^2\|\leq \hat{c}\,\|C_{t,n}^2\|_1\stackrel{\eqref{z60}}{\leq} \hat{c}\,\|A_{t,n}^{\textnormal{ext}}\|_1\|B_{n}^{\textnormal{ext}}\|\leq \hat{c}\,\e^{-2\Delta t},
\end{eqnarray*}
also with $\hat{c}=\hat{c}(n,\Delta,m)>0$. Combined in \eqref{z65}, this proves the exponentially fast convergence of $T_{22}^m$ to the zero operator on $\mathcal{H}_1$ in operator norm. Our proof of \eqref{z61} is now completed.
\end{proof}
Corollaries \ref{z:cor1}, \ref{z:cor2}, \ref{z:cor3} and Lemma \ref{z:lem8} conclude the content of this section on the $\mathcal{I}(\mathcal{H}_2)$-valued RHP \ref{zmaster}. We will now use this problem in the proof of Theorem \ref{itheo1}.

\section{The integro-differential Painlev\'e-II hierarchy - proof of Theorem \ref{itheo1}, part 1}\label{zsec4}
In order to arrive at the integro-differential dynamical system \eqref{i17} we will first derive a certain operator-valued Lax pair from RHP \ref{zmaster}, see Proposition \ref{z:prop4} and equation \eqref{z74} below. This approach follows the methodology in \cite[Section $9.3$]{B} and plays out as follows. First, we view the multiplication operators $M_i(\zeta)$ and $N_i(\zeta)$, see \eqref{z33} and \eqref{z39}, as integral operators on $\mathcal{H}_1$ with distributional kernels. In detail, we replace
\begin{align*}
	m_i(\zeta|x)\mapsto m_i(\zeta|x,y):=&\,m_i(\zeta|x)\delta(x-y)(w'(y))^{-1},\\
	n_i(\zeta|x)\mapsto n_i(\zeta|x,y):=&\,n_i(\zeta|x)\delta(x-y)(w'(y))^{-1},
\end{align*}
with $(x,y)\in\mathbb{R}$ and where, by definition,
\begin{equation*}
	\int_{-\infty}^{\infty}\delta(x-y)(w'(y))^{-1}f(y)\,\d\sigma(y):=f(x),\ \ \ \ f\in\mathcal{H}_1.
\end{equation*}
Then, differentiating \eqref{z33}, we find the kernel identity (recall the definition of ${\bf M}(\zeta)$ in Corollary \ref{z:cor1})
\begin{equation*}
	\frac{\partial}{\partial\zeta}\,{\bf M}(\zeta|x,y)=\begin{bmatrix}-\im(\frac{1}{2}\zeta^{2n}+t+x) & 0\smallskip\\
	0 & \frac{\im}{2}\zeta^{2n}\end{bmatrix}{\bf M}(\zeta|x,y),\ \ \ \ (\zeta,x,y)\in\Sigma\times\mathbb{R}^2,
\end{equation*}
or equivalently the operator identity
\begin{equation}\label{z66}
	\frac{\partial}{\partial\zeta}\,{\bf M}(\zeta)=\big(\zeta^{2n}{\bf A}_0+\widehat{{\bf A}}_{2n}\big){\bf M}(\zeta),\ \ \ \zeta\in\Sigma,
\end{equation}
where the operators ${\bf A}_0,\widehat{{\bf A}}_{2n}:\mathcal{H}_2\rightarrow\mathcal{H}_2$ are $\zeta$-independent and have kernels
\begin{equation}\label{z67}
	{\bf A}_0(x,y):=\delta(x-y)\frac{1}{2}\begin{bmatrix}-\im & 0\smallskip\\
	0 & \im\end{bmatrix}\big(w'(y)\big)^{-1},\ \ \ \ \widehat{{\bf A}}_{2n}(x,y):=\delta(x-y)\begin{bmatrix}-\im(t+x) & 0\smallskip\\
	0 & 0\end{bmatrix}\big(w'(y)\big)^{-1}.
\end{equation}
Similarly,
\begin{equation}\label{z68}
	\frac{\partial}{\partial t}\,{\bf M}(\zeta)=\big(\zeta{\bf B}_0\big){\bf M}(\zeta),\ \ \ \ \zeta\in\Sigma,
\end{equation}
where ${\bf B}_0:\mathcal{H}_2\rightarrow\mathcal{H}_2$ has kernel
\begin{equation}\label{z69}
	{\bf B}_0(x,y):=\delta(x-y)\begin{bmatrix}-\im & 0\smallskip\\ 0 & 0\end{bmatrix}\big(w'(y)\big)^{-1}.
\end{equation}
At this point we return to \eqref{z47}.
\begin{prop}\label{z:prop4} There exist $(t,\lambda,n)$-dependent, analytic in $\zeta\in\mathbb{C}$ integral operators ${\bf A}(\zeta),{\bf B}(\zeta)$ on $\mathcal{H}_2$ such that for every $\zeta\in\Sigma$ and $(t,\lambda,n)\in\mathbb{R}\times\overline{\mathbb{D}_1(0)}\times\mathbb{N}$,
\begin{equation*}
	\frac{\partial{\bf N}}{\partial\zeta}(\zeta)={\bf A}(\zeta){\bf N}(\zeta),\ \ \ \ \frac{\partial{\bf N}}{\partial t}(\zeta)={\bf B}(\zeta){\bf N}(\zeta).
\end{equation*}
\end{prop}
\begin{proof} We $\zeta$-differentiate the first identity in \eqref{z47} using \eqref{z66},
\begin{equation*}
	\frac{\partial{\bf N}}{\partial\zeta}(\zeta)=\underbrace{\left[\frac{\partial{\bf X}}{\partial\zeta}(\zeta)\big({\bf X}(\zeta)\big)^{-1}+{\bf X}(\zeta)\big(\zeta^{2n}{\bf A}_0+\widehat{{\bf A}}_{2n}\big)\big({\bf X}(\zeta)\big)^{-1}\right]}_{=:{\bf A}(\zeta)}{\bf N}(\zeta).
\end{equation*}
Here, ${\bf A}(\zeta)\in\mathcal{I}(\mathcal{H}_2)$ by Theorem \ref{z:theo1}, Corollary \ref{z:cor1}, and ${\bf A}(\zeta)$ is analytic for $\zeta\in\mathbb{C}\setminus\Sigma$ with continuous boundary values ${\bf A}_{\pm}(\zeta)\in\mathcal{I}(\mathcal{H}_2)$ on $\Sigma$ by the same reasoning. Recalling \eqref{z34} we then compute on $\Sigma$,
\begin{align}
	{\bf A}_+(\zeta)=\bigg[\frac{\partial{\bf X}_-}{\partial\zeta}(\zeta){\bf G}(\zeta)+&\,{\bf X}_-(\zeta)\frac{\partial{\bf G}}{\partial\zeta}(\zeta)\bigg]\big({\bf G}(\zeta)\big)^{-1}\big({\bf X}_-(\zeta)\big)^{-1}\nonumber\\
	&\,+{\bf X}_-(\zeta){\bf G}(\zeta)\big(\zeta^{2n}{\bf A}_0+\widehat{{\bf A}}_{2n}\big)\big({\bf G}(\zeta)\big)^{-1}\big({\bf X}_-(\zeta)\big)^{-1},\label{z70}
\end{align}
and with \eqref{z34},\eqref{z67} derive for $\zeta\in\Sigma$,
\begin{equation*}
	\frac{\partial{\bf G}}{\partial\zeta}(\zeta|x,y)=\int_{-\infty}^{\infty}\Big\{\big(\zeta^{2n}{\bf A}_0(x,z)+\widehat{{\bf A}}_{2n}(x,z)\big){\bf G}_0(\zeta|z,y)-{\bf G}_0(\zeta|x,z)\big(\zeta^{2n}{\bf A}_0(z,y)+\widehat{{\bf A}}_{2n}(z,y)\big)\Big\}\,\d\sigma(z).
\end{equation*}
Here we abbreviate, as in the proof of Lemma \ref{z:lem7}, ${\bf G}(\zeta)=\mathbb{I}_2+{\bf G}_0(\zeta)$ and note that the last kernel identity is equivalent to the operator commutator identity
\begin{equation}\label{z71}
	\frac{\partial{\bf G}}{\partial\zeta}(\zeta)=\big[\zeta^{2n}{\bf A}_0+\widehat{{\bf A}}_{2n},{\bf G}(\zeta)\big]\in\mathcal{I}(\mathcal{H}_2),\ \ \ \ \zeta\in\Sigma.
\end{equation}
Inserting \eqref{z71} into \eqref{z70} we find at once
\begin{equation*}
	{\bf A}_+(\zeta)=\frac{\partial{\bf X}_-}{\partial\zeta}(\zeta)\big({\bf X}_-(\zeta)\big)^{-1}+{\bf X}_-(\zeta)\big(\zeta^{2n}{\bf A}_0+\widehat{{\bf A}}_{2n}\big)\big({\bf X}_-(\zeta)\big)^{-1}={\bf A}_-(\zeta),\ \ \zeta\in\Sigma,
\end{equation*}
i.e. ${\bf A}(\zeta)$ extends analytically across $\Sigma$. In turn, ${\bf A}(\zeta)$ is analytic for every $\zeta\in\mathbb{C}$ given that $(x,y)\mapsto{\bf A}(\zeta|x,y)$ is in $L^2(\mathbb{R}^2,\d\sigma\otimes\d\sigma;\mathbb{C}^{2\times 2})$ for every $\zeta\in\mathbb{C}$ by construction. This proves our first identity and the reasoning for the second one is analogous: first differentiate \eqref{z47} using \eqref{z68},
\begin{equation}\label{z72}
	\frac{\partial{\bf N}}{\partial t}(\zeta)=\underbrace{\left[\frac{\partial{\bf X}}{\partial t}(\zeta)\big({\bf X}(\zeta)\big)^{-1}+{\bf X}(\zeta)\big(\zeta{\bf B}_0\big)\big({\bf X}(\zeta)\big)^{-1}\right]}_{=:{\bf B}(\zeta)}{\bf N}(\zeta).
\end{equation}
Since ${\bf B}(\zeta)\in\mathcal{I}(\mathcal{H}_2)$ and ${\bf B}(\zeta)$ is analytic for $\zeta\in\mathbb{C}\setminus\Sigma$ with continuous boundary values ${\bf B}_{\pm}(\zeta)\in\mathcal{I}(\mathcal{H}_2)$ on $\Sigma$, compare Theorem \ref{z:theo1} and Corollary \ref{z:cor1}, we simply compute for $\zeta\in\Sigma$
\begin{equation}\label{z73}
	{\bf B}_+(\zeta)=\left[\frac{\partial{\bf X}_-}{\partial t}{\bf G}(\zeta)+{\bf X}_-(\zeta)\frac{\partial{\bf G}}{\partial t}(\zeta)\right]\big({\bf G}(\zeta)\big)^{-1}\big({\bf X}_-(\zeta)\big)^{-1}+{\bf X}_-(\zeta){\bf G}(\zeta)\big(\zeta{\bf B}_0\big)\big({\bf G}(\zeta)\big)^{-1}\big({\bf X}_-(\zeta)\big)^{-1}.
\end{equation}
But from \eqref{z34},\eqref{z69},
\begin{equation*}
	\frac{\partial{\bf G}}{\partial t}(\zeta|x,y)=\int_{-\infty}^{\infty}\Big\{\big(\zeta{\bf B}_0(x,z)\big){\bf G}_0(\zeta|z,y)-{\bf G}_0(\zeta|x,z)\big(\zeta{\bf B}_0(z,y)\big)\Big\}\,\d\sigma(z),
\end{equation*}
leading to the following replacement of \eqref{z71}
\begin{equation*}
	\frac{\partial{\bf G}}{\partial t}(\zeta)=\big[\zeta{\bf B}_0,{\bf G}(\zeta)\big]\in\mathcal{I}(\mathcal{H}_2),\ \ \ \zeta\in\Sigma.
\end{equation*}
Once substituted back into \eqref{z73} we find at once ${\bf B}_+(\zeta)={\bf B}_-(\zeta)$ for $\zeta\in\Sigma$, i.e. ${\bf B}(\zeta)$ is analytic for $\zeta\in\mathbb{C}$. This concludes our proof.
\end{proof}
In our next step we will express the coefficient operators ${\bf A}(\zeta),{\bf B}(\zeta)$ introduced in Proposition \ref{z:prop4} to some extent in terms of the solution of RHP \ref{zmaster}.
\begin{prop}\label{z:prop5} We have 
\begin{equation}\label{z74}
	{\bf B}(\zeta)=\zeta{\bf B}_0+{\bf B}_1,\ \ \ \ \ \ {\bf A}(\zeta)=\zeta^{2n}{\bf A}_0+\sum_{k=1}^{2n}{\bf A}_k\zeta^{2n-k}+\widehat{{\bf A}}_{2n},
\end{equation}
where ${\bf B}_j:\mathcal{H}_2\rightarrow\mathcal{H}_2$ are the $\zeta$-independent integral operators with kernels written in \eqref{z69} and \eqref{z76} below. Likewise, ${\bf A}_j:\mathcal{H}_2\rightarrow\mathcal{H}_2$ are $\zeta$-independent, the kernels of ${\bf A}_0$ and $\widehat{{\bf A}}_{2n}$ are written in \eqref{z67} and the entries of ${\bf A}_k$ are polynomials in $\int_{\Sigma}N_i(\eta)\otimes K_j(\eta)\eta^m\d\eta$ and $\int_{\Sigma}M_i(\eta)\otimes L_j(\eta)\eta^m\d\eta$ with $m\in\mathbb{Z}_{\geq 0},i,j\in\{1,2\}$.
\end{prop}
\begin{proof} Return to \eqref{z38} and \eqref{z45}, write $\frac{1}{\eta-\zeta}=-\frac{1}{\zeta}\sum_{k=0}^{2n-1}(\frac{\eta}{\zeta})^k+\frac{\eta^{2n}}{\zeta^{2n}(\eta-\zeta)}$ for $\zeta\neq\eta$ and derive
\begin{align*}
	{\bf X}(\zeta)=\mathbb{I}_2-\lambda^{\frac{1}{2}}\sum_{k=1}^{2n}\frac{1}{\zeta^k}\int_{\Sigma}\begin{bmatrix}N_1(\eta)\otimes K_1(\eta) & N_1(\eta)\otimes K_2(\eta)\smallskip\\
	N_2(\eta)\otimes K_1(\eta) & N_2(\eta)\otimes K_2(\eta)\end{bmatrix}\eta^{k-1}\,\d\eta+{\bf X}_{\textnormal{asy}}(\zeta),\ \ \zeta\notin\Sigma,
\end{align*}
where ${\bf X}_{\textnormal{asy}}(\zeta)\in\mathcal{I}(\mathcal{H}_2)$ and there exists $c=c(n,t)>0$ such that for $\zeta\in\mathbb{C}\setminus\Sigma$,
\begin{equation}\label{z75}
	\|{\bf X}_{\textnormal{asy}}(\zeta|x,y)\|\leq\frac{c\sqrt{|\lambda|}}{1+|\zeta|^{2n+1}}\Delta^{-\frac{1}{4n}}\e^{-\frac{(-1)^n\Delta}{2(2n+1)}\Delta^{2n}}\e^{\Delta(|x|+|y|+|t|)}
\end{equation}
uniformly in $(x,y)\in\mathbb{R}^2$ and $\lambda\in\overline{\mathbb{D}_1(0)}$. Similarly from \eqref{z45},
\begin{align*}
	\big({\bf X}(\zeta)\big)^{-1}=\mathbb{I}_2+\lambda^{\frac{1}{2}}\sum_{k=1}^{2n}\frac{1}{\zeta^k}\int_{\Sigma}&\,\begin{bmatrix}M_1(\eta)\otimes L_1(\eta) & M_1(\eta)\otimes L_2(\eta)\smallskip\\
	M_2(\eta)\otimes L_1(\eta) & M_2(\eta)\otimes L_2(\eta)\end{bmatrix}\eta^{k-1}\,\d\eta+{\bf X}_{\textnormal{asy}}^{-1}(\zeta),\ \ \zeta\notin\Sigma,
\end{align*}
where\footnote{Please note that ${\bf X}_{\textnormal{asy}}^{-1}$ is not the inverse of ${\bf X}_{\textnormal{asy}}$, as the notation might suggest.} ${\bf X}_{\textnormal{asy}}^{-1}(\zeta)\in\mathcal{I}(\mathcal{H}_2)$ also satisfies \eqref{z75}. At this point we first return to the defining equation of ${\bf B}(\zeta)$ in \eqref{z72} and apply Liouville's theorem combined with Corollary \ref{z:cor2},
\begin{equation*}
	{\bf B}(\zeta)=\zeta{\bf B}_0+{\bf B}_1,\ \ \ \zeta\in\mathbb{C},
\end{equation*}
where ${\bf B}_1$ is the integral operator on $\mathcal{H}_2$ with kernel ${\bf B}_1(x,y)=[B_1^{ij}(x,y)]_{i,j=1}^2$ and
\begin{align}\label{z76}
	B_1^{11}(x,y)=B_1^{22}(x,y)=0,\ \ \ \ \ \ \ \ \ B_1^{12}(x,y)=-\im U(x,y),\ \ \ B_1^{21}(x,y)=\im V(x,y),
\end{align}
where we use the shorthand $U(x,y)$ and $V(x,y)$ for the kernels of $U:=\lambda^{\frac{1}{2}}\int_{\Sigma}N_1(\eta)\otimes K_2(\eta)\d\eta$ and $V:=\lambda^{\frac{1}{2}}\int_{\Sigma}N_2(\eta)\otimes K_1(\eta)\d\eta$.	
Next, by similar logic, 
\begin{equation*}
	{\bf A}(\zeta)=\zeta^{2n}{\bf A}_0+\sum_{k=1}^{2n}{\bf A}_k\zeta^{2n-k}+\widehat{{\bf A}}_{2n},\ \ \ \zeta\in\mathbb{C},
\end{equation*}
obtained from inserting the above representations for ${\bf X}(\zeta)$ and $({\bf X}(\zeta))^{-1}$ into the defining equation for ${\bf A}(\zeta)$ and applying Liouville's theorem. This concludes the proof of \eqref{z74}.
\end{proof}
Before moving on and before we employ \eqref{z74} in the operator-valued Lax pair
\begin{equation}\label{z77}
	\frac{\partial{\bf N}}{\partial\zeta}(\zeta)={\bf A}(\zeta){\bf N}(\zeta),\ \ \ \ \frac{\partial{\bf N}}{\partial t}(\zeta)={\bf B}(\zeta){\bf N}(\zeta),
\end{equation}
we first record the following connection formula.
\begin{lem}\label{z:lem9} For every $(t,\lambda,n)\in\mathbb{R}\times\overline{\mathbb{D}_1(0)}\times\mathbb{N}$,
\begin{equation*}
	\frac{\partial}{\partial t}\ln D_n(t,\lambda)=-\im\lambda^{\frac{1}{2}}\tr_{\mathcal{H}_1}\int_{\Sigma}N_1(\xi)\otimes K_1(\xi)\,\d\xi,
\end{equation*}
followed by
\begin{equation*}
	\frac{\partial^2}{\partial t^2}\ln D_n(t,\lambda)=-\lambda\tr_{\mathcal{H}_1}\int_{\Sigma}\int_{\Sigma}\big(N_1(\eta)\otimes K_2(\eta)\big)\big(N_2(\xi)\otimes K_1(\xi)\big)\,\d\eta\,\d\xi.
\end{equation*}
\end{lem}
\begin{proof} We begin with 
\begin{equation}\label{z78}
	\frac{\partial}{\partial t}\ln D_n(t,\lambda)\stackrel{\eqref{z37}}{=}\frac{\partial}{\partial t}\ln\det(I-\lambda^{\frac{1}{2}}C_{t,n}\upharpoonright_{L^2(\Sigma)})=-\lambda^{\frac{1}{2}}\tr_{L^2(\Sigma)}\left[(I-\lambda^{\frac{1}{2}}C_{t,n}\upharpoonright_{L^2(\Sigma)})^{-1}\frac{\partial}{\partial t}C_{t,n}\right]
\end{equation}
and compute from \eqref{z32} the kernel derivative
\begin{equation*}
	\frac{\partial}{\partial t}C_{t,n}(\xi,\eta)=\frac{\im}{2\pi}\int_{\mathbb{R}}\e^{\frac{\im}{2}(\psi_n(\xi,2t+2z)-\psi_n(\eta,2t+2z))}\chi_{\Gamma_{\alpha}}(\xi)\chi_{\Gamma_{\beta}}(\eta)\,\d\sigma(z)\stackrel{\eqref{z33}}{=}\im\int_{\mathbb{R}}k_1(\xi|z)m_1(\eta|z)\,\d\sigma(z).
\end{equation*}
Hence back in \eqref{z78},
\begin{eqnarray*}
	\frac{\partial}{\partial t}\ln D_n(t,\lambda)&=&-\lambda^{\frac{1}{2}}\int_{\Sigma}\int_{\Sigma}(I-\lambda^{\frac{1}{2}}C_{t,n})^{-1}(\eta,\xi)\frac{\partial}{\partial t} C_{t,n}(\xi,\eta)\,\d\xi\,\d\eta\\
	&\stackrel{\eqref{z39}}{=}&-\im\lambda^{\frac{1}{2}}\int_{\mathbb{R}}\left[\int_{\Sigma}\big(N_1(\xi)\otimes K_1(\xi)\big)(z,z)\,\d\xi\right]\d\sigma(z)=-\im\lambda^{\frac{1}{2}}\tr_{\mathcal{H}_1}\int_{\Sigma}N_1(\xi)\otimes K_1(\xi)\,\d\xi,
\end{eqnarray*}
as claimed in the first identity. For the second identity we revisit our proof of Proposition \ref{z:prop5} and explicitly compute the $\mathcal{O}(\zeta^{-1})$ correction when inserting the asymptotic representations of ${\bf X}(\zeta)$ and $({\bf X}(\zeta))^{-1}$ into the defining equation of ${\bf B}(\zeta)$ in \eqref{z72}. The same $\mathcal{O}(\zeta^{-1})$ correction has to vanish identically by Liouville's theorem and this yields the operator commutator identity
\begin{equation*}
	({\bf X}_1)_t=[{\bf B}_0,{\bf X}_2]-{\bf X}_1[{\bf B}_0,{\bf X}_1],\ \ \ \ \ {\bf X}_k:=\int_{\Sigma}\begin{bmatrix}M_1(\eta)\otimes L_1(\eta) & M_1(\eta)\otimes L_2(\eta)\smallskip\\
	M_2(\eta)\otimes L_1(\eta) & M_2(\eta)\otimes L_2(\eta)\end{bmatrix}\,\eta^{k-1}\,\d\eta,
\end{equation*}
where ${\bf B}_0$ is written in \eqref{z69}. Reading the last identity entry wise and using \eqref{z54},\eqref{z55} yields in particular
\begin{equation*}
	\frac{\partial}{\partial t}\left(\lambda^{\frac{1}{2}}\int_{\Sigma}N_1(\xi)\otimes K_1(\xi)\,\d\xi\right)=-\im\lambda\int_{\Sigma}\int_{\Sigma}\big(N_1(\eta)\otimes K_2(\eta)\big)\big(N_2(\xi)\otimes K_1(\xi)\big)\,\d\eta\,\d\xi,
\end{equation*}
and therefore the second identity when applying the first.
\end{proof}
At this point we use our Lax pair \eqref{z77},\eqref{z74} and write out its compatibility condition
\begin{equation}\label{z79}
	{\bf A}(\zeta){\bf B}(\zeta)-{\bf B}(\zeta){\bf A}(\zeta)=\frac{\partial{\bf B}}{\partial\zeta}(\zeta)-\frac{\partial{\bf A}}{\partial t}(\zeta),\ \ \ \ \ \zeta\in\mathbb{C},
\end{equation}
keeping in mind that the entries of ${\bf A}(\zeta)$ and ${\bf B}(\zeta)$ are integral operators which in general do not commute. 
\begin{lem}\label{z:lem10} Recall $U,V:\mathcal{H}_1\rightarrow\mathcal{H}_1$ in \eqref{z76}
and introduce the integral operator $M_t:\mathcal{H}_1\rightarrow\mathcal{H}_1$ with distributional kernel $M_t(x,y):=(t+x)\delta(x-y)(w'(y))^{-1}$. Then \eqref{z79} is equivalent to the operator-valued system \eqref{z80}, \eqref{z81} and \eqref{z82} written out below where $A_k^{ij}$ are the entries of ${\bf A}_k$ in \eqref{z74}.
\end{lem}
\begin{proof} The polynomial equation \eqref{z79} yields at once (given that ${\bf B}_0$ and ${\bf A}_0$ as well as ${\bf B}_0$ and $\widehat{{\bf A}}_{2n}$ commute)
\begin{equation*}
	\sum_{k=1}^{2n}\frac{\partial{\bf A}_k}{\partial t}\zeta^{2n-k}=\big[{\bf B}_1,{\bf A}_{2n}+\widehat{{\bf A}}_{2n}\big]+\sum_{k=0}^{2n-1}\Big(\big[{\bf B}_0,{\bf A}_{k+1}\big]+\big[{\bf B}_1,{\bf A}_k\big]\Big)\zeta^{2n-k},\ \ \ \zeta\in\mathbb{C},
\end{equation*}
and therefore, after matching powers in $\zeta$, first to order $\mathcal{O}(\zeta^{2n})$,
\begin{equation}\label{z80}
	A_1^{12}=-\im U,\ \ \ A_1^{21}=\im V,
\end{equation}
followed by all orders $\mathcal{O}(\zeta^{2n-k})$ for $k=1,\ldots,2n-1$,
\begin{equation}\label{z81}
	 \begin{cases}\displaystyle\frac{\partial A_k^{11}}{\partial t}=-\im(UA_k^{21}+A_k^{12}V),\ \ \ \frac{\partial A_k^{12}}{\partial t}=-\im(A_{k+1}^{12}+UA_k^{22}-A_k^{11}U)\bigskip\\
	\displaystyle\frac{\partial A_k^{22}}{\partial t}=\im(VA_k^{12}+A_k^{21}U),\ \ \ \ \ \,\frac{\partial A_k^{21}}{\partial t}=\im(A_{k+1}^{21}+VA_k^{11}-A_k^{22}V)\end{cases},
\end{equation}
and finally the order $\mathcal{O}(\zeta^0)$,
\begin{equation}\label{z82}
	 \begin{cases}\displaystyle\frac{\partial A_{2n}^{11}}{\partial t}=-\im(UA_{2n}^{21}+A_{2n}^{12}V),\ \ \ \frac{\partial A_{2n}^{12}}{\partial t}=-\im(UA_{2n}^{22}-A_{2n}^{11}U+\im M_tU)\bigskip\\
	\displaystyle\frac{\partial A_{2n}^{22}}{\partial t}=\im(VA_{2n}^{12}+A_{2n}^{21}U),\ \ \ \ \ \,\frac{\partial A_{2n}^{21}}{\partial t}=\im(VA_{2n}^{11}-A_{2n}^{22}V-\im VM_t)\end{cases}.
\end{equation}
This completes our proof of the Lemma.
\end{proof}
The system \eqref{z81},\eqref{z82} allows us to explicitly integrate the diagonal equations for $A_k^{11}$ and $A_k^{22}$ with the help of Corollary \ref{z:cor3} and Proposition \ref{z:prop5}.
\begin{lem}\label{z:lem11} We have on $\mathcal{H}_1$ for $k=1,2,\ldots,2n$,
\begin{equation*}
	A_k^{11}=-\im\sum_{j=1}^{k-1}\left(A_j^{11}A_{k-j}^{11}+A_j^{12}A_{k-j}^{21}\right)\ \ \ \ \ \textnormal{and}\ \ \ \ \ A_k^{22}=\im\sum_{j=1}^{k-1}\left(A_j^{22}A_{k-j}^{22}+A_j^{21}A_{k-j}^{12}\right),
\end{equation*}
and thus in particular $A_1^{11}=A_1^{22}=0$.
\end{lem}
\begin{proof} Inspired by \cite[$(18)$]{WE} we first compute the composition operator ${\bf C}(\zeta)={\bf A}(\zeta){\bf A}(\zeta)$ on $\mathcal{H}_2$ from \eqref{z74},
\begin{equation}\label{z83}
	{\bf C}(\zeta)=\sum_{k=0}^{4n}\bigg(\sum_{j=0}^k{\bf A}_j{\bf A}_{k-j}\bigg)\zeta^{4n-k}+\sum_{k=0}^{2n}\big({\bf A}_k\widehat{{\bf A}}_{2n}+\widehat{{\bf A}}_{2n}{\bf A}_k\big)\zeta^{2n-k}+\widehat{{\bf A}}_{2n}\widehat{{\bf A}}_{2n}\equiv\sum_{k=0}^{4n}{\bf C}_k\zeta^{4n-k},
\end{equation}
and then use the compatibility constraint \eqref{z79},
\begin{equation}\label{z84}
	\frac{\partial{\bf C}}{\partial t}(\zeta)=\big\{{\bf A}(\zeta),{\bf B}_0\big\}+\big[{\bf B}(\zeta),{\bf C}(\zeta)\big],
\end{equation}
where the curly brackets indicate the anticommutator). Matching powers $\mathcal{O}(\zeta^{4n-k})$ for $k=0,\ldots,2n-1$ in \eqref{z84} while using \eqref{z83} and \eqref{z74} yields at once
\begin{equation}\label{z85}
	\begin{cases}\displaystyle \frac{\partial C_k^{11}}{\partial t}=-\im(UC_k^{21}+C_k^{12}V),\ \ \ \frac{\partial C_k^{12}}{\partial t}=-\im(C_{k+1}^{12}+UC_k^{22}-C_k^{11}U)\bigskip\\
	\displaystyle\frac{\partial C_k^{22}}{\partial t}=\im(VC_k^{12}+C_k^{21}U),\ \ \ \ \ \frac{\partial C_k^{21}}{\partial t}=\im(C_{k+1}^{21}+VC_k^{11}-C_k^{22}V)
	\end{cases},\ \ \ k=0,\ldots,2n-1,
\end{equation}
and
\begin{equation}\label{z86}
	\frac{\partial C_{2n}^{11}}{\partial t}=-\im(B_0^{11}+UC_{2n}^{21}+C_{2n}^{12}V),\ \ \ \ \ \ \ \ \ \ \frac{\partial C_{2n}^{22}}{\partial t}=\im(VC_{2n}^{12}+C_{2n}^{21}U),
\end{equation}
for some of the coefficient operator entries of ${\bf C}_k$ with $k=0,1,\ldots,2n$. In turn, system \eqref{z85} shows that the operators ${\bf C}_k$ are trivial for $k=1,\ldots,2n-1$ and
\begin{equation*}
	C_{2n}^{12}=C_{2n}^{21}=C_{2n}^{22}=0.
\end{equation*}
Indeed, using \eqref{z80},\eqref{z81} and Corollary \ref{z:cor3} we find that $A_1^{11}=A_1^{22}=0$ on $\mathcal{H}_1$ and so by direct computation from \eqref{z83},
\begin{equation*}
	{\bf C}_1=\sum_{j=0}^1{\bf A}_j{\bf A}_{1-j}\stackrel[\eqref{z80}]{\eqref{z67}}{=}{\bf 0}\ \ \ \textnormal{on}\ \mathcal{H}_2.
\end{equation*}
Hence, proceeding inductively and assuming ${\bf C}_j={\bf 0}$ for all $j=1,\ldots,k$ with arbitrary $k\in\{1,\ldots,2n-2\}$ we first use the off-diagonal equations in \eqref{z85} to conclude that
\begin{equation*}
	C_{j+1}^{12}=\im\frac{\partial C_j^{12}}{\partial t}-UC_j^{22}+C_j^{11}U=0,\ \ \ \ C_{j+1}^{21}=-\im\frac{\partial C_j^{21}}{\partial t}-VC_j^{11}+C_j^{22}V=0,
\end{equation*}
by induction hypothesis. Hence, again by \eqref{z85}, this time through the diagonal equations,
\begin{equation*}
	\frac{\partial C_{j+1}^{11}}{\partial t}=-\im(UC_j^{21}+C_j^{12}V)=0,\ \ \ \ \frac{\partial C_{j+1}^{22}}{\partial t}=\im(VC_j^{12}+C_j^{21}U)=0,
\end{equation*}
yielding $C_{j+1}^{11}=C_{j+1}^{22}=0$ on $\mathcal{H}_1$ by Corollary \ref{z:cor3} and Proposition \ref{z:prop5} since ${\bf C}_k=\sum_{j=0}^k{\bf A}_j{\bf A}_{k-j}$ for $k=1,\ldots,2n-1$ by \eqref{z83} vanishes uniformly as $t\rightarrow+\infty$. Moving ahead the proclaimed vanishing of $C_{2n}^{12},C_{2n}^{21}$ and $C_{2n}^{22}$ follows now from the off-diagonal equations in \eqref{z85} as well as the second equation in \eqref{z86}. We are now prepared to prove the stated formul\ae\,for $A_k^{11}$ and $A_k^{22}$. First, from \eqref{z83},
\begin{align*}
	{\bf C}_{2n}=&\,\sum_{j=1}^{2n-1}{\bf A}_j{\bf A}_{2n-j}+{\bf A}_0({\bf A}_{2n}+\widehat{{\bf A}}_{2n})+({\bf A}_{2n}+\widehat{{\bf A}}_{2n}){\bf A}_0,\\
	{\bf C}_k=&\,\sum_{j=1}^{k-1}{\bf A}_j{\bf A}_{k-j}+{\bf A}_0{\bf A}_k+{\bf A}_k{\bf A}_0,\ \ \ \ k=2,\ldots,2n-1,
\end{align*}
so reading off $(22)$-entries, with the aforementioned fact that $C_k^{22}=0$ for $k=1,\ldots,2n$ and with \eqref{z67},
\begin{equation}\label{z87}
	0=C_k^{22}=\sum_{j=1}^{k-1}\big(A_j^{22}A_{k-j}^{22}+A_j^{21}A_{k-j}^{12}\big)+\im A_k^{22},\ \ \ k=2,\ldots,2n.
\end{equation}
Combined with the $(22)$-equation in \eqref{z81}, identity \eqref{z80} and again Corollary \ref{z:cor3}, \eqref{z87} yields the desired equation for $A_k^{22},k=1,\ldots,2n$. By similar logic
\begin{equation*}
	0=C_k^{11}=\sum_{j=1}^{k-1}\big(A_j^{11}A_{k-j}^{11}+A_j^{12}A_{k-j}^{21}\big)-\im A_k^{11},\ \ \ \ k=2,\ldots,2n-1;\ \ \ \frac{\partial A_1^{11}}{\partial t}\stackrel[\eqref{z81}]{\eqref{z80}}{=}0
\end{equation*}
which confirms the stated equation for $A_k^{11}$ provided $k=1,\ldots,2n-1$ after another application of Corollary \ref{z:cor3}. The $A_{2n}^{11}$ formula has to be treated slightly different since by \eqref{z86}, after our above workings,
\begin{equation*}
	\frac{\partial C_{2n}^{11}}{\partial t}=-\im B_0^{11},
\end{equation*}
and in addition
\begin{equation*}
	C_{2n}^{11}=\sum_{j=1}^{2n-1}\big(A_j^{11}A_{2n-j}^{11}+A_j^{12}A_{2n-j}^{21}\big)-\im A_{2n}^{11}-\im\widehat{A}_{2n}^{11}.
\end{equation*}
However $\frac{\partial}{\partial t}\widehat{A}_{2n}^{11}=B_0^{11}$, so the last two identities yield
\begin{equation*}
	0=\frac{\partial}{\partial t}\left[\sum_{j=1}^{2n-1}\big(A_j^{11}A_{2n-j}^{11}+A_j^{12}A_{2n-j}^{21}\big)-\im A_{2n}^{11}\right]
\end{equation*}
and hence after $t$-integration and an application of Corollary \ref{z:cor3} indeed the stated identity for $A_{2n}^{11}$. This concludes our proof of the Lemma.
\end{proof}
Observe that Lemma \ref{z:lem10} and \ref{z:lem11} allow us to express all operators ${\bf A}_k,k=1,\ldots,2n$ in \eqref{z74} in terms of $U,V$ given in \eqref{z76} and their $t$-derivatives. Indeed we have derived the following recursive recipe
\begin{cor}\label{z:cor4} On $\mathcal{H}_2$,
\begin{equation}\label{z88}
	{\bf A}_1=\begin{bmatrix}0 & -\im U\smallskip\\ \im V & 0\end{bmatrix},\ \ \ \ {\bf A}_{k+1}=\begin{bmatrix}A_{k+1}^{11} & A_{k+1}^{12}\smallskip\\
	A_{k+1}^{21} & A_{k+1}^{22}\end{bmatrix},\ \ k=1,\ldots,2n-1,
\end{equation}
where
\begin{equation}\label{z89}
	\begin{cases}\displaystyle A_{k+1}^{12}=\im\frac{\partial A_k^{12}}{\partial t}-UA_k^{22}+A_k^{11}V\smallskip\\
	\displaystyle A_{k+1}^{21}=-\im\frac{\partial A_k^{21}}{\partial t}-VA_k^{11}+A_k^{22}V\end{cases},\ \ \ \begin{cases}\displaystyle A_{k+1}^{11}=-\im\sum_{j=1}^k(A_j^{11}A_{k+1-j}^{11}+A_j^{12}A_{k+1-j}^{21})\smallskip\\
	\displaystyle A_{k+1}^{22}=\im\sum_{j=1}^k(A_j^{22}A_{k+1-j}^{22}+A_j^{21}A_{k+1-j}^{12})
	\end{cases}.
\end{equation}
Moreover,
\begin{equation}\label{z90}
	\frac{\partial A_{2n}^{12}}{\partial t}=-\im(UA_{2n}^{22}-A_{2n}^{11}U+\im M_tU),\ \ \ \ \frac{\partial A_{2n}^{21}}{\partial t}=\im(VA_{2n}^{11}-A_{2n}^{22}V-\im VM_t),
\end{equation}
and \eqref{z82},\eqref{z88},\eqref{z89},\eqref{z90} combined together yield the following $(2n)$-th order coupled, operator-valued system for $U$ and $V$,
\begin{equation}\label{z91}
	\mathcal{D}^{2n}\begin{bmatrix}-\im U\\\,\,\,\,\im V\end{bmatrix}=\begin{bmatrix}\,\,\,\,\im M_tU\,\\ -\im VM_t\,\end{bmatrix},\ \ \ \ \ \ \ \ \mathcal{D}\begin{bmatrix}A\smallskip\\ B\end{bmatrix}:=\begin{bmatrix}\im\frac{\partial A}{\partial t}-\im UD_t^{-1}(VA+BU)-\im D_t^{-1}(UB+AV)U\smallskip\\
	-\im\frac{\partial B}{\partial t}+\im VD_t^{-1}(UB+AV)+\im D_t^{-1}(VA+BU)V\end{bmatrix}
\end{equation}
where $\mathcal{D}$ acts entrywise on operators $A$ and $B$ on $\mathcal{H}_1$ and $D_t^{-1}$ denotes the formal $t$-antiderivative.
\end{cor}
\begin{proof} By \eqref{z81} and \eqref{z89},
\begin{equation}\label{z92}
	\mathcal{D}\begin{bmatrix}A_k^{12}\smallskip\\ A_k^{21}\end{bmatrix}=\begin{bmatrix}A_{k+1}^{12}\smallskip\\ A_{k+1}^{21}\end{bmatrix},\ \ k=1,\ldots,2n-1
\end{equation}
since $D_t^{-1}(VA_k^{12}+A_k^{21}U)=-\im A_k^{22}$ and $D_t^{-1}(UA_k^{21}+A_k^{12}V)=\im A_k^{11}$. Likewise, by \eqref{z90} and \eqref{z82},
\begin{equation}\label{z93}
	\mathcal{D}\begin{bmatrix}A_{2n}^{12}\smallskip\\ A_{2n}^{21}\end{bmatrix}=\begin{bmatrix}\,\,\,\,\im M_tU\,\smallskip\\ -\im VM_t\,\end{bmatrix},
\end{equation}
where we use $D_t^{-1}(VA_{2n}^{12}+A_{2n}^{21}U)=-\im A_{2n}^{22}$ and $D_t^{-1}(UA_{2n}^{21}+A_{2n}^{12}V)=\im A_{2n}^{11}$. Hence, iterating \eqref{z92},\eqref{z93} with the initial data \eqref{z88} we arrive at the desired system \eqref{z91} which does not contain any antiderivative terms because of the iterative formul\ae\,for $A_k^{11}$ and $A_k^{22}$ written in \eqref{z89}.
\end{proof}
The operator-valued ODE system \eqref{z91} for $U$ and $V$ naturally encodes the desired integro-differential Painlev\'e-II hierarchy \eqref{i17}. This final step in our proof of Theorem \ref{itheo1} is worked out in Section \ref{zsec5} below.

\section{The integro-differential Painlev\'e-II hierarchy - proof of Theorem \ref{itheo1}, part 2}\label{zsec5}

We will now evaluate system \eqref{z91} for the underlying kernels which will lead us to \eqref{i17}. First, we collect the following crucial symmetry constraints.
\begin{lem}\label{z:lem12} Let $k\in\{1,\ldots,2n\}$, then $A_k^{12}(x,y)$ and $A_k^{21}(y,x)$ are $y$-independent and we have
\begin{equation}\label{z94}
	A_k^{12}(x,y)=(-1)^kA_k^{21}(y,x),\ \ \ \ (x,y)\in\mathbb{R}^2.
\end{equation}
\end{lem}
\begin{proof} We have $A_1^{12}(x,y)=-\im U(x,y)$ and $A_1^{21}(y,x)=\im V(y,x)$ by \eqref{z88}. Using Lemma \ref{z:lem8}, we thus obtain \eqref{z94} for $k=1$ and since
\begin{equation*}
	U(x,y)=\lambda^{\frac{1}{2}}\int_{\Sigma}n_1(\eta|x)k_2(\eta|y)\,\d\eta,
\end{equation*}
the $y$-independence of $A_1^{12}(x,y)$ follows from \eqref{z33}. But $U(x,y)=V(y,x)$, so the $y$-independence of $A_1^{21}(y,x)$ follows similarly. Proceeding inductively, we assume that the claims have been proven for $k\in\{1,\ldots,m\}$ and some $1\leq m\leq 2n-1$. Since by \eqref{z92},
\begin{align*}
	A_{k+1}^{12}(x,y)=&\,\im\frac{\partial A_k^{12}}{\partial t}(x,y)-\im\int_{\mathbb{R}}U(x,z)\int^t\!\int_{\mathbb{R}}\big(V(z,w)A_k^{12}(w,y)+A_k^{21}(z,w)U(w,y)\big)\d\sigma(w)\d t\,\d\sigma(z)\\
	&\,-\im\int^t\!\int_{\mathbb{R}}\int_{\mathbb{R}}\big(U(x,z)A_k^{21}(z,w)+A_k^{12}(x,z)V(z,w)\big)U(w,y)\,\d\sigma(z)\d\sigma(w)\d t
\end{align*}
we see that $A_{k+1}^{12}(x,y)$ is $y$-independent by the induction hypothesis and base case. Moreover, using explicitly the induction hypothesis in the form $A_k^{12}(x,y)=(-1)^kA_k^{21}(y,x)$, we obtain
\begin{align}
	A_{k+1}^{12}(x,y)=&\,(-1)^{k+1}\bigg[-\im\frac{\partial A_k^{21}}{\partial t}(y,x)+\im\int_{\mathbb{R}}U(x,z)\int^t\!\int_{\mathbb{R}}\big(V(z,w)A_k^{21}(y,w)+A_k^{12}(w,z)U(w,y)\big)\,\times\label{z95}\\
	&\,\times\d\sigma(w)\d t\,\d\sigma(z)+\im \int^t\!\int_{\mathbb{R}}\int_{\mathbb{R}}\big(U(x,z)A_k^{12}(w,z)+A_k^{21}(z,x)V(z,w)\big)U(w,y)\,\d\sigma(z)\d\sigma(w)\d t\bigg].\nonumber
\end{align}
On the other hand, \eqref{z92} also says
\begin{align*}
	A_{k+1}^{21}(x,y)=&\,-\im\frac{\partial A_k^{21}}{\partial t}(x,y)+\im\int_{\mathbb{R}}V(x,z)\int^t\!\int_{\mathbb{R}}\big(U(z,w)A_k^{21}(w,y)+A_k^{12}(z,w)V(w,y)\big)\d\sigma(w)\d t\,\d\sigma(z)\\
	&\,+\im\int^t\int_{\mathbb{R}}\int_{\mathbb{R}}\big(V(x,z)A_k^{12}(z,w)+A_k^{21}(x,z)U(z,w)\big)V(w,y)\,\d\sigma(w)\d\sigma(z)\d t,
\end{align*}
and thus $A_{k+1}^{21}(x,y)$ is $x$-independent by the induction hypothesis and base case. Finally, relabelling the integration variables $z\leftrightarrow w$ in the last equality and using the induction base case six times in the form $U(x,y)=V(y,x)$ we see at once that with \eqref{z95},
\begin{equation*}
	A_{k+1}^{21}(x,y)=(-1)^{k+1}A_{k+1}^{12}(y,x),\ \ \ (x,y)\in\mathbb{R}^2.
\end{equation*}
This completes our proof.
\end{proof}
The conditions laid out in Lemma \ref{z:lem12} tell us that both equations in \eqref{z91}, when evaluated on the kernel level, lead to the same dynamical system in $t$ which only depends on one additional variable $x$, say - modulo their secondary $(\lambda,n)$-dependence. Indeed, introducing $u(t|x;n,\lambda)\equiv u(t|x):=U(x,x)$ (recall $U$ and $V$ both depend on $t,\lambda$ and $n$) we have by Lemma \ref{z:lem8} and \ref{z:lem12} that $u(t|x)=U(x,x)=U(x,y)=V(y,x)=V(x,x)$ for all $(t,x)\in\mathbb{R}^2$. More generally, setting 
\begin{equation*}
	a_k(t|x):=A_k^{12}(x,x)=(-1)^kA_k^{21}(x,x),\ \ \ \ k\in\{1,\ldots,2n\},
\end{equation*} 
the recursion \eqref{z81} for $A_k^{12}$ (replacing $A_k^{22}$ and $A_k^{11}$ with their antiderivative expressions) is equivalent to
\begin{equation}\label{z96}
	a_{k+1}(t|x)=\begin{cases}
	(\mathcal{L}_{+}^ua_k)(t|x),&k\equiv 0\mod 2\smallskip\\
	(\mathcal{L}_-^ua_k)(t|x),&k\equiv 1\mod 2
	\end{cases},\ \ \ \ \ k=1,2\ldots,2n-1;\ \ \ \ \ \ \ \ \ a_1(t|x):=-\im u(t|x)
\end{equation}
with the operators $\mathcal{L}_{\pm}^u,[\cdot,\cdot],\{\cdot,\cdot\}$ and the weighted bilinear form $\langle\cdot,\cdot\rangle$ of Definition \ref{idef1}. In deriving \eqref{z96} we use \eqref{z94} throughout and recall that $\d\sigma$ is a probability measure. Furthermore, next to the $2n-1$ equations \eqref{z96}, we also find from \eqref{z82} that 
\begin{equation*}
	-(t+x)a_1(t|x)=(\mathcal{L}_+^ua_{2n})(t|x)
\end{equation*}
and thus iteratively, with \eqref{z96},
\begin{equation}\label{z97}
	-(t+x)a_1(t|x)=\big((\mathcal{L}_+^u\mathcal{L}_-^u)^na_1\big)(t|x)\ \ \ \ \ \Rightarrow\ \ \ \ \ \ (t+x)u(t|x)=-\big((\mathcal{L}^u_+\mathcal{L}^u_-)^nu\big)(t|x).
\end{equation}
It now remains to record a representation formula for $D_n(t,\lambda)$ in terms of $u(t|x)$, see \eqref{z98} below, and its boundary behavior at $t=+\infty$. Once established, the same formula together with \eqref{z97} and Remark \ref{zrem:1} completes the proof of Theorem \ref{itheo1}.
\begin{lem}\label{z:lem13} For every $(t,\lambda,n)\in\mathbb{R}\times\overline{\mathbb{D}_1(0)}\times\mathbb{N}$,
\begin{equation}\label{z98}
	D_n(t,\lambda)=\exp\left[-\int_t^{\infty}(s-t)\left(\int_{\mathbb{R}}u^2(s|x)\d\sigma(x)\right)\d s\right].
\end{equation}
where $u(t|x)\equiv u(t|x;n,\lambda)$ solves the dynamical system \eqref{z97} such that $u(t|x)\sim\lambda^{\frac{1}{2}}\textnormal{Ai}_n(t+x)$ as $t\rightarrow+\infty$, pointwise in $x\in\mathbb{R}$.
\end{lem}
\begin{proof} By Lemma \ref{z:lem9},
\begin{align}
	\frac{\partial^2}{\partial t^2}\ln D_n(t,\lambda)=&\,-\tr_{\mathcal{H}_1}(UV)=-\int_{\mathbb{R}}\int_{\mathbb{R}}U(x,y)V(y,x)\,\d\sigma(y)\,\d\sigma(x)\stackrel{\eqref{z56}}{=}-\int_{\mathbb{R}}\int_{\mathbb{R}}U^2(x,y)\,\d\sigma(x)\d\sigma(y)\nonumber\\
	\stackrel{\eqref{z33}}{=}&\,-\int_{\mathbb{R}}\int_{\mathbb{R}}U^2(x,x)\,\d\sigma(x)\d\sigma(y)=-\int_{\mathbb{R}}u^2(t|x)\,\d\sigma(x),\label{z99}
\end{align}
given that $k_2(\zeta|y)$ is $y$-independent and $\d\sigma$ a probability measure. However,
\begin{align*}
	u(t|x)=\lambda^{\frac{1}{2}}&\,\int_{\Sigma}n_1(\eta|x)k_2(\eta|x)\,\d\eta\stackrel{\eqref{z33}}{=}\frac{\lambda^{\frac{1}{2}}}{2\pi}\int_{\Gamma_{\beta}}\e^{-\im\psi_n(\eta,t+x)}\,\d\eta+\lambda\int_{\Sigma}\big(C_{t,n}^{\ast}m_1(\cdot|x)\big)(\eta)k_2(\eta|x)\,\d\eta\\
	&\,+\lambda^{\frac{1}{2}}\int_{\Sigma}\Big[n_1(\eta|x)-m_1(\eta|x)-\lambda^{\frac{1}{2}}\big(C_{t,n}^{\ast}m_1(\cdot|x)\big)(\eta)\Big]k_2(\eta|x)\,\d\eta,
\end{align*}
so by \eqref{z17} indeed $u(t|x)\sim\lambda^{\frac{1}{2}}\textnormal{Ai}_n(t+x)$ as $t\rightarrow+\infty$ once we estimate the two remaining integrals involving $m_1(\cdot|x)$ as in our proof of Corollary \ref{z:cor3}. All together, \eqref{z98} follows from \eqref{z99} after integration since $u(t|x)\sim\lambda^{\frac{1}{2}}\textnormal{Ai}_n(t+x)$ yields $\int_{\mathbb{R}}u^2(t|x)\d\sigma(x)\rightarrow 0$ exponentially fast as $t\rightarrow+\infty$ because of \eqref{z3} and \eqref{i13}. This completes our proof of the Lemma.
\end{proof}
\begin{rem}\label{zrem:1} The smoothness of $t\mapsto u(t|x;n,\lambda)$ for fixed $(x,\lambda,n)\in\mathbb{R}\times\overline{\mathbb{D}_1(0)}\times\mathbb{N}$ follows from the identity
\begin{equation*}
	u(t|x)=U(x,x)=\lambda^{\frac{1}{2}}\int_{\Sigma}n_1(\eta|x)k_2(\eta|x)\d\eta,
\end{equation*}
from the analyticity of \eqref{z33} and analyticity of the resolvent, compare \eqref{z39}. Likewise, the unique solvability of the boundary value problem \eqref{i17} for $(t,x,\lambda,n)\in\mathbb{R}^2\times\overline{\mathbb{D}_1(0)}\times\mathbb{N}$ follows from \eqref{z47} and the fact that RHP \ref{zmaster} is uniquely solvable for $(t,\lambda,n)\in\mathbb{R}\times\overline{\mathbb{D}_1(0)}\times\mathbb{N}$.
\end{rem}

\section{The integro-differential mKdV hierarchy - proof of Theorem \ref{itheo2}}\label{zsec6}

In deriving the integro-differential PDE hierarchy \eqref{i21} we first take a step back, return to \eqref{i15} and introduce another parameter $\tau\in\mathbb{R}_+$. Precisely we consider the Fredholm determinant
\begin{equation}\label{kd1}
	D_n(t,\lambda,\tau):=\det(I-\lambda K_{t,n}^{\tau}\upharpoonright_{L^2(\mathbb{R}_+)}),\ \ \ \ \ (t,\lambda,\tau,n)\in\mathbb{R}\times\mathbb{C}\times\mathbb{R}_+\times\mathbb{N},
\end{equation}
with kernel
\begin{equation*}
	K_{t,n}^{\tau}(x,y):=\int_{\mathbb{R}}\textnormal{Ai}_n(x+z+t)\textnormal{Ai}_n(z+y+t)w(\tau z)\,\d z,
\end{equation*}
and weight $w:\mathbb{R}\rightarrow\mathbb{R}_+$ of general type \eqref{i13}. Observe that the Fermi factor \eqref{i8} yields a standard example of the parametric setup \eqref{kd1} with $\tau=\alpha$. Given that the $\tau$-rescaling does not alter any of the results in Section \ref{zsec1} (in particular \eqref{z12} still holds after substitution) we then apply the methods of Sections \ref{zsec2} and \ref{zsec3} to \eqref{kd1}, en route verifying the below result.
\begin{lem}\label{kd:lem1} For every $(t,\lambda,\tau,n)\in\mathbb{R}\times\mathbb{C}\times\mathbb{R}_+\times\mathbb{N}$,
\begin{equation*}
	D_n(t,\lambda,\tau)=\det(I-\lambda^{\frac{1}{2}}C_{t,n}^{\tau}\upharpoonright_{L^2(\Sigma)}),
\end{equation*}
where $C_{t,n}^{\tau}:L^2(\Sigma)\rightarrow L^2(\Sigma)$ with $\Sigma=\mathbb{R}\sqcup\Gamma_{\beta}\sqcup\Gamma_{\alpha}$ shown in Figure \ref{figz:3} is trace class and has kernel
\begin{align*}
	(\xi-\eta)C_{t,n}^{\tau}(\xi,\eta):=\frac{1}{2\pi}\int_{\mathbb{R}}\Big(\e^{\frac{\im}{2}(\phi_n^{\tau}(\xi,2t\tau+2z)-\phi_n^{\tau}(\eta,2t\tau+2z))}&\,\chi_{\Gamma_{\alpha}}(\xi)\chi_{\Gamma_{\beta}}(\eta)\\
	&\,+\e^{-\frac{\im}{2}(\phi_n^{\tau}(\xi,0)-\phi_n^{\tau}(\eta,0))}\chi_{\Gamma_{\beta}}(\xi)\chi_{\Gamma_{\alpha}}(\eta)\Big)\d\sigma(z),
\end{align*}
defined in terms of $\phi_n^{\tau}(\lambda,z):=\frac{\tau^{2n+1}}{2n+1}\lambda^{2n+1}+z\lambda$.
\end{lem}
\begin{proof} Replacing all contours in Section \ref{zsec2} according to the admissible rule $\Gamma_{\alpha}\mapsto\tau\Gamma_{\alpha}$ and $\Gamma_{\beta}\mapsto\tau\Gamma_{\beta}$, compare Lemma \ref{zlem:1}, we obtain at once 
\begin{equation*}
	D_n(t,\lambda,\tau)=\det(I-\lambda J_{t,n}^{\tau}\upharpoonright_{L^2(\mathbb{R})}),\ \ \ \ \ (t,\lambda,\tau,n)\in\mathbb{R}\times\mathbb{C}\times\mathbb{R}_+\times\mathbb{N},
\end{equation*}
where $J_{t,n}^{\tau}:=A_{t,n}^{\tau}B_n:L^2(\mathbb{R})\rightarrow L^2(\mathbb{R})$ equals the composition of the Hilbert-Schmidt transformations $A_{t,n}^{\tau}:L^2(\tau\Gamma_{\beta})\rightarrow L^2(\mathbb{R})$ and $B_n:L^2(\mathbb{R})\rightarrow L^2(\tau\Gamma_{\beta})$ with kernels
\begin{equation}\label{kd2}
	A_{t,n}^{\tau}(\alpha,\beta):=\frac{1}{2\pi}\frac{\e^{\frac{\im}{2}\psi_n(\alpha,2t)-\frac{\im}{2}(\beta,2t)}}{\alpha-\beta}\left[\int_{\mathbb{R}}\e^{\im z(\alpha-\beta)/\tau}\d\sigma(z)\right],\ \ \ \ \ B_n(\beta,\gamma):=\frac{1}{2\pi}\frac{\e^{-\frac{\im}{2}\psi_n(\beta,0)+\frac{\im}{2}\psi_n(\gamma,0)}}{\beta-\gamma}.
\end{equation}
Generalizing afterwards Proposition \ref{zprop:2} to $J_{t,n}^{\circ,\tau}:L^2(\tau\Gamma_{\alpha})\rightarrow L^2(\tau\Gamma_{\alpha})$ with
\begin{equation}\label{kd3}
	(J_{t,n}^{\circ,\tau})(\xi):=\int_{\tau\Gamma_{\alpha}}J_{t,n}^{\tau}(\xi,\eta)f(\eta)\,\d\eta,\ \ \ \ \ f\in L^2(\tau\Gamma_{\alpha}),
\end{equation}
we then derive the analogue of \eqref{z26} for \eqref{kd1} in terms of \eqref{kd3}. Once done, it then remains to check that the previous extension \eqref{z28} applies verbatim to our $\tau$-dependent setup (replacing again $\Gamma_{\alpha}\mapsto\tau\Gamma_{\alpha}$ and $\Gamma_{\beta}\mapsto\tau\Gamma_{\beta}$ throughout) and Lemma \ref{zlem:7} goes through as well with the sum kernel
\begin{equation*}
	\big(A_{t,n}^{\tau,\textnormal{ext}}+B_n^{\textnormal{ext}}\big)(\xi,\eta)=A_{t,n}^{\tau}(\xi,\eta)\chi_{\tau\Gamma_{\alpha}}(\xi)\chi_{\tau\Gamma_{\beta}}(\eta)+B_n(\xi,\eta)\chi_{\tau\Gamma_{\beta}}(\eta)\chi_{\tau\Gamma_{\alpha}}(\xi).
\end{equation*}
Finally, using \eqref{kd2} and a simple rescaling, we obtain the desired determinant equality with the indicated kernel for $C_{t,n}^{\tau}$ on $L^2(\Sigma)$. This completes our proof of the Lemma.
\end{proof}

We now proceed as in the second part of Section \ref{zsec3} and thus introduce the following operators.
\begin{definition} Let $M_i^{\tau}(\zeta)\otimes K_j^{\tau}(\zeta)\in\mathcal{I}(\mathcal{H}_1),i,j=1,2$ denote the $\Sigma\ni\zeta$-parametric family of rank one integral operators with kernels
\begin{equation*}
	\big(M_i^{\tau}(\zeta)\otimes K_j^{\tau}(\zeta)\big)(x,y):=m_i^{\tau}(\zeta|x)k_j^{\tau}(\zeta|y),\ \ \ \ \ \ x,y\in\mathbb{R},
\end{equation*}
defined in terms of the $\Sigma\ni\zeta$-parametric family of functions
\begin{equation*}
	k_1^{\tau}(\zeta|y):=\frac{1}{2\pi}\e^{\frac{\im}{2}\phi_n^{\tau}(\zeta,2t\tau+2y)}\chi_{\Gamma_{\alpha}}(\zeta),\ \ k_2^{\tau}(\zeta|y):=\frac{1}{2\pi}\e^{-\frac{\im}{2}\phi_n^{\tau}(\zeta,0)}\chi_{\Gamma_{\beta}}(\zeta),\ \ m_1^{\tau}(\zeta|x):=\e^{-\frac{\im}{2}\phi_n^{\tau}(\zeta,2t\tau+2x)}\chi_{\Gamma_{\beta}}(\zeta),
\end{equation*}
\begin{equation}\label{kd30}
	m_2^{\tau}(\zeta|x):=\e^{\frac{\im}{2}\phi_n^{\tau}(\zeta,0)}\chi_{\Gamma_{\alpha}}(\zeta).
\end{equation}
\end{definition}
In turn, the following $\tau$-dependent RHP generalizes RHP \ref{zmaster}.
\begin{problem}\label{kdmaster} Given $(t,\lambda,\tau,n)\in\mathbb{R}\times\overline{\mathbb{D}_1(0)}\times\mathbb{R}_+\times\mathbb{N}$, find ${\bf X}^{\tau}(\zeta)={\bf X}^{\tau}(\zeta;t,\lambda,n)\in\mathcal{I}(\mathcal{H}_2)$ such that
\begin{enumerate}
	\item[(1)] ${\bf X}^{\tau}(\zeta)=\mathbb{I}_2+{\bf X}_0^{\tau}(\zeta)$ and ${\bf X}_0^{\tau}(\zeta)\in\mathcal{I}(\mathcal{H}_2)$ with kernel ${\bf X}_0^{\tau}(\zeta|x,y)$ is analytic in $\mathbb{C}\setminus\Sigma$.
	\item[(2)] ${\bf X}^{\tau}(\zeta)$ admits continuous boundary values ${\bf X}_{\pm}^{\tau}(\zeta)\in\mathcal{I}(\mathcal{H}_2)$ on $\Sigma$ which satisfy ${\bf X}_+^{\tau}={\bf X}_-^{\tau}(\zeta){\bf G}^{\tau}(\zeta)$ with
	\begin{equation*}
		{\bf G}^{\tau}(\zeta)=\mathbb{I}_2+2\pi\im\lambda^{\frac{1}{2}}\begin{bmatrix}M_1^{\tau}(\zeta)\otimes K_1^{\tau}(\zeta) & M_1^{\tau}(\zeta)\otimes K_2^{\tau}(\zeta)\smallskip\\
		M_2^{\tau}(\zeta)\otimes K_1^{\tau}(\zeta) & M_2^{\tau}(\zeta)\otimes K_2^{\tau}(\zeta)\end{bmatrix},\ \ \ \ \zeta\in\Sigma.
	\end{equation*}
	\item[(3)] There exists $c=c(t,n)>0$ such that for $\zeta\in\mathbb{C}\setminus\Sigma$,
	\begin{equation*}
		\|{\bf X}_0^{\tau}(\zeta|x,y)\|\leq\frac{c\sqrt{|\lambda|}}{1+|\zeta|}\Delta^{-\frac{1}{4n}}\tau^{-\frac{2n+1}{4n}}\e^{-\frac{(-1)^n\tau\Delta}{2(2n+1)}(\tau\Delta)^{2n}}\e^{\Delta(|x|+|y|+\tau|t|)},\ \ \ \Delta=\textnormal{dist}(\Gamma_{\alpha},\mathbb{R})=\textnormal{dist}(\Gamma_{\beta},\mathbb{R})>0,
	\end{equation*}
	uniformly in $(x,y)\in\mathbb{R}^2$ and $(\lambda,\tau)\in\overline{\mathbb{D}_1(0)}\times\mathbb{R}_+$.
\end{enumerate}
\end{problem}
The last problem is uniquely solvable by the proof methods of Lemma \ref{z:lem7} and Theorem \ref{z:theo1}, we only summarize the relevant results for our upcoming analysis without repeating the necessary proofs.
\begin{theo}\label{kd:theo1} For every $(t,\lambda,\tau,n)\in\mathbb{R}\times\overline{\mathbb{D}_1(0)}\times\mathbb{R}_+\times\mathbb{N}$, the integral operator
\begin{equation*}
	{\bf X}^{\tau}(\zeta)=\mathbb{I}_2+\lambda^{\frac{1}{2}}\int_{\Sigma}\begin{bmatrix}N_1^{\tau}(\eta)\otimes K_1^{\tau}(\eta) & N_1^{\tau}(\eta)\otimes K_2^{\tau}(\eta)\smallskip\\
	N_2^{\tau}(\eta)\otimes K_1^{\tau}(\eta) & N_2^{\tau}(\eta)\otimes K_2^{\tau}(\eta)\end{bmatrix}\frac{\d\eta}{\eta-\zeta},\ \ \ \ \ \ \zeta\in\mathbb{C}\setminus\Sigma,
\end{equation*}
is the unique solution of RHP \ref{kdmaster}, where $N_i^{\tau}(\eta)$ are the operators on $\mathcal{H}_1$ which multiply by the functions $n_i^{\tau}(\eta|x)$ determined through the equation
\begin{equation*}
	\big(I-\lambda^{\frac{1}{2}}C_{t,n}^{\tau\ast}\upharpoonright_{L^2(\Sigma)}\big)n_i^{\tau}(\cdot|x)=m_i^{\tau}(\cdot|x),\ \ \ \ i=1,2,
\end{equation*}
with $x\in\mathbb{R}$ and the real adjoint $C_{t,n}^{\tau\ast}$ of $C_{t,n}^{\tau}$. Moreover,
\begin{equation*}
	\big({\bf X}^{\tau}(\zeta)\big)^{-1}=\mathbb{I}_2-\lambda^{\frac{1}{2}}\int_{\Sigma}\begin{bmatrix}M_1^{\tau}(\eta)\otimes L_1^{\tau}(\eta) & M_1^{\tau}(\eta)\otimes L_2^{\tau}(\eta)\smallskip\\
	M_2^{\tau}(\eta)\otimes L_1^{\tau}(\eta) & M_2^{\tau}(\eta)\otimes L_2^{\tau}(\eta)\end{bmatrix}\frac{\d\eta}{\eta-\zeta},\ \ \ \ \zeta\in\mathbb{C}\setminus\Sigma,
\end{equation*}
where $L_i^{\tau}(\eta)$ are the integral operators on $\mathcal{H}_1$ with kernel $\ell_i^{\tau}(\eta|y)$ determined from the equation
\begin{equation*}
	\big(I-\lambda^{\frac{1}{2}}C_{t,n}^{\tau}\upharpoonright_{L^2(\Sigma)})\ell_i^{\tau}(\cdot |y)=k_i^{\tau}(\cdot |y),\ \ i=1,2,\ \ \ y\in\mathbb{R}.
\end{equation*}
In addition, we have the representation formula ${\bf N}^{\tau}(\zeta)={\bf X}^{\tau}(\zeta){\bf M}^{\tau}(\zeta),\zeta\in\Sigma$ with the vector-valued operators
\begin{equation*}
	{\bf N}^{\tau}(\zeta):=\big[N_1^{\tau}(\zeta),N_2^{\tau}(\zeta)\big]^{\top},\ \ \ \ {\bf M}^{\tau}(\zeta):=\big[M_1^{\tau}(\zeta),M_2^{\tau}(\zeta)\big]^{\top},
\end{equation*}
and equations \eqref{z54},\eqref{z55}, \eqref{z56},\eqref{z61} carry over to the $\tau$-modified setup with obvious $\tau$-superscript modifications.
\end{theo}
With Theorem \ref{kdmaster} at hand we now proceed as in Sections \ref{zsec4} and \ref{zsec5}: view $M_i^{\tau}(\zeta)$ and $N_i^{\tau}(\zeta)$ as integral operators on $\mathcal{H}_1$ with appropriate distributional kernels as in the beginning of Section \ref{zsec4} and introduce the following two mKdV variables\footnote{We follow the standard mKdV convention in denoting the time variable with $t_{2n+1}$, cf. \cite[$(3.3)$]{CM}. The spatial variable in the ordinary mKdV is typically $x$ but for us $t_1$. The variable $x$ enters in the evaluation of the bilinear form, see Definition \ref{idef2}.}
\begin{equation}\label{kd4}
	t_1:=\tau t\in\mathbb{R},\ \ \ \ \ \ \ \ \ \ t_{2n+1}:=\frac{\tau^{2n+1}}{2n+1}\in\mathbb{R}_+.
\end{equation}
In turn we find
\begin{equation*}
	\frac{\partial}{\partial t_{2n+1}}{\bf M}^{\tau}(\zeta|x,y)=\begin{bmatrix}-\frac{\im}{2}\zeta^{2n+1} & 0\smallskip\\
	0 & \frac{\im}{2}\zeta^{2n+1}\end{bmatrix}{\bf M}^{\tau}(\zeta|x,y),\ \ \ \ (\zeta,x,y)\in\Sigma\times\mathbb{R}^2,
\end{equation*}
or equivalently the operator identity
\begin{equation}\label{kd5}
	\frac{\partial}{\partial t_{2n+1}}{\bf M}(\zeta)=\big(\zeta^{2n+1}{\bf A}_0\big){\bf M}^{\tau}(\zeta),\ \ \ \ \zeta\in\Sigma,
\end{equation}
with ${\bf A}_0:\mathcal{H}_2\rightarrow\mathcal{H}_2$ as in \eqref{z67}. Similarly,
\begin{equation}\label{kd6}
	\frac{\partial}{\partial t_1}{\bf M}(\zeta)=\big(\zeta{\bf B}_0\big){\bf M}^{\tau}(\zeta),\ \ \ \zeta\in\Sigma
\end{equation}
where ${\bf B}_0:\mathcal{H}_2\rightarrow\mathcal{H}_2$ has the kernel \eqref{z69}. Combining \eqref{kd5} and \eqref{kd6} we deduce the below mKdV equivalent of Proposition \ref{z:prop4}.
\begin{prop}\label{kd:prop1} There exist $(t_1,\lambda,t_{2n+1},n)$-dependent, analytic in $\zeta\in\mathbb{C}$ integral operators ${\bf A}^{\tau}(\zeta),{\bf B}^{\tau}(\zeta)$ on $\mathcal{H}_2$ such that for every $\zeta\in\Sigma$ and $(t_1,\lambda,t_{2n+1},n)\in\mathbb{R}\times\overline{\mathbb{D}_1(0)}\times\mathbb{R}_+\times\mathbb{N}$,
\begin{equation}\label{kd60}
	\frac{\partial{\bf N}^{\tau}}{\partial t_{2n+1}}(\zeta)={\bf A}^{\tau}(\zeta){\bf N}^{\tau}(\zeta),\ \ \ \ \ \ \ \frac{\partial{\bf N}^{\tau}}{\partial t_1}(\zeta)={\bf B}^{\tau}(\zeta){\bf N}^{\tau}(\zeta).
\end{equation}
\end{prop}
\begin{proof} Using the aforementioned representation formula and \eqref{kd5},
\begin{equation}\label{kd7}
	\frac{\partial{\bf N}^{\tau}}{\partial t_{2n+1}}(\zeta)=\underbrace{\bigg[\frac{\partial{\bf X}^{\tau}}{\partial t_{2n+1}}(\zeta)\big({\bf X}^{\tau}(\zeta)\big)^{-1}+{\bf X}^{\tau}(\zeta)\big(\zeta^{2n+1}{\bf A}_0\big)\big({\bf X}^{\tau}(\zeta)\big)^{-1}\bigg]}_{=:{\bf A}^{\tau}(\zeta)}{\bf N}^{\tau}(\zeta),
\end{equation}
where ${\bf A}^{\tau}(\zeta)\in\mathcal{I}(\mathcal{H}_2)$ by Theorem \ref{kdmaster} and ${\bf A}^{\tau}(\zeta)$ is analytic for $\zeta\in\mathbb{C}\setminus\Sigma$ with continuous boundary values ${\bf A}_{\pm}^{\tau}(\zeta)\in\mathcal{I}(\mathcal{H}_2)$ on $\Sigma$. In fact,
\begin{align*}
	{\bf A}_+^{\tau}(\zeta)=\bigg[\frac{\partial{\bf X}_-^{\tau}}{\partial t_{2n+1}}(\zeta)G^{\tau}(\zeta)\,+&\,{\bf X}_-^{\tau}(\zeta)\frac{\partial{\bf G}^{\tau}}{\partial t_{2n+1}}(\zeta)\bigg]\big({\bf G}^{\tau}(\zeta)\big)^{-1}\big({\bf X}_-^{\tau}(\zeta)\big)^{-1}\\
	&\,\,+{\bf X}_-^{\tau}(\zeta){\bf G}^{\tau}(\zeta)\big(\zeta^{2n+1}{\bf A}_0\big)\big({\bf G}^{\tau}(\zeta)\big)^{-1}\big({\bf X}_-^{\tau}(\zeta)\big)^{-1},\ \ \ \zeta\in\Sigma,
\end{align*}
and through the commutator identity
\begin{equation*}
	\frac{\partial{\bf G}^{\tau}}{\partial t_{2n+1}}(\zeta)=\big[\zeta^{2n+1}{\bf A}_0,{\bf G}^{\tau}(\zeta)\big]\in\mathcal{I}(\mathcal{H}_2),\ \ \ \ \ \zeta\in\Sigma,
\end{equation*}
therefore
\begin{equation*}
	{\bf A}_+^{\tau}(\zeta)=\frac{\partial{\bf X}_-^{\tau}}{\partial t_{2n+1}}(\zeta)\big({\bf X}_-^{\tau}(\zeta)\big)^{-1}+{\bf X}_-^{\tau}(\zeta)\big(\zeta^{2n+1}{\bf A}_0\big)\big({\bf X}_-^{\tau}(\zeta)\big)^{-1}={\bf A}_-^{\tau}(\zeta),\ \ \ \zeta\in\Sigma.
\end{equation*}
This shows that ${\bf A}^{\tau}(\zeta)$ as defined in \eqref{kd7} is analytic for every $\zeta\in\mathbb{C}$. Quite similar,
\begin{equation}\label{kd8}
	\frac{\partial{\bf N}^{\tau}}{\partial t_1}(\zeta)=\underbrace{\bigg[\frac{\partial{\bf X}^{\tau}}{\partial t_1}(\zeta)\big({\bf X}^{\tau}(\zeta)\big)^{-1}+{\bf X}^{\tau}(\zeta)\big(\zeta {\bf B}_0\big)\big({\bf X}^{\tau}(\zeta)\big)^{-1}\bigg]}_{=:{\bf B}^{\tau}(\zeta)}{\bf N}^{\tau}(\zeta),
\end{equation}
where ${\bf B}^{\tau}(\zeta)$ is also analytic for $\zeta\in\mathbb{C}\setminus\Sigma$ by Theorem \ref{kdmaster} with continuous boundary values ${\bf B}_{\pm}^{\tau}(\zeta)\in\mathcal{I}(\mathcal{H}_2)$ on $\Sigma$ that satisfy
\begin{align*}
	{\bf B}_+^{\tau}(\zeta)=\bigg[\frac{\partial{\bf X}_-^{\tau}}{\partial t_1}(\zeta)G^{\tau}(\zeta)\,+&\,{\bf X}_-^{\tau}(\zeta)\frac{\partial{\bf G}^{\tau}}{\partial t_1}(\zeta)\bigg]\big({\bf G}^{\tau}(\zeta)\big)^{-1}\big({\bf X}_-^{\tau}(\zeta)\big)^{-1}\\
	&\,\,+{\bf X}_-^{\tau}(\zeta){\bf G}^{\tau}(\zeta)\big(\zeta{\bf B}_0\big)\big({\bf G}^{\tau}(\zeta)\big)^{-1}\big({\bf X}_-^{\tau}(\zeta)\big)^{-1},\ \ \ \zeta\in\Sigma.
\end{align*}
But in light of the commutator identity
\begin{equation*}
	\frac{\partial{\bf G}^{\tau}}{\partial t_1}(\zeta)=\big[\zeta{\bf B}_0,{\bf G}^{\tau}(\zeta)\big]\in\mathcal{I}(\mathcal{H}_2),\ \ \ \ \zeta\in\Sigma,
\end{equation*}
we find
\begin{equation*}
	{\bf B}_+^{\tau}(\zeta)=\frac{\partial{\bf X}_-^{\tau}}{\partial t_1}(\zeta)\big({\bf X}_-^{\tau}(\zeta)\big)^{-1}+{\bf X}_-^{\tau}(\zeta)\big(\zeta{\bf B}_0\big)\big({\bf X}_-^{\tau}(\zeta)\big)^{-1}={\bf B}_-^{\tau}(\zeta),\ \ \ \ \zeta\in\Sigma,
\end{equation*}
and so the analyticity of ${\bf B}(\zeta)$ for $\zeta\in\mathbb{C}$. This concludes our proof of the Proposition.
\end{proof}
In our next move we will explicitly compute ${\bf A}^{\tau}(\zeta)$ and ${\bf B}^{\tau}(\zeta)$ in terms of RHP data (this time RHP \ref{kdmaster}) as previously done in Proposition \ref{z:prop5} for the $\tau$-independent problem.
\begin{prop}\label{kd:prop2} We have
\begin{equation}\label{kd9}
	{\bf B}^{\tau}(\zeta)=\zeta{\bf B}_0+{\bf B}_1^{\tau},\ \ \ \ \ \ \ \ \ {\bf A}^{\tau}(\zeta)=\zeta^{2n+1}{\bf A}_0+\sum_{k=1}^{2n+1}{\bf A}_k^{\tau}\,\zeta^{2n+1-k}
\end{equation}
where the kernels of ${\bf B}_0$ and ${\bf B}_1^{\tau}$ are written in \eqref{z69} and \eqref{kd10} below. Likewise, the kernel of ${\bf A}_0$ appeared in \eqref{z67} and the entries of ${\bf A}_k^{\tau}$ are polynomials in $\int_{\Sigma}N_i^{\tau}(\eta)\otimes K_j^{\eta}\eta^m\d\eta$ and $\int_{\Sigma}M_i^{\tau}(\eta)\otimes L_j^{\tau}(\eta)\eta^m\d\eta$ with $m\in\mathbb{Z}_{\geq 0}$ and $i,j\in\{1,2\}$.
\end{prop}
\begin{proof} Similar to the proof of Proposition \ref{z:prop5} we use Theorem \ref{kdmaster} above and conclude at once that
\begin{equation}\label{kd10}
	{\bf B}_1^{\tau}(x,y)=\begin{bmatrix} 0 & -\im U^{\tau}\smallskip\\
	\im V^{\tau}& 0\end{bmatrix}(x,y),\ \ U^{\tau}:=\lambda^{\frac{1}{2}}\int_{\Sigma}N_1^{\tau}(\eta)\otimes K_2^{\tau}(\eta)\,\d\eta,\  V^{\tau}:=\lambda^{\frac{1}{2}}\int_{\Sigma}N_2^{\tau}(\eta)\otimes K_1^{\tau}(\eta)\,\d\eta
\end{equation}
The corresponding expression for ${\bf A}^{\tau}(\zeta)$ requires no further explanation.
\end{proof}
Continuing our analysis we now study the compatibility condition
\begin{equation}\label{kd11}
	{\bf A}^{\tau}(\zeta){\bf B}^{\tau}(\zeta)-{\bf B}^{\tau}(\zeta){\bf A}^{\tau}(\zeta)=\frac{\partial{\bf B}^{\tau}}{\partial t_{2n+1}}(\zeta)-\frac{\partial{\bf A}^{\tau}}{\partial t_1}(\zeta)
\end{equation}
of system \eqref{kd60},\eqref{kd9}.
\begin{prop} By \eqref{kd9}, the constraint \eqref{kd11} is equivalent to the following equations for the coefficients $U^{\tau},V^{\tau}$ and $A_k^{\tau ij}$ with $i,j=1,2$,
\begin{equation}\label{kd12}
	A_1^{\tau 12}=-\im U^{\tau},\ \ \ \ A_1^{\tau 21}=\im V^{\tau},
\end{equation}
followed by
\begin{equation}\label{kd13}
	\begin{cases}\displaystyle\frac{\partial A_k^{\tau 11}}{\partial t_1}=-\im(U^{\tau}A_k^{\tau 21}+A_k^{\tau 12}V^{\tau}),\ \ \frac{\partial A_k^{\tau 12}}{\partial t_1}=-\im(A_{k+1}^{\tau 12}+U^{\tau}A_k^{\tau 22}-A_k^{\tau 11}U)\bigskip\\
	\displaystyle\frac{\partial A_k^{\tau 22}}{\partial t_1}=\im(V^{\tau}A_k^{\tau 12}+A_k^{\tau 21}U^{\tau}),\ \ \ \ \frac{\partial A_k^{\tau 21}}{\partial t_1}=\im(A_{k+1}^{\tau 21}+V^{\tau}A_k^{\tau 11}-A_k^{\tau 22}V^{\tau})\end{cases},\ \ k=1,\ldots,2n,
\end{equation}
and concluding with
\begin{equation}\label{kd14}
	\begin{cases}\displaystyle\frac{\partial A_{2n+1}^{\tau 11}}{\partial t_1}=-\im(U^{\tau}A_{2n+1}^{\tau 21}+A_{2n+1}^{\tau 12}V^{\tau}),\ \ \frac{\partial A_{2n+1}^{\tau 12}}{\partial t_1}=-\im(U^{\tau}A_{2n+1}^{\tau 22}-A_{2n+1}^{\tau 11}U^{\tau})-\im\frac{\partial U^{\tau}}{\partial t_{2n+1}}\bigskip\\
	\displaystyle\frac{\partial A_{2n+1}^{\tau 22}}{\partial t_1}=\im(V^{\tau}A_{2n+1}^{\tau 12}+A_{2n+1}^{\tau 21}U^{\tau}),\ \ \ \ \frac{\partial A_{2n+1}^{\tau 21}}{\partial t_1}=\im(V^{\tau}A_{2n+1}^{\tau 11}-A_{2n+1}^{\tau 22}V^{\tau})+\im\frac{\partial V^{\tau}}{\partial t_{2n+1}}\end{cases}.
\end{equation}
\end{prop}
\begin{proof} By \eqref{kd11} and \eqref{kd9},
\begin{equation*}
	\frac{\partial {\bf B}_1^{\tau}}{\partial t_{2n+1}}=\sum_{k=1}^{2n+1}\zeta^{2n+1-k}\bigg(\frac{\partial{\bf A}_k^{\tau}}{\partial t_1}+\big[{\bf A}_k^{\tau},{\bf B}_1^{\tau}\big]\bigg)+\zeta^{2n+1}\big[{\bf A}_0,{\bf B}_1^{\tau}\big]+\sum_{k=0}^{2n}\zeta^{2n+1-k}\big[{\bf A}_{k+1}^{\tau},{\bf B}_0\big]
\end{equation*}
with the operator commutator $[\cdot,\cdot]$ on $\mathcal{H}_1$. Reading this equality to order $\mathcal{O}(\zeta^{2n+1})$ yields \eqref{kd12} and when read to all orders $\mathcal{O}(\zeta^{2n+1-k})$ with $k=1,\ldots,2n$ in turn \eqref{kd13}. Finally, \eqref{kd14} follows from order $\mathcal{O}(\zeta^0)$ and concludes our proof.
\end{proof}
As done in \eqref{z91}, the operator-valued equations \eqref{kd12},\eqref{kd13} and \eqref{kd14} constitute a coupled operator-valued PDE system for $U^{\tau}$ and $V^{\tau}$. However, we have no further use for this system and will therefore immediately evaluate \eqref{kd12},\eqref{kd13} and \eqref{kd14} on the kernel level, as done in Section \ref{zsec5} in the $\tau$-independent setting. First, noticing that \eqref{kd12} and \eqref{kd13} are formally equivalent to \eqref{z80} and \eqref{z81}, we can record the following crucial symmetry constraint (the mKdV analogue of Lemma \ref{z:lem12}).
\begin{lem}\label{kd:lem2} Let $k\in\{1,\ldots,2n+1\}$, then $A_k^{\tau 12}(x,y)$ and $A_k^{\tau 21}(y,x)$ are $y$-independent and we have
\begin{equation*}
	A_k^{\tau 12}(x,y)=(-1)^kA_k^{\tau 21}(y,x),\ \ \ \ (x,y)\in\mathbb{R}^2.
\end{equation*}
\end{lem}
\begin{proof} Exactly as in the proof of Lemma \ref{z:lem12}, this time using \eqref{kd13},\eqref{kd14} together with the aforementioned $\tau$-extension of \eqref{z56}, i.e. the kernel equality
\begin{equation*}
	U^{\tau}(x,y)=V^{\tau}(y,x),\ \ \ \ (x,y)\in\mathbb{R}^2,
\end{equation*}
and \eqref{kd12}.
\end{proof}
Hence, with the abbreviation $v(t_1,t_{2n+1}|x;n,\lambda)\equiv v(t_1,t_{2n+1}|x):=U^{\tau}(x,x)$ (recall $U^{\tau}$ depends on $n$ and $\lambda$) as well as
\begin{equation*}
	a_k^{\tau}(t_1,t_{2n+1}|x):=A_k^{\tau 12}(x,x)=(-1)^kA_k^{\tau 21}(x,x),\ \ \ \ k\in\{1,\ldots,2n+1\}
\end{equation*}
we have the direct $\tau$-analogue of \eqref{z96} in the form
\begin{equation}\label{kd15}
	a_{k+1}^{\tau}(t_1,t_{2n+1}|x)=\begin{cases}(\mathcal{L}_+^va_k^{\tau})(t_1,t_{2n+1}|x),&k\equiv 0\mod 2\smallskip\\
	(\mathcal{L}_-^va_k^{\tau})(t_1,t_{2n+1}|x),&k\equiv 1\mod 2\end{cases},\ \ \ \ \ \ \ \ \ k=1,2,\ldots,2n;
\end{equation}
with initial data $a_1^{\tau}(t_1,t_{2n+1}|x):=-\im v(t_1,t_{2n+1}|x)$, using the operators $\mathcal{L}_{\pm}^v$ of Definition \ref{idef2}.
Moreover, from \eqref{kd14} we find in addition
\begin{equation}\label{kd16}
	\frac{\partial v}{\partial t_{2n+1}}(t_1,t_{2n+1}|x)=(\mathcal{L}_-^va_{2n+1}^{\tau})(t_1,t_{2n+1}|x),
\end{equation}
and \eqref{kd15},\eqref{kd16} combined yield
\begin{equation}\label{kd17}
	\frac{\partial v}{\partial t_{2n+1}}(t_1,t_{2n+1}|x)=\left((\mathcal{L}_-^v\mathcal{L}_+^v)^n\frac{\partial v}{\partial t_1}\right)(t_1,t_{2n+1}|x).
\end{equation}
Finally, returning to \eqref{z96},\eqref{z33},\eqref{kd30}, we obtain through a contour deformation argument
\begin{equation*}
	v(t_1,t_{2n+1}|x)=U^{\tau}(x,x)=\frac{1}{\tau}U\left(\frac{x}{\tau},\frac{x}{\tau}\right)=\frac{1}{\tau}u\left(t\Big|\frac{x}{\tau}\right)
\end{equation*}
subject to the mKdV variable choice \eqref{kd4}. In summary, and this proves Theorem \ref{itheo2},
\begin{prop}\label{kd:prop3} Suppose $u(t|x)=u(t|x;n)$ solves the integro-differential Painlev\'e-II hierarchy \eqref{z97}, then 
\begin{equation*}
	v(t_1,t_{2n+1}|x)=v(t_1,t_{2n+1}|x;n):=\frac{1}{\tau}u\left(t\Big|\frac{x}{\tau}\right),\ \ \ \ \ t_1=\tau t\in\mathbb{R},\ \ \ \ \ t_{2n+1}=\frac{\tau^{2n+1}}{2n+1}\in\mathbb{R}_+,
\end{equation*}
solves the integro-differential mKdV hierarchy \eqref{kd17}.
\end{prop}
The outstanding Corollary \ref{icor1} is now a straightforward consequence of Proposition \ref{kd:prop3}. Indeed, we first have the following mKdV extension of Lemma \ref{z:lem9}.
\begin{lem}\label{kd:lem3} For every $(t,\lambda,\tau,n)\in\mathbb{R}\times\overline{\mathbb{D}_1(0)}\times\mathbb{N}$,
\begin{equation*}
	\frac{\partial}{\partial t_1}\ln D_n(t,\lambda,\tau)=-\im\lambda^{\frac{1}{2}}\tr_{\mathcal{H}_1}\int_{\Sigma}N_1^{\tau}(\xi)\otimes K_1^{\tau}(\xi)\,\d\xi,
\end{equation*}
and
\begin{equation*}
	\frac{\partial^2}{\partial t_1^2}\ln D_n(t,\lambda,\tau)=-\lambda\tr_{\mathcal{H}_1}\int_{\Sigma}\int_{\Sigma}\big(N_1^{\tau}(\eta)\otimes K_2^{\tau}(\eta)\big)\big(N_2^{\tau}(\xi)\otimes K_1(\xi)\big)\,\d\eta\,\d\xi.
\end{equation*}
\end{lem}
\begin{proof} Since ${\bf B}(\zeta)$ in \eqref{z74} and ${\bf B}^{\tau}(\zeta)$ in \eqref{kd9} are structurally identical, the above identities follow exactly as in the proof of Lemma \ref{z:lem9}, using en route
\begin{equation*}
	\frac{\partial}{\partial t_1}C_{t,n}^{\tau}(\xi,\eta)=\im\int_{\mathbb{R}}k_1^{\tau}(\xi|z)m_1^{\tau}(\eta|z)\,\d\sigma(z),\ \ \ \ (\xi,\eta)\in\Sigma\times\Sigma.
\end{equation*}
\end{proof}
Second, using Lemma \ref{kd:lem3} (now specialized to $\lambda=1,\tau=\alpha$) and \eqref{kd10} as well as Lemma \ref{kd:lem2}, we find
\begin{equation}\label{kd18}
	\frac{\partial^2}{\partial t_1^2}\ln D_n(t,\lambda,\tau)=-\lambda\tr_{\mathcal{H}_1}\big(U^{\tau}V^{\tau}\big)=-\int_{\mathbb{R}}\big(U^{\tau}(x,x)\big)^2\,\d\sigma(x)=-\int_{\mathbb{R}}v^2(t_1,t_{2n+1}|x)\,\d\sigma(x)
\end{equation}
and, just as in the proof of Lemma \ref{z:lem13}, as $t_1\rightarrow+\infty$, pointwise in $(t_{2n+1},x)\in\mathbb{R}_+\times\mathbb{R}$,
\begin{equation*}
	v(t_1,t_{2n+1}|x)=U^{\alpha}(x,x)\stackrel{\eqref{kd10}}{=}\int_{\Sigma}n_1^{\alpha}(\eta|x)k_2^{\alpha}(\eta|x)\,\d\eta\sim\int_{\Sigma}m_1^{\alpha}(\eta|x)k_2^{\alpha}(\eta|x)\,\d\eta=\frac{1}{2\pi}\int_{\Gamma_{\beta}}\e^{-\im\phi_n^{\alpha}(\eta,t_1+x)}\,\d\eta. 
\end{equation*}
Combining the last expansion with \eqref{kd18}, \eqref{z17}, Proposition \ref{kd:prop3} and Cauchy's as well as Fubini's theorem we finally arrive at \eqref{i22} and \eqref{i23}, and have thus completed our proof of Corollary \ref{icor1}.

\begin{appendix}
\section{Two auxiliary results}\label{appA}
The following two subsections summarize analytic results used in the proofs of Propositions \ref{zprop:1} and \ref{zprop:2}.
\subsection{On a Fourier-Stieltjes integral}\label{FSint}
Consider the function $f:\overline{\mathbb{H}}_{\epsilon}\rightarrow\mathbb{C}$ defined as the Fourier-Stieltjes integral
\begin{equation*}
	f(\lambda):=\int_{\mathbb{R}}\e^{-\im z\lambda}\d\sigma(z),
\end{equation*}
in the closed horizontal strip $\overline{\mathbb{H}}_{\epsilon}:=\{\lambda\in\mathbb{C}:\,|\Im\lambda|\leq\frac{\omega}{2}-\epsilon\}$ for any fixed $0<\epsilon<\frac{\omega}{2}$. Given that $\d\sigma$ is a positive Borel probability measure on $\mathbb{R}$ with $\frac{\d\sigma}{\d z}(z)=w'(z)\leq\e^{-\omega|z|},|z|\geq z_0$, see \eqref{i13}, we deduce that $\overline{\mathbb{H}}_{\epsilon}\ni\lambda\mapsto f(\lambda)$ is uniformly continuous. Next, given an arbitrary piecewise smooth closed curve $\Gamma\subset\mathbb{H}_{\epsilon}$ in the open strip $\mathbb{H}_{\epsilon}$, we obtain by Fubini's theorem,
\begin{equation*}
	\oint_{\Gamma}f(\lambda)\d\lambda=\oint_{\Gamma}\left[\int_{-\infty}^{\infty}\e^{-\im z\lambda}\d\sigma(z)\right]\d\lambda=\int_{-\infty}^{\infty}\left[\oint_{\Gamma}\e^{-\im z\lambda}\,\d\lambda\right]\d\sigma(z)=0.
\end{equation*}
Hence, by Morera's theorem, $f$ is analytic in $\mathbb{H}_{\epsilon}$.
\subsection{Analyticity implies stability}\label{Stab}
Consider the function $J_{t,n}:(\overline{\mathbb{H}}_{\epsilon}\setminus\Gamma_{\beta})\times(\overline{\mathbb{H}}_{\epsilon}\setminus\Gamma_{\beta})\rightarrow\mathbb{C}$ defined as the iterated integral
\begin{equation*}
	J_{t,n}(\lambda,\mu)=\frac{1}{(2\pi)^2}\int_{\Gamma_{\beta}}\frac{\e^{\frac{\im}{2}\psi_n(\lambda,2t)-\frac{\im}{2}\psi_n(\beta,2t)}}
	{\lambda-\beta}\left[\int_{\mathbb{R}}\e^{\im z(\lambda-\beta)}\d\sigma(z)\right]\frac{\e^{-\frac{\im}{2}\psi_n(\beta,0)+\frac{\im}{2}\psi_n(\mu,0)}}{\beta-\mu}\,\d\beta,
\end{equation*}
where, as before, $\overline{\mathbb{H}}_{\epsilon}=\{\lambda\in\mathbb{C}:\,|\Im\lambda|\leq\frac{\omega}{2}-\epsilon\}$ for any fixed $0<\epsilon<\frac{\omega}{2}$. Since
\begin{equation*}
	J_t(\lambda,\mu)=\frac{1}{(2\pi)^2}\e^{\frac{\im}{2}\psi_n(\lambda,2t)+\frac{\im}{2}\psi_n(\mu,0)}\int_{\Gamma_{\beta}}\frac{\e^{-\im\psi_n(\beta,t)}}{(\lambda-\beta)(\beta-\mu)}\left[\int_{\mathbb{R}}\e^{\im z(\lambda-\beta)}\,\d\sigma(z)\right]\d\beta,
\end{equation*}
we conclude that both, $\mathbb{H}_{\epsilon}\setminus\Gamma_{\beta}\ni\lambda\mapsto J_t(\lambda,\mu)$, resp. $\mathbb{H}_{\epsilon}\setminus\Gamma_{\beta}\ni\mu\mapsto J_t(\lambda,\mu)$, are analytic for every fixed $\mu\in\mathbb{H}_{\epsilon}\setminus\Gamma_{\beta}$, resp. for every fixed $\lambda\in\mathbb{H}_{\epsilon}\setminus\Gamma_{\beta}$. This is because of Appendix \ref{appA} and the fact that 
\begin{equation*}
	\beta\mapsto\frac{\e^{-\im\psi_n(\beta,t)}}{\lambda-\beta}\left[\int_{\mathbb{R}}\e^{\im z(\lambda-\beta)}\d\sigma(z)\right],\ \ \lambda\in\overline{\mathbb{H}}_{\epsilon}\setminus\Gamma_{\beta},
\end{equation*}
and
\begin{equation*}
	\beta\mapsto\frac{\e^{-\im\psi_n(\beta,t)}}{\mu-\beta}\left[\int_{-\infty}^{\infty}\e^{\im z(\lambda-\beta)}\d\sigma(z)\right],\ \ \mu\in\overline{\mathbb{H}}_{\epsilon}\setminus\Gamma_{\beta}
\end{equation*}
are locally Lipschitz continuous on $\Gamma_{\beta}$. In turn, using Hartog's theorem, we conclude that $(\lambda,\mu)\mapsto J_{t,n}(\lambda,\mu)$ is analytic on $(\mathbb{H}_{\epsilon}\setminus\Gamma_{\beta})\times(\mathbb{H}_{\epsilon}\setminus\Gamma_{\beta})$.
\section{Abbreviations and terminology}\label{appC}
The following abbreviations and terminology are used throughout Sections \ref{zsec3}, \ref{zsec4}, \ref{zsec5} and \ref{zsec6}.
\begin{definition}[{\cite[Definition $9.2$]{B}}] Let $p\in\mathbb{Z}_{\geq 1}$ and $\d\sigma(z)=w'(z)\d z$. We use the below abbreviations.
\begin{enumerate}
	\item[(1)] The direct sum Hilbert space
	\begin{equation*}
		\mathcal{H}_p:=\bigoplus_{j=1}^pL^2(\mathbb{R},\d\sigma)=\big\{{\bf f}=(f_1,\ldots,f_p)^{\top}\in\mathbb{C}^{p\times 1}:\,f_j\in L^2(\mathbb{R},\d\sigma)\big\}
	\end{equation*}
	equipped with its standard inner product and associated norm.
	\item[(2)] The space $L^2(\mathbb{R},\d\sigma;\mathbb{C}^{p\times p})$ of $p\times p$ matrix-valued functions with entries in $L^2(\mathbb{R},\d\sigma)$, equipped with the induced Frobenius integral norm.
	\item[(3)] The space $\mathcal{I}(\mathcal{H}_p)$ of Hilbert-Schmidt integral operators on $\mathcal{H}_p$ of the form
	\begin{equation*}
		({\bf K}{\bf f})(x)=\int_{\mathbb{R}}{\bf K}(x,y){\bf f}(y)\,\d\sigma(y),
	\end{equation*}
	with kernel ${\bf K}\in L^2(\mathbb{R}^2,\d\sigma\otimes\d\sigma;\mathbb{C}^{p\times p})$.
	\item[(4)] The matrix identity operator $\mathbb{I}_p$ on $\mathcal{H}_p$.
\end{enumerate}
\end{definition}
Next, we recall the notion of an \textit{analytic integral operator} as defined in \cite[Definition $9.3$]{B}, see also \cite[page $1781$]{IK}: Let ${\bf K}={\bf K}(\zeta)\in\mathcal{I}(\mathcal{H}_p)$ depend on an auxiliary variable $\zeta\in\Omega$ for some fixed region $\Omega\subset\mathbb{C}$.
\begin{definition}{{\cite[Definition $9.3$]{B}}}\label{c2} We say that ${\bf K}(\zeta)\in\mathcal{I}(\mathcal{H}_p)$ with kernel ${\bf K}(\zeta|x,y)$ is analytic in $\zeta\in\Omega$, if
\begin{enumerate}
	\item[(1)] for any $(x,y)\in\mathbb{R}^2$, the map $z\mapsto{\bf K}(\zeta|x,y)$ is analytic in $\Omega$.
	\item[(2)] for any $\zeta\in\Omega$, the map $(x,y)\mapsto{\bf K}(\zeta|x,y)$ is in $L^2(\mathbb{R}^2,\d\sigma\otimes\d\sigma;\mathbb{C}^{p\times p})$.
\end{enumerate}
\end{definition}
Furthermore, if $\Sigma\subset\Omega\subset\mathbb{C}$ is an oriented contour consisting of a finite union of smooth oriented curves in $\mathbb{CP}^1$ with finitely many self-intersections, then
\begin{definition}{{\cite[Definition $9.4$]{B}}} We say that an analytic in $\zeta\in\Omega\setminus\Sigma$ operator ${\bf K}(\zeta)\in\mathcal{I}(\mathcal{H}_p)$ admits continuous boundary values ${\bf K}_{\pm}(\zeta)\in\mathcal{I}(\mathcal{H}_p)$ on $\Sigma'\subset\Sigma$ with kernels ${\bf K}_{\pm}(\zeta|x,y)$ if
\begin{enumerate}
	\item[(1)] for any $(x,y)\in\mathbb{R}^2$, the map $\zeta\mapsto {\bf K}_{\pm}(\zeta|x,y)$ is continuous on $\Sigma'$.
	\item[(2)] for any $(x,y)\in\mathbb{R}^2$, the non-tangential limits
	\begin{equation*}
		\lim_{\lambda\rightarrow\zeta}{\bf K}(\lambda|x,y)={\bf K}_{\pm}(\zeta|x,y),\ \ \ \ \lambda\,\in\,\pm\,\textnormal{side of}\,\,\Sigma'\,\,\textnormal{at}\,\zeta
	\end{equation*}
	exist.
\end{enumerate}
\end{definition}
\end{appendix}



\begin{bibsection}
\begin{biblist}


\bib{AM}{article}{
AUTHOR = {Adler, M.},
author={Moser, J.},
     TITLE = {On a class of polynomials connected with the {K}orteweg-de
              {V}ries equation},
   JOURNAL = {Comm. Math. Phys.},
  FJOURNAL = {Communications in Mathematical Physics},
    VOLUME = {61},
      YEAR = {1978},
    NUMBER = {1},
     PAGES = {1--30},
      ISSN = {0010-3616},
   MRCLASS = {58F05 (35Q99)},
  MRNUMBER = {501106},
MRREVIEWER = {Alexander A. Pankov},
       URL = {http://projecteuclid.org.bris.idm.oclc.org/euclid.cmp/1103904169},
}



\bib{Airault}{article}{
AUTHOR = {Airault, H.},
     TITLE = {Rational solutions of {P}ainlev\'{e} equations},
   JOURNAL = {Stud. Appl. Math.},
  FJOURNAL = {Studies in Applied Mathematics},
    VOLUME = {61},
      YEAR = {1979},
    NUMBER = {1},
     PAGES = {31--53},
      ISSN = {0022-2526},
   MRCLASS = {58A17 (34A20 58F07)},
  MRNUMBER = {535866},
MRREVIEWER = {H. Hochstadt},
       DOI = {10.1002/sapm197961131},
       URL = {https://doi-org.bris.idm.oclc.org/10.1002/sapm197961131},
}


\bib{ACQ}{article}{
AUTHOR = {Amir, Gideon},
author={Corwin, Ivan} 
author={Quastel, Jeremy},
     TITLE = {Probability distribution of the free energy of the continuum
              directed random polymer in {$1+1$} dimensions},
   JOURNAL = {Comm. Pure Appl. Math.},
  FJOURNAL = {Communications on Pure and Applied Mathematics},
    VOLUME = {64},
      YEAR = {2011},
    NUMBER = {4},
     PAGES = {466--537},
      ISSN = {0010-3640},
   MRCLASS = {60K35 (60B20 60F05 60H15 82C22 82C44)},
  MRNUMBER = {2796514},
MRREVIEWER = {Timo Sepp\"{a}l\"{a}inen},
       DOI = {10.1002/cpa.20347},
       URL = {https://doi-org.bris.idm.oclc.org/10.1002/cpa.20347},
}


\bib{BB}{article}{
AUTHOR = {Baik, Jinho}
author={Bothner, Thomas},
     TITLE = {The largest real eigenvalue in the real {G}inibre ensemble and
              its relation to the {Z}akharov-{S}habat system},
   JOURNAL = {Ann. Appl. Probab.},
  FJOURNAL = {The Annals of Applied Probability},
    VOLUME = {30},
      YEAR = {2020},
    NUMBER = {1},
     PAGES = {460--501},
      ISSN = {1050-5164},
   MRCLASS = {60B20 (45M05 60G70)},
  MRNUMBER = {4068316},
       DOI = {10.1214/19-AAP1509},
       URL = {https://doi-org.bris.idm.oclc.org/10.1214/19-AAP1509},
}

\bib{Bel}{article}{
  title={Null octagon from Deift-Zhou steepest descent},
  author={Belitsky, A.V.},
 year={2020},
eprint={https://arxiv.org/abs/2012.10446},
    archivePrefix={arXiv},
    primaryClass={hep-th},
}

\bib{BC}{article}{
AUTHOR = {Bertola, M.},
author={Cafasso, M.},
     TITLE = {The transition between the gap probabilities from the
              {P}earcey to the Airy process---a {R}iemann-{H}ilbert
              approach},
   JOURNAL = {Int. Math. Res. Not. IMRN},
  FJOURNAL = {International Mathematics Research Notices. IMRN},
      YEAR = {2012},
    NUMBER = {7},
     PAGES = {1519--1568},
      ISSN = {1073-7928},
   MRCLASS = {60B20 (35Q15)},
  MRNUMBER = {2913183},
MRREVIEWER = {Beno\^{\i}t Collins},
       DOI = {10.1093/imrn/rnr066},
       URL = {https://doi-org.bris.idm.oclc.org/10.1093/imrn/rnr066},
}

\bib{BeB}{article}{
AUTHOR = {Betea, Dan},
author={Bouttier, J\'{e}r\'{e}mie},
     TITLE = {The periodic {S}chur process and free fermions at finite
              temperature},
   JOURNAL = {Math. Phys. Anal. Geom.},
  FJOURNAL = {Mathematical Physics, Analysis and Geometry. An International
              Journal Devoted to the Theory and Applications of Analysis and
              Geometry to Physics},
    VOLUME = {22},
      YEAR = {2019},
    NUMBER = {1},
     PAGES = {Paper No. 3, 47},
      ISSN = {1385-0172},
   MRCLASS = {82C23 (05E05 60G55 60K35)},
  MRNUMBER = {3903828},
       DOI = {10.1007/s11040-018-9299-8},
       URL = {https://doi-org.bris.idm.oclc.org/10.1007/s11040-018-9299-8},
}

\bib{BBW}{article}{
author={Betea, Dan},
author={Bouttier, J\'er\'emie},
author={Walsh, Harriet},
title={Multicritical random partitions},
year={2020},
eprint={https://arxiv.org/abs/2012.01995},
    archivePrefix={arXiv},
    primaryClass={math.CO},
}

\bib{BCF}{article}{
AUTHOR = {Borodin, Alexei},
author={Corwin, Ivan},
author={Ferrari, Patrik},
     TITLE = {Free energy fluctuations for directed polymers in random media
              in {$1+1$} dimension},
   JOURNAL = {Comm. Pure Appl. Math.},
  FJOURNAL = {Communications on Pure and Applied Mathematics},
    VOLUME = {67},
      YEAR = {2014},
    NUMBER = {7},
     PAGES = {1129--1214},
      ISSN = {0010-3640},
   MRCLASS = {82D10 (60B20 60H15 60K35 60K37)},
  MRNUMBER = {3207195},
       DOI = {10.1002/cpa.21520},
       URL = {https://doi-org.bris.idm.oclc.org/10.1002/cpa.21520},
}

\bib{BS}{book}{
AUTHOR = {Berezin, F. A.},
author={Shubin, M. A.},
     TITLE = {The {S}chr\"{o}dinger equation},
    SERIES = {Mathematics and its Applications (Soviet Series)},
    VOLUME = {66},
      NOTE = {Translated from the 1983 Russian edition by Yu. Rajabov, D. A.
              Le\u{\i}tes and N. A. Sakharova and revised by Shubin,
              With contributions by G. L. Litvinov and Le\u{\i}tes},
 PUBLISHER = {Kluwer Academic Publishers Group, Dordrecht},
      YEAR = {1991},
     PAGES = {xviii+555},
      ISBN = {0-7923-1218-X},
   MRCLASS = {81-01 (35J10 35P05 46N50 47F05 47N50)},
  MRNUMBER = {1186643},
       DOI = {10.1007/978-94-011-3154-4},
       URL = {https://doi-org.bris.idm.oclc.org/10.1007/978-94-011-3154-4},
}


\bib{B}{article}{
author={Bothner, Thomas},
 TITLE={On the origins of Riemann-Hilbert problems in mathematics},
    YEAR={2020},
    eprint={https://arxiv.org/abs/2003.14374},
    archivePrefix={arXiv},
    primaryClass={math.PH},
}

\bib{Bf}{article}{
author={Bothner, Thomas},
title={A Riemann-Hilbert approach to Fredholm determinants of integral Hankel composition operators: scalar kernels}
year={in preparation},
}



\bib{CCG}{article}{
author={Cafasso, Mattia},
author={Claeys, Tom},
author={Girotti, Manuela},
title={Fredholm determinant solutions of the Painlev\'e II hierarchy and gap probabilities of determinantal point processes},
journal={Int. Math. Res. Not. IMRN},
fjournal={International Mathematics Research Notices. IMRN},
volume={168},
year={2019},
URL={https://doi.org/10.1093/imrn/rnz168},
DOI={10.1093/imrn/rnz168},
}

\bib{CCR}{article}{
author={Cafasso, Mattia},
author={Claeys, Tom},
author={Ruzza, Giulio},
title={Airy kernel determinant solutions to the KdV equation and integro-differential Painlev\'e equations},
year={2020},
eprint={https://arxiv.org/abs/2010.07723},
    archivePrefix={arXiv},
    primaryClass={math.PH}
}




\bib{CIK}{article}{
AUTHOR = {Claeys, T.},
author={Its, A.},
author={Krasovsky, I.},
     TITLE = {Higher-order analogues of the {T}racy-{W}idom distribution and
              the {P}ainlev\'{e} {II} hierarchy},
   JOURNAL = {Comm. Pure Appl. Math.},
  FJOURNAL = {Communications on Pure and Applied Mathematics},
    VOLUME = {63},
      YEAR = {2010},
    NUMBER = {3},
     PAGES = {362--412},
      ISSN = {0010-3640},
   MRCLASS = {34M55 (33E17 34M50 37K15 47B10 60B20 82C05)},
  MRNUMBER = {2599459},
       DOI = {10.1002/cpa.20284},
       URL = {https://doi-org.bris.idm.oclc.org/10.1002/cpa.20284},
}

\bib{CV}{article}{
AUTHOR = {Claeys, T.},
author={Vanlessen, M.},
     TITLE = {Universality of a double scaling limit near singular edge
              points in random matrix models},
   JOURNAL = {Comm. Math. Phys.},
  FJOURNAL = {Communications in Mathematical Physics},
    VOLUME = {273},
      YEAR = {2007},
    NUMBER = {2},
     PAGES = {499--532},
      ISSN = {0010-3616},
   MRCLASS = {15A52 (62H20 82B31)},
  MRNUMBER = {2318316},
       DOI = {10.1007/s00220-007-0256-9},
       URL = {https://doi-org.bris.idm.oclc.org/10.1007/s00220-007-0256-9},
}

\bib{CJM}{article}{
AUTHOR = {Clarkson, Peter A.},
author={Joshi, Nalini}
author={Mazzocco, Marta},
     TITLE = {The {L}ax pair for the m{K}d{V} hierarchy},
 BOOKTITLE = {Th\'{e}ories asymptotiques et \'{e}quations de {P}ainlev\'{e}},
    SERIES = {S\'{e}min. Congr.},
    VOLUME = {14},
     PAGES = {53--64},
 PUBLISHER = {Soc. Math. France, Paris},
      YEAR = {2006},
   MRCLASS = {37K10 (34M55 37K15 37K20)},
  MRNUMBER = {2353461},
MRREVIEWER = {Maria A. Agrotis},
}

\bib{CM}{article}{
AUTHOR = {Clarkson, Peter A.},
author={Mansfield, Elizabeth L.},
     TITLE = {The second {P}ainlev\'{e} equation, its hierarchy and associated
              special polynomials},
   JOURNAL = {Nonlinearity},
  FJOURNAL = {Nonlinearity},
    VOLUME = {16},
      YEAR = {2003},
    NUMBER = {3},
     PAGES = {R1--R26},
      ISSN = {0951-7715},
   MRCLASS = {34M55},
  MRNUMBER = {1975781},
MRREVIEWER = {Andrei A. Kapaev},
       DOI = {10.1088/0951-7715/16/3/201},
       URL = {https://doi-org.bris.idm.oclc.org/10.1088/0951-7715/16/3/201},
}

\bib{C}{article}{
AUTHOR = {Corwin, Ivan},
     TITLE = {The {K}ardar-{P}arisi-{Z}hang equation and universality class},
   JOURNAL = {Random Matrices Theory Appl.},
  FJOURNAL = {Random Matrices. Theory and Applications},
    VOLUME = {1},
      YEAR = {2012},
    NUMBER = {1},
     PAGES = {1130001, 76},
      ISSN = {2010-3263},
   MRCLASS = {82B31 (60B20 60K35 60K37)},
  MRNUMBER = {2930377},
       DOI = {10.1142/S2010326311300014},
       URL = {https://doi-org.bris.idm.oclc.org/10.1142/S2010326311300014},
}



\bib{DDMS}{article}{
title = {Noninteracting fermions at finite temperature in a $d$-dimensional trap: Universal correlations},
  author = {Dean, David S.},
  author={Le Doussal, Pierre},
  author={Majumdar, Satya N.},
  author={Schehr, Gr\'egory},
  journal = {Phys. Rev. A},
  volume = {94},
  issue = {6},
  pages = {063622},
  numpages = {41},
  year = {2016},
  month = {Dec},
  publisher = {American Physical Society},
  doi = {10.1103/PhysRevA.94.063622},
  url = {https://link.aps.org/doi/10.1103/PhysRevA.94.063622}
}

\bib{DDMS2}{article}{
 title = {Wigner function of noninteracting trapped fermions},
  author = {Dean, David S.},
  author={Le Doussal, Pierre},
  author={Majumdar, Satya N.},
  author={Schehr, Gr\'egory},
  journal = {Phys. Rev. A},
  volume = {97},
  issue = {6},
  pages = {063614},
  numpages = {14},
  year = {2018},
  month = {Jun},
  publisher = {American Physical Society},
  doi = {10.1103/PhysRevA.97.063614},
  url = {https://link.aps.org/doi/10.1103/PhysRevA.97.063614}
}

\bib{D}{article}{
AUTHOR = {Dimitrov, Evgeni},
     TITLE = {K{PZ} and {A}iry limits of {H}all-{L}ittlewood random plane
              partitions},
   JOURNAL = {Ann. Inst. Henri Poincar\'{e} Probab. Stat.},
  FJOURNAL = {Annales de l'Institut Henri Poincar\'{e} Probabilit\'{e}s et
              Statistiques},
    VOLUME = {54},
      YEAR = {2018},
    NUMBER = {2},
     PAGES = {640--693},
      ISSN = {0246-0203},
   MRCLASS = {60K35 (05A17 05E05 33D52 60H15 82B23)},
  MRNUMBER = {3795062},
       DOI = {10.1214/16-AIHP817},
       URL = {https://doi-org.bris.idm.oclc.org/10.1214/16-AIHP817},
}

\bib{FN}{article}{
	author = {Flaschka, Hermann} 
	author = {Newell, Alan C.},
	coden = {CMPHAY},
	date-added = {2010-11-08 17:10:12 -0500},
	date-modified = {2010-11-08 17:10:21 -0500},
	fjournal = {Communications in Mathematical Physics},
	issn = {0010-3616},
	journal = {Comm. Math. Phys.},
	mrclass = {35Q20 (14D05 58F07 81C05)},
	mrnumber = {588248 (82g:35103)},
	mrreviewer = {H{{\'e}}l{{\`e}}ne Airault},
	number = {1},
	pages = {65--116},
	title = {Monodromy and spectrum preserving deformations. {I}},
	url = {http://projecteuclid.org/getRecord?id=euclid.cmp/1103908189},
	volume = {76},
	year = {1980},
	Bdsk-Url-1 = {http://www.ams.org/mathscinet-getitem?mr=588248}}


\bib{GGK}{book}{
AUTHOR = {Gohberg, Israel}
author={Goldberg, Seymour}
author={Krupnik, Nahum},
     TITLE = {Traces and determinants of linear operators},
    SERIES = {Operator Theory: Advances and Applications},
    VOLUME = {116},
 PUBLISHER = {Birkh\"{a}user Verlag, Basel},
      YEAR = {2000},
     PAGES = {x+258},
      ISBN = {3-7643-6177-8},
   MRCLASS = {47B10 (45B05 45P05 47A53 47G10 47L10)},
  MRNUMBER = {1744872},
MRREVIEWER = {Hermann K\"{o}nig},
       DOI = {10.1007/978-3-0348-8401-3},
       URL = {https://doi.org/10.1007/978-3-0348-8401-3},
}

\bib{IK}{article}{
AUTHOR = {Its, A. R.},
author={Kozlowski, K. K.},
     TITLE = {Large-{$x$} analysis of an operator-valued {R}iemann-{H}ilbert
              problem},
   JOURNAL = {Int. Math. Res. Not. IMRN},
  FJOURNAL = {International Mathematics Research Notices. IMRN},
      YEAR = {2016},
    NUMBER = {6},
     PAGES = {1776--1806},
      ISSN = {1073-7928},
   MRCLASS = {47G10 (35Q15 45P05)},
  MRNUMBER = {3509940},
MRREVIEWER = {Lu\'{\i}s P. Castro},
       DOI = {10.1093/imrn/rnv188},
       URL = {https://doi-org.bris.idm.oclc.org/10.1093/imrn/rnv188},
}

\bib{Joh}{article}{
author={Johansson, Kurt},
 TITLE={Random matrices and determinantal processes},
    YEAR={2005},
    journal={Lecture notes from the Les Houches summer school on Mathematical Statistical Physics},
    eprint={https://arxiv.org/abs/math-ph/0510038v1},
    archivePrefix={arXiv},
    primaryClass={math.PH},
}

\bib{J}{article}{
AUTHOR = {Johansson, Kurt},
     TITLE = {From {G}umbel to {T}racy-{W}idom},
   JOURNAL = {Probab. Theory Related Fields},
  FJOURNAL = {Probability Theory and Related Fields},
    VOLUME = {138},
      YEAR = {2007},
    NUMBER = {1-2},
     PAGES = {75--112},
      ISSN = {0178-8051},
   MRCLASS = {60G70 (15A52 60G07 62G32 82B41)},
  MRNUMBER = {2288065},
MRREVIEWER = {Alexander Roitershtein},
       DOI = {10.1007/s00440-006-0012-7},
       URL = {https://doi-org.bris.idm.oclc.org/10.1007/s00440-006-0012-7},
}

\bib{KZ}{article}{
author={Kimura, Taro},
author={Zahabi, Ali},
title={Universal edge scaling in random partitions},
year={2020}
eprint={https://arxiv.org/abs/2012.06424},
    archivePrefix={arXiv},
    primaryClass={cond-mat.stat-mech}
}

\bib{Ko}{article}{
AUTHOR = {Kohno, Mitsuhiko},
     TITLE = {An extended {A}iry function of the first kind},
   JOURNAL = {Hiroshima Math. J.},
  FJOURNAL = {Hiroshima Mathematical Journal},
    VOLUME = {9},
      YEAR = {1979},
    NUMBER = {2},
     PAGES = {473--489},
      ISSN = {0018-2079},
   MRCLASS = {33A70 (34E05)},
  MRNUMBER = {535522},
MRREVIEWER = {F. W. J. Olver},
       URL = {http://projecteuclid.org.bris.idm.oclc.org/euclid.hmj/1206134896},
}


\bib{Kra}{article}{
author={Krajenbrink, Alexandre},
 title={From Painlev\'e to Zakharov-Shabat and beyond: Fredholm determinants and integro-differential hierarchies},
    year={2020},
    eprint={https://arxiv.org/abs/2008.01509},
    archivePrefix={arXiv},
    primaryClass={math.PH}
}



%
%
\bib{KBI}{book}{
AUTHOR = {Korepin, Vladimir}
author={Bogoliubov, N. M.}
author={Izergin, A. G.},
     TITLE = {Quantum inverse scattering method and correlation functions},
    SERIES = {Cambridge Monographs on Mathematical Physics},
 PUBLISHER = {Cambridge University Press, Cambridge},
      YEAR = {1993},
     PAGES = {xx+555},
      ISBN = {0-521-37320-4; 0-521-58646-1},
   MRCLASS = {81U40 (81-02 81T40 82-02 82B10)},
  MRNUMBER = {1245942},
MRREVIEWER = {Makoto Idzumi},
       DOI = {10.1017/CBO9780511628832},
       URL = {https://doi.org/10.1017/CBO9780511628832},
}

\bib{DoussalKP}{article}{
doi = {10.1088/1742-5468/ab75e4},
	url = {https://doi.org/10.1088/1742-5468/ab75e4},
	month = {apr},
	publisher = {{IOP} Publishing},
	volume = {2020},
	year={2020},
	number = {4},
	pages = {043201},
	author = {Le Doussal, Pierre},
	title = {Large deviations for the Kardar{\textendash} Parisi{\textendash}Zhang equation from the Kadomtsev{\textendash}Petviashvili equation},
	journal = {Journal of Statistical Mechanics: Theory and Experiment},
}


\bib{DMS}{article}{
  title = {Multicritical Edge Statistics for the Momenta of Fermions in Nonharmonic Traps},
  author = {Le Doussal, Pierre}
  author={Majumdar, Satya N.}
  author={Schehr, Gr\'egory},
  journal = {Phys. Rev. Lett.},
  volume = {121},
  issue = {3},
  pages = {030603},
  numpages = {7},
  year = {2018},
  month = {Jul},
  publisher = {American Physical Society},
  doi = {10.1103/PhysRevLett.121.030603},
  url = {https://link.aps.org/doi/10.1103/PhysRevLett.121.030603}
}

\bib{DMSa}{article}{
 title = {Multicritical Edge Statistics for the Momenta of Fermions in Nonharmonic Traps},
  author = {Le Doussal, Pierre}
  author={Majumdar, Satya N.}
  author={Schehr, Gr\'egory},
year={2018},
eprint={https://arxiv.org/abs/1802.06436},
    archivePrefix={arXiv},
    primaryClass={cond-mat.stat-mech}
}


\bib{LW}{article}{
AUTHOR = {Liechty, Karl},
author={Wang, Dong},
     TITLE = {Asymptotics of free fermions in a quadratic well at finite
              temperature and the {M}oshe-{N}euberger-{S}hapiro random
              matrix model},
   JOURNAL = {Ann. Inst. Henri Poincar\'{e} Probab. Stat.},
  FJOURNAL = {Annales de l'Institut Henri Poincar\'{e} Probabilit\'{e}s et
              Statistiques},
    VOLUME = {56},
      YEAR = {2020},
    NUMBER = {2},
     PAGES = {1072--1098},
      ISSN = {0246-0203},
   MRCLASS = {60B20 (15B52 82B23)},
  MRNUMBER = {4076776},
       DOI = {10.1214/19-AIHP994},
       URL = {https://doi-org.bris.idm.oclc.org/10.1214/19-AIHP994},
}

\bib{M}{book}{
AUTHOR = {Muskhelishvili, N. I.},
     TITLE = {Singular integral equations},
      NOTE = {Boundary problems of function theory and their application to
              mathematical physics,
              Translated from the second (1946) Russian edition and with a
              preface by J. R. M. Radok,
              Corrected reprint of the 1953 English translation},
 PUBLISHER = {Dover Publications, Inc., New York},
      YEAR = {1992},
     PAGES = {447},
      ISBN = {0-486-66893-2},
   MRCLASS = {45-02 (30E25 45Exx 47G10 47N20)},
  MRNUMBER = {1215485},
}

\bib{NIST}{book}{
TITLE = {N{IST} handbook of mathematical functions},
    EDITOR = {Olver, Frank W. J.}
    editor={Lozier, Daniel W.}
    editor={Boisvert, Ronald F.}
    editor={Clark, Charles W.},
 PUBLISHER = {U.S. Department of Commerce, National Institute of Standards
              and Technology, Washington, DC; Cambridge University Press,
              Cambridge},
      YEAR = {2010},
     PAGES = {xvi+951},
      ISBN = {978-0-521-14063-8},
   MRCLASS = {33-00 (00A20 65-00)},
  MRNUMBER = {2723248},
}

\bib{QR}{article}{
  title={KP governs random growth off a one dimensional substrate},
  year={2019}
  author={Quastel, Jeremy}
  author={Remenik, Daniel},
  eprint={https://arxiv.org/abs/1908.10353},
  archivePrefix={arXiv},
  primaryClass={math.PR},
}



\bib{S}{book}{
AUTHOR = {Simon, Barry},
     TITLE = {Trace ideals and their applications},
    SERIES = {Mathematical Surveys and Monographs},
    VOLUME = {120},
   EDITION = {Second},
 PUBLISHER = {American Mathematical Society, Providence, RI},
      YEAR = {2005},
     PAGES = {viii+150},
      ISBN = {0-8218-3581-5},
   MRCLASS = {47L20 (47A40 47A55 47B10 47B36 47E05 81Q15 81U99)},
  MRNUMBER = {2154153},
MRREVIEWER = {Pavel B. Kurasov},
}

\bib{S2}{book}{
AUTHOR = {Simon, Barry},
     TITLE = {Advanced complex analysis},
    SERIES = {A Comprehensive Course in Analysis, Part 2B},
 PUBLISHER = {American Mathematical Society, Providence, RI},
      YEAR = {2015},
     PAGES = {xvi+321},
      ISBN = {978-1-4704-1101-5},
   MRCLASS = {30-01 (11-01 33-01 34-01 60J67)},
  MRNUMBER = {3364090},
MRREVIEWER = {Fritz Gesztesy},
       DOI = {10.1090/simon/002.2},
       URL = {https://doi-org.bris.idm.oclc.org/10.1090/simon/002.2},
}

\bib{S3}{book}{
AUTHOR = {Simon, Barry},
     TITLE = {Operator theory},
    SERIES = {A Comprehensive Course in Analysis, Part 4},
 PUBLISHER = {American Mathematical Society, Providence, RI},
      YEAR = {2015},
     PAGES = {xviii+749},
      ISBN = {978-1-4704-1103-9},
   MRCLASS = {47-01 (34-01 35-01 42B35 42B37 43-01 46-01 81-01)},
  MRNUMBER = {3364494},
MRREVIEWER = {Fritz Gesztesy},
       DOI = {10.1090/simon/004},
       URL = {https://doi-org.bris.idm.oclc.org/10.1090/simon/004},
}

\bib{Sos}{article}{
AUTHOR = {Soshnikov, A.},
     TITLE = {Determinantal random point fields},
   JOURNAL = {Uspekhi Mat. Nauk},
  FJOURNAL = {Uspekhi Matematicheskikh Nauk},
    VOLUME = {55},
      YEAR = {2000},
    NUMBER = {5(335)},
     PAGES = {107--160},
      ISSN = {0042-1316},
   MRCLASS = {60G55 (60F05 60K05)},
  MRNUMBER = {1799012},
MRREVIEWER = {Boris A. Khoruzhenko},
       DOI = {10.1070/rm2000v055n05ABEH000321},
       URL = {https://doi-org.bris.idm.oclc.org/10.1070/rm2000v055n05ABEH000321},
}

\bib{T}{book}{
AUTHOR = {Takhtajan, Leon A.},
     TITLE = {Quantum mechanics for mathematicians},
    SERIES = {Graduate Studies in Mathematics},
    VOLUME = {95},
 PUBLISHER = {American Mathematical Society, Providence, RI},
      YEAR = {2008},
     PAGES = {xvi+387},
      ISBN = {978-0-8218-4630-8},
   MRCLASS = {81-01 (37J05 47N50 70H05)},
  MRNUMBER = {2433906},
MRREVIEWER = {H. Hogreve},
       DOI = {10.1090/gsm/095},
       URL = {https://doi-org.bris.idm.oclc.org/10.1090/gsm/095},
}

\bib{Tar}{article}{
author={Tarricone, Sofia},
title={A fully noncommutative Painlev\'e II hierarchy: Lax pair and solutions related to Fredholm determinants},
year={2020},
eprint={https://arix.org.abs/2007.05707},
archivePrefix={arXiv},
primaryClass={math.PH}
}


\bib{TW0}{article}{
AUTHOR = {Tracy, Craig A.}
author={Widom, Harold},
     TITLE = {Level-spacing distributions and the {A}iry kernel},
   JOURNAL = {Comm. Math. Phys.},
  FJOURNAL = {Communications in Mathematical Physics},
    VOLUME = {159},
      YEAR = {1994},
    NUMBER = {1},
     PAGES = {151--174},
      ISSN = {0010-3616},
   MRCLASS = {82B05 (33C90 47A75 47G10 47N55 82B10)},
  MRNUMBER = {1257246},
MRREVIEWER = {Estelle L. Basor},
       URL = {http://projecteuclid.org/euclid.cmp/1104254495},
}

\bib{TW}{article}{
AUTHOR = {Tracy, Craig A.}
author={Widom, Harold},
     TITLE = {Asymptotics in {ASEP} with step initial condition},
   JOURNAL = {Comm. Math. Phys.},
  FJOURNAL = {Communications in Mathematical Physics},
    VOLUME = {290},
      YEAR = {2009},
    NUMBER = {1},
     PAGES = {129--154},
      ISSN = {0010-3616},
   MRCLASS = {60K35 (47N30 82C22)},
  MRNUMBER = {2520510},
MRREVIEWER = {Timo Sepp\"{a}l\"{a}inen},
       DOI = {10.1007/s00220-009-0761-0},
       URL = {https://doi-org.bris.idm.oclc.org/10.1007/s00220-009-0761-0},
}

\bib{WE}{article}{
AUTHOR = {Warren, Oliver H.},
author={Elgin, John N.},
     TITLE = {The vector nonlinear {S}chr\"{o}dinger hierarchy},
   JOURNAL = {Phys. D},
  FJOURNAL = {Physica D. Nonlinear Phenomena},
    VOLUME = {228},
      YEAR = {2007},
    NUMBER = {2},
     PAGES = {166--171},
      ISSN = {0167-2789},
   MRCLASS = {37K10 (14H70 35Q55 37K20)},
  MRNUMBER = {2331592},
MRREVIEWER = {Beatrice Pelloni},
       DOI = {10.1016/j.physd.2007.03.006},
       URL = {https://doi.org/10.1016/j.physd.2007.03.006},
}


\end{biblist}
\end{bibsection}
\end{document}